\documentclass{amsart}
\usepackage[backend=bibtex,maxnames=4,firstinits=true,url=true,doi=false,isbn=false,citestyle=numeric]{biblatex}
\usepackage{enumerate}
\usepackage{graphicx}
\graphicspath{{figures/}}
\usepackage{float}
\usepackage[all]{xy}
\usepackage{aliascnt}
\usepackage{hyperref}
\usepackage{microtype}
\usepackage{adjustbox}
\usepackage[margin=4cm]{geometry}
\usepackage[all=normal,paragraphs=tight]{savetrees}

\makeatletter
\newcommand\Autoref[1]{\@first@ref#1,@}
\def\@throw@dot#1.#2@{#1}
\def\@set@refname#1{%
    \edef\@tmp{\getrefbykeydefault{#1}{anchor}{}}%
    \def\@refname{\@nameuse{\expandafter\@throw@dot\@tmp.@autorefname}s}%
}
\def\@first@ref#1,#2{%
  \ifx#2@\autoref{#1}\let\@secondref\@gobble% 
  \else%
    \@set@refname{#1}%
    \@refname~\ref{#1}%
    \let\@secondref\@second@ref%
  \fi%
  \@secondref#2%
}
\def\@second@ref#1,#2{%
  \ifx#2@ and~\ref{#1}\let\@nextref\@gobble%
  \else, \ref{#1}%
    \let\@nextref\@next@ref%
  \fi%
  \@nextref#2%
}
\def\@next@ref#1,#2{%
   \ifx#2@, and~\ref{#1}\let\@nextref\@gobble%
   \else, \ref{#1}%
   \fi%
   \@nextref#2%
}
\makeatother

\def\equationautorefname~#1\null{%
  Equation~(#1)\null
}

\numberwithin{equation}{section}

\theoremstyle{plain}

\newtheorem{theorem}{Theorem}[section]

\newaliascnt{lemma}{theorem}
\newtheorem{lemma}[lemma]{Lemma}
\aliascntresetthe{lemma}

\newaliascnt{corollary}{theorem}

\aliascntresetthe{corollary}

\theoremstyle{definition}

\newaliascnt{assumption}{theorem}
\newtheorem{assumption}[assumption]{Assumption}
\aliascntresetthe{assumption}

\newaliascnt{definition}{theorem}
\newtheorem{definition}[definition]{Definition}
\aliascntresetthe{definition}

\newaliascnt{remark}{theorem}

\aliascntresetthe{remark} 

\newtheorem*{remark*}{Remark}

\newaliascnt{example}{theorem}

\aliascntresetthe{example}

\newaliascnt{algorithm}{theorem}
\newtheorem{algorithm}[algorithm]{Algorithm}
\aliascntresetthe{algorithm}

\numberwithin{figure}{section}
\numberwithin{table}{section}

\begin{document}

\title{Consistent recalibration of yield curve models}
\date{September 2016}

\author{Philipp Harms}
\address{Institute of Mathematics, Albert-Ludwigs University Freiburg}
\email{philipp.harms@stochastik.uni-freiburg.de}
\thanks{Supported in part by SNF grant 149879}
\thanks{We gratefully acknowledge support by ETH Foundation}
\author{David Stefanovits}
\address{Department of Mathematics\\ ETH Z\"urich}
\email{david.stefanovits@math.ethz.ch}
\author{Josef Teichmann}
\email{josef.teichmann@math.ethz.ch}
\author{Mario V.~W\"uthrich}
\email{mario.wuethrich@math.ethz.ch}

\begin{abstract}
The analytical tractability of affine (short rate) models, such as the Vasi\v cek and the Cox-Ingersoll-Ross models, has made them a popular choice for modelling the dynamics of interest rates. However, in order to account properly for the dynamics of real data, these models need to exhibit time-dependent or even stochastic parameters. This in turn breaks their tractability, and modelling and simulating becomes an arduous task. We introduce a new class of Heath-Jarrow-Morton (HJM) models that both fit the dynamics of real market data and remain tractable. We call these models consistent recalibration (CRC) models. These CRC models appear as limits of concatenations of forward rate increments, each belonging to a Hull-White extended affine factor model with possibly different parameters. That is, we construct HJM models from ``tangent'' affine models. We develop a theory for a continuous path version of such models and discuss their numerical implementations within the Vasi\v cek and Cox-Ingersoll-Ross frameworks.
\end{abstract}

\maketitle

\section{Introduction}\label{sec:introduction}

\subsection{Principles of yield curve modelling}

Modelling the stochastic evolution of yield curves is an important task in risk management, forecasting, decision making, pricing and hedging. We emphasise here three principles of yield curve modelling (or any other traded instrument in finance): we certainly require that all models for traded assets' prices are free of arbitrage; therefore we do not state this as a principal requirement.

\begin{itemize}
\item {\bf Robust calibration:}  the model is selected simultaneously from time series and prevailing market prices, as explained in \cite{cuchiero2013fourier}. Model parameters which are invariant under equivalent measure changes should be \emph{estimated} by a statistical procedure from time series data. The remaining parameters are \emph{calibrated} by solving an inverse problem with respect to the prevailing market prices. All model parameters should be constant during the life time of the model; only state variables may change. 
\item {\bf Consistency:} an interest rate model is called \emph{consistent} if the stochastic process of yield curves does not leave a pre-specified set $\mathcal I$ of possible market observables (in \cite{filipovic2001consistency} the set $\mathcal I$ is assumed to be a finite dimensional sub-manifold of curves corresponding to a curve fitting method). Here, we add the following requirement: the yield curve process should be able to reach any neighbourhood of any  yield curve in $\mathcal I$ with positive probability because any newly arriving market configuration is a possible model state. Consequently, the model can be recalibrated to a new market configuration without losing consistency to the model choice with old parameters; we say that the model satisfies the \emph{consistent recalibration} property.
\item {\bf Analytic tractability:} relevant quantities of a model can be calculated quickly and accurately. In particular, one should be able to simulate state variable increments efficiently. This can be a delicate problem in the presence of boundary conditions.
\end{itemize}

We briefly comment on some of these principles.

\begin{itemize}
\item By a \emph{model} for the term structure of interest rates, we understand a fully specified stochastic process taking values in the pre-specified set of yield curves $\mathcal I$. We shall always consider a parametrised class of models consisting of one fully specified model for each initial state in $\mathcal I$ and each parameter value.

\item In practice, interest rate models are \emph{recalibrated} on a regular basis (e.g.\ daily) to market data. Suppose that the consistent recalibration property does not hold for today's model. Then tomorrow's market yield curve might lie outside of the set of possible realisations of today's model. If this happens, then tomorrow's recalibration necessarily implicates a rejection of today's model. On the other hand, no inconsistencies occur if recalibration is an update of the state variables of the model and does not involve a change of model parameters.

\item 
The \emph{robust calibration principle} separates the easier task of estimating volatilities from the more difficult task of estimating drifts. Moreover, it tells us exactly which parameters may be estimated from time series (namely those which are invariant under equivalent changes of measure). 
In \cite{harms2016consistent} we use our results to model and filter the market price of risk. 
\end{itemize}

\subsection{Consistent recalibration models} 

The goal of this work is to present a new model class for yield curve evolutions satisfying all three principal requirements with respect to sets $\mathcal I$ which are sufficiently large to be of practical use (think of open subsets of a Hilbert space of curves). Mathematically speaking, we look for yield curve models with full or large support, which are in addition analytically tractable. 

Often the full support property does not accord with analytic tractability beyond elliptic models, which are too restrictive in infinite dimension. We illustrate this with an example. Consider the Hull-White extended Cox-Ingersoll Ross (CIR) model. In this model the short rate is given by the SDE
\begin{equation*}
dr(t) = \big(\theta(t)+\beta r(t)\big) dt + \sqrt{\alpha r(t)} dW(t), 
\end{equation*}
where $\theta(t) \geq 0$ determines the time-dependent level of mean reversion, $\beta<0$ the speed of mean reversion, and $\alpha>0$ the level of volatility.

The model can be calibrated to any initial yield curve from a large subset $\mathcal I$ of curves by choosing an appropriate Hull-White extension $\theta$. For any fixed initial yield curve, the distribution of yield curves at some future $t>0$ is concentrated on a one-dimensional affine subspace of curves. Therefore, market observations are generally not in the support of the model, and the consistent recalibration property does not hold with respect to $\mathcal I$. Moreover, the low dimensionality of the model is apparent at the level of realised covariations of yields for different times to maturity: the matrix of covariations has rank one, which is in stark contrast to observations from the market (see \autoref{fig:num_market_covariation}). Finally, calibrated model parameters vary significantly over time as shown in \Autoref{fig:num_sigma_cir,fig:num_beta_cir}, which contradicts the requirement of robust calibration. 

As a remedy, one could make some model parameters stochastic and include them as state variables. For example, one could make $\alpha=\alpha_y$ and $\beta=\beta_y$ depend on a parameter $y$ and write dynamics of the form 
\begin{align*}
dr(t) &= \big(\theta(t)+\beta_{Y(t)} r(t)\big) dt + \sqrt{\alpha_{Y(t)} r(t)} dW(t), 
\\
dY(t) &= \mu\big(Y(t)\big) dt + \sigma\big(Y(t)\big) d\widetilde W(t).
\end{align*}
Unfortunately, this usually breaks the analytic tractability of the model in the sense that zero-coupon bond prices cannot be calculated anymore analytically.

The key idea of \emph{consistent recalibration (CRC) models} is to lift the short rate model to a HJM model and to introduce stochastic parameters on that level. Let $h(t)$ denote the forward rate curve at time $t$ in Musiela parametrisation (i.e., as a function of time to maturity). Then CRC models are defined by the joint dynamics 
\begin{align*}
dh(t)&=\Big(\mathcal Ah(t)+\mu^{\mathrm{HJM}}_{Y(t)}\big(r(t)\big)\Big)dt+\sigma^{\mathrm{HJM}}_{Y(t)}\big(r(t)\big)dW(t),
\\
dY(t) &= \mu\big(Y(t)\big) dt + \sigma\big(Y(t)\big) d\widetilde W(t),
\end{align*}
where $r(t)=h(t)(0)$, $\mathcal A$ is the generator of the shift semigroup, and $\mu^{\mathrm{HJM}}_y,\sigma^{\mathrm{HJM}}_y$ are the HJM drift and volatility of the short rate model with constant parameter $y$. This model, and more generally the class of CRC models, has the following properties:
\begin{itemize}
\item It is a full-fledged HJM model providing the benefits of \emph{robust calibration}. Indeed, all model parameters can be estimated from realised covariations of yields rather than calibrated by solving high-dimensional inverse problems. Thus, the parameters are estimated from yield curve dynamics instead of calibrated to static yield curves, while an exact match to the current yield curve is guaranteed by the Hull-White extension $\theta$.  

\item In the language of term structure equations, see \cite{richter2014,carmona2012,kallsen2013}, CRC models are \emph{tangent} to affine factor models. This means that they are (limits of) concatenations of affine factor models. The  construction is illustrated in \autoref{fig:foliations}. The concatenated models are allowed to have distinct static parameters. Making the parameters stochastic (in an independent or dependent way) will lead generically towards \emph{consistent recalibration}, since the conditions for finite-dimensional realisations are not fulfilled anymore. The consistent recalibration property is also reflected in the much higher ranks of the covariation matrices of yields with different times to maturity. Indeed, our empirical analysis shows that they are closer to those observed in the market than in the corresponding models without CRC extension (see \autoref{fig:num_market_covariation}). 

\item Despite all this flexibility, the model remains \emph{analytically tractable}: zero-coupon bond prices are state variables, and state variable increments can be simulated efficiently because they look infinitesimally like Hull-White extended affine processes. This means that the Fourier transform of the infinitesimal increments is known and that all sampling techniques of affine processes apply. In particular, efficient high-order positivity-preserving simulation schemes for the CIR process such as \cite{alfonsi2010high} can be used. No similar schemes are available for general HJM equations with non-Lipschitz vector fields. In our numerical implementation, we achieve first-order convergence of the splitting scheme. We also give a theoretical proof of first-order convergence in the Vasi\v cek case.
\end{itemize}

These properties are important in risk management and in the current regulatory framework \cite{BIS}, where one needs tractable and  realistic (non-Gaussian) models of long-term returns on bond portfolios. Moreover, the CRC approach allows one to easily implement stress tests for risk management purposes by selecting a suitable model for the parameter process. First evidence of improved fits is provided in \cite{harms2016consistent}.

The principle behind CRC models applies also to the modelling of more general term structure dynamics. For example, it could be applied to multi-curve interest rate models and to models of the term structure of option prices \cite{richter2014}.

\subsection{Organisation of the paper}

In \autoref{sec:lit} we discuss relations to other yield curve modelling approaches. In \autoref{sec:hw}, we introduce Hull-White extended affine short rate models, which are the building blocks of CRC models introduced in \autoref{sec:crc}. In \Autoref{sec:va,sec:cir}, the one-factor Vasi\v cek and CIR case is developed in full detail. In \autoref{sec:num}, our numerical implementation and some empirical results are presented.

\section{Relations to other yield curve models}\label{sec:lit}

Several in part overlapping approaches to yield curve modelling have been developed. The models can be roughly categorised as factor, HJM, principal component analysis (PCA), and filtered historical simulation models. We briefly analyse these models with respect to our requirements and compare them to the new class of CRC models. 

\subsection{Factor models} Factor models are based on a factor process, which usually describes certain market factors, from which -- by means of basic principles -- the entire yield curve can be derived (see \cite{filipovic2009term} for an overview). Let $ X=(X(t))_{t\geq 0} $ be a factor Markov process acting on a finite dimensional state space and depending on a parameter vector $y$, and let $ B := g(X) $ be the bank account process, for some positive functional $g$. Then one obtains -- with respect to the pricing measure $\mathbb P$ -- the relation
\[
P(t,T)=\mathbb E \left[ \frac{B(t)}{B(T)} \middle| \mathcal{F}(t) \right] = G(t,T,X(t)), 
\]
for some function $G$ also depending on the parameter vector $y$. 

Market data arrive in the form of daily yield curves. By means of calibration the initial state $x_0$ and a parameter vector $y_0 $ are chosen to explain today's market data. By choosing the parameter vector rich enough, one receives good fits to today's market data. Apparently a recalibration at time $t=1$ can (and will) lead to another state $x_1$ and another parameter vector $y_1$. As states may vary stochastically, the change of $x_0$ to $x_1$ is in principle not a problem, but the change of parameter vector is. This means that one has to decide at time $t=1$ whether to continue with the model specified by $y_0$ or whether to switch to the model specified by $y_1$. This problem can be alleviated to some extent by using a combination of filtering and calibration techniques to stabilise the choice of $y$, as described in, e.g., \cite{gombani2005filtered,barone1998market}. Nevertheless, robust calibration remains an unresolved issue.

The consistent recalibration property does not hold unless the set $\mathcal I$ is very small. (If $\mathcal I$ is a sub-manifold, its dimension cannot be larger than that of the factor process.) However, on the positive side, factor models are often analytically tractable, for instance, within the affine class (see e.g.\ \cite{filipovic2009term,duffie2003affine}).

\subsection{HJM-models} \label{sec:intro:hjm}
Markovian HJM models are an extreme version of factor models: the yield curve itself is taken as state variable (possibly together with some hidden state variables). Calibration to daily arriving yield curves is now a matter of statistical estimation from the time series of market data. An appropriate parametrisation of instantaneous co-variance, jump structure, and drifts will lead to a statistical inference problem, an infinite dimensional one though. Hence, the paradigm of robust calibration, including the requirement of calibration through estimation, is fulfilled in the optimal sense. If the process acting on yield curves is ``irreducible'', i.e., every neighbourhood of a state can be reached with positive probability, then even consistent recalibration is possible (see \cite{baudoin2005hypoellipticity}). However, one usually encounters a severe lack of analytic tractability within this model class. Euler and higher order schemes (often) require strong assumptions on the vector fields (c.f.~\cite{doersek2013efficient}). Usually no exact or high-order simulations of infinitesimal increments are at hand in contrast to CRC models, where this is often the case. 

\subsection{PCA- or local PCA models} Principal component analysis (PCA) or local PCA considers yield curves as outcomes of a statistical model, which is estimated by standard PCA techniques (see e.g. \cite{litterman1991common,bouchaud1999phenomenology,cont2005modeling}). When the statistical model is too simplistic, often arbitrage enters the field, which is an undesirable feature. A more refined version is actually equivalent to a HJM model with constant vector fields (as e.g. in \cite{henseler2013tractable}). Here preserving floored interest rates, which is desired in some situations, is not possible. PCA inspired models, correctly implemented, allow for robust calibration and consistent recalibration, but are usually not very tractable from an analytic point of view.

\subsection{Filtered (historical) simulation} Historical simulation is a standard industry technique to simulate distributions of yield curves by considering the relative returns as independent samples of an unknown distribution, see \cite{gombani2005filtered,barone1998market}. Certainly this assumption can lead to difficulties with the absence of arbitrage, but this can be solved as in \cite{ortega2009new,teichmann2012consistent}. The most important problem is the state-independence of the distribution. Again also (filtered) historical simulation can be embedded into the realm of HJM models. These models then allow for robust calibration and consistent recalibration, but are usually not very tractable from an analytic point of view.

\section{Hull-White extended affine short rate models}\label{sec:hw}

\subsection{Overview}

We set the stage for CRC models by describing Hull-White extended affine short rate models, focusing first on the correspondence between Hull-White extensions and initial forward rate curves. The one-dimensional short rate model of the introduction is replaced by more general multi-dimensional factor models for the short rate. The parameter $y$, which becomes stochastic in the CRC setting, is kept constant and fixed for the moment. 

\subsection{Setup and notation}\label{sec:hw:setup}

$\left(\Omega, \mathcal F ,(\mathcal F(t))_{t\geq0}, \mathbb P\right)$ is a filtered probability space. The filtration satisfies the usual conditions. The measure $\mathbb P$ plays the role of a risk-neutral measure.
All processes are defined on $\Omega$, adapted to $(\mathcal F(t))_{t\geq0}$, and c\`adl\`ag. 
$W=(W(t))_{t\geq 0}$ is a $d$-dimensional $(\mathcal F(t))_{t\geq 0}$-Brownian motion. 

The short rate process $r=(r(t))_{t\geq 0}$ is determined by a factor process $X=(X(t))_{t\geq0}$ with values in a state space $\mathbb X$. The evolution of the factor process depends on a parameter process $Y=(Y(t))_{t\geq 0}$ with values in a space $\mathbb Y$. In all of \autoref{sec:hw}, the parameter process $Y(t)\equiv y$ is assumed to be constant and fixed, whereas it is allowed to vary in \autoref{sec:crc} below. 

The spaces $\mathbb X$ and $\mathbb Y$ are both subsets of some finite dimensional vector space. $\mathbb X$ is, up to permutation of coordinates, of the canonical form $\mathbb R^{d_1}_+ \times \mathbb R^{d_2}$ with $d_1+d_2=d\geq 1$. The canonical basis vectors in $\mathbb R^d$ are denoted by $e_1,\ldots,e_d$, and $\langle\cdot,\cdot\rangle$ denotes the Euclidean scalar product. Of course we could consider more general affine processes here, which take values, e.g., in products of cones of positive-semidefinite matrices and real lines like Wishart-Heston models.

For each $(x,y) \in \mathbb X\times\mathbb Y$, there is a symmetric positive semidefinite matrix $A_y(x) \in \mathbb R^{d\times d}$ and a vector $B_y(x) \in \mathbb R^d$, determining the volatility and the drift of $X$. The expressions $A_y(x)$ and $B_y(x)$ are affine in $x$, i.e., 
\begin{equation*}
A_y(x)=a_y+\sum_{i=1}^d \alpha^i_y x^i, 
\qquad
B_y(x)=b_y+\sum_{i=1}^d \beta^i_y x^i,
\qquad
\text{for all } (x,y)\in\mathbb X\times\mathbb Y,
\end{equation*}
for symmetric positive semidefinite matrices $a_y,\alpha^1_y,\dots,\alpha^d_y \in \mathbb R^{d\times d}$ and $b_y,\beta^1_y,\dots,\beta^d_y \in \mathbb R^d$. We denote by $\sqrt{A_y(x)}$ the symmetric positive semidefinite square root of $A_y(x)$. Note: other choices of square roots are possible, and \autoref{ass:hw:x} below does not depend on the choice of square root by \cite[Chapter~V.19--20]{rogers2000diffusions}. Moreover, a function $\theta\in C(\mathbb R_+)$ is given, which is used to make the drift of $X$ time-inhomogeneous.

\subsection{Factor process and short rate}\label{sec:hw:shortrate}

The \emph{factor process} $X$ is a continuous, $\mathbb X$-valued solution of the SDE
\begin{align}\label{equ:hw:x}
dX(t)=\sqrt{A_y\big(X(t)\big)}dW(t)+\Big(\theta(t)e_1+B_y\big(X(t)\big)\Big)dt
\end{align} 
with initial condition $X(0)=x \in \mathbb X$. The \emph{short rate} is given by
\begin{equation}\label{equ:hw:r}
r(t)=\ell+\langle \lambda,X(t)\rangle, \qquad \text{for all } t\geq 0,
\end{equation}
for some fixed $\ell \in \mathbb R$ and $\lambda \in \mathbb R^d$ satisfying $\langle\lambda,e_1\rangle\neq 0$. 

\begin{assumption}\label{ass:hw:x}
It is assumed that SDE \eqref{equ:hw:x} has a unique continuous, $\mathbb X$-valued solution $X$, for each initial condition $X(s)=x$, where $(s,x)\in\mathbb R_+\times\mathbb X$. In this case, the parameters $(y,\theta)$ are called \emph{admissible}. Moreover, it is assumed that $X$ satisfies the \emph{moment condition} 
\begin{equation}\label{equ:hw:moment}
\mathbb E\left[e^{-\int_0^t\left(\ell+\langle \lambda,X(s)\rangle\right)ds}\right]<\infty, 
\qquad
\text{for all } t\geq 0.
\end{equation}
\end{assumption}

\subsection{Exponential moments and Riccati equations}\label{sec:hw:riccati}

The process $X$, or rather the family of processes obtained by varying the initial conditions in SDE \eqref{equ:hw:x}, is time-inhomogeneous affine. All coefficients in SDE \eqref{equ:hw:x} are independent of time, except for the drift $\theta$; we call $X$ \emph{Hull-White extended affine} and $\theta$ its Hull-White extension. This we are going to highlight in detail below. Our main reference for time-inhomogeneous affine processes is \cite{filipovic2005time}. 

Functions $(\Phi_y,\Psi_y) \in C^\infty(\mathbb R_+)\times C^\infty(\mathbb R_+;\mathbb R^d)$ are called solutions of the \emph{Riccati equations} (with parameter $y\in\mathbb Y$) if
\begin{subequations}\label{equ:hw:riccati}
\begin{align}
\Phi_y'&=F_y\circ \Psi_y, & \Phi_y(0)&=0,
\\\label{equ:hw:riccati_psi}
\Psi_y'&=R_y\circ \Psi_y-\lambda, & \Psi_y(0)&=0
\end{align}
\end{subequations}
holds, where $(F_y,R_y) \in C(\mathbb R^d)\times C(\mathbb R^d;\mathbb R^d)$ are given by
\begin{align*}
F_y(u)&=\frac12 \langle u,a_y u\rangle + \langle u,b_y\rangle, 
&
R_y^i(u)&= \frac12 \langle u,\alpha_y^i u\rangle+\langle u,\beta_y^i\rangle,
\end{align*}
for all $u\in\mathbb R^d$ and $i\in\{1,\dots,d\}$.

\begin{lemma}\label{lem:hw:moment}
Let $X$ be a solution of SDE \eqref{equ:hw:x} for some admissible parameters $(y,\theta)$. Then $X$ satisfies moment condition \eqref{equ:hw:moment} if and only if there exists a solution $(\Phi_y,\Psi_y)$ of the Riccati equations \eqref{equ:hw:riccati} for all times $t\geq 0$. Moreover, if there exists a solution, even a local one, of the Riccati equations, it is unique.
\end{lemma}

Note that the Riccati equations \eqref{equ:hw:riccati} only depend on $y$, and not on the choice of the Hull-White extension $\theta$.

\begin{proof}
Let $Z$ be the $\mathbb X^2$-valued process $Z(t)=(Z_1(t),Z_2(t))=(X(t),\int_0^t X(s)ds)$. Then $Z$ is an It\= o diffusion whose drift and volatility at time $t\geq 0$ are given by 
\begin{align*}
\big(\theta(t)e_1+B_y(Z_1(t)),Z_1(t)\big) \in \mathbb R^{2d}, 
&&
\begin{pmatrix}A_y(Z_1(t)) &0\\0&0\end{pmatrix}\in\mathbb R^{2d\times 2d},
\end{align*}
respectively. Clearly, these expressions are affine in $Z(t)$. Moreover, the time-inhomo\-geneity $\theta(t)e_1$ is (by definition) continuous in $t$. Therefore, the process $Z$, or rather the family of processes obtained by varying the initial condition of $Z$, is strongly regular affine, see \cite[Theorem 2.14]{filipovic2005time}. For each $(t,u_1,u_2)\in \mathbb R_+\times\mathbb R^d\times\mathbb R^d$, the functional characteristics of $Z$ are given by
\begin{align*}
F(t,u_1,u_2)&=\theta(t) \langle u_1,e_1 \rangle+ F_y(u_1) \in \mathbb R, 
\\
R(t,u_1,u_2)&=(R_y(u_1)+u_2,0)\in \mathbb R^d\times\mathbb R^d. 
\end{align*}
Moment condition \eqref{equ:hw:moment}, expressed in terms of $Z$, reads as follows: 
\begin{equation*}
\mathbb E\left[e^{-\langle\lambda,Z_2(T)\rangle}\right]<\infty,
\qquad
\text{for all } T\geq 0.
\end{equation*}
By \cite{harms2015exponential}, the moment condition is equivalent to the existence of a solution $(\phi,\psi_1,\psi_2)$ of the following Riccati system associated to $Z$:
\begin{align*}
-\partial_t \phi(t,T) &=\theta(t)\langle\psi_1(t,T),e_1 \rangle+ F_y(\psi_1(t,T)),
&
\phi(T,T)&=0, 
\\
-\partial_t \psi_1(t,T) &=R_y(\psi_1(t,T))+\psi_2(t,T),
&
\psi_1(T,T) &= 0,
\\
-\partial_t \psi_2(t,T) &= 0, 
&
\psi_2(T,T)&=-\lambda.
\end{align*}
Equivalently, the relations $\psi_2(t,T)=-\lambda$ and
\begin{align*}
\phi(t,T)&=\int_t^T \theta(s)\langle\Psi_y(T-s),e_1\rangle ds 
+ \Phi_y(T-t),
&
\psi_1(t,T)=\Psi_y(T-t)
\end{align*}
hold identically, where $(\Phi_y,\Psi_y)$ is a solution of the Riccati equations \eqref{equ:hw:riccati}. Uniqueness holds for these equations because the vector fields are locally Lipschitz.
\end{proof}

\subsection{Bond prices and forward rates}
By no-arbitrage arguments zero-coupon bond prices in the short rate model of \autoref{sec:hw:shortrate} are given by
\[
P(t,T)=\mathbb E\left[e^{-\int_{t}^{T}r(s)ds}\middle| \mathcal F(t)\right]=\mathbb E\left[e^{-\int_t^T\left(\ell+\langle \lambda,X(s)\rangle\right)ds}\middle| \mathcal F(t)\right],\qquad T\geq t\geq 0.
\]
For essentials on short rate models we refer to \cite[Chapter 5]{filipovic2009term}. We define the (instantaneous) forward rates by
\[
h(t,\tau)=h(t)(\tau)=-\partial_\tau\log \big(P(t,t+\tau)\big),\qquad t,\tau\geq0.
\]
The parametrisation of the forward rate as a function of $t$ and $\tau$ is called Musiela parametrisation. It is particularly useful in this paper since $\left(h(t)\right)_{t\geq0}$ will be interpreted as a stochastic process taking values in a suitable space of functions on $\mathbb R_{+}$.

By the affine nature of the factor process $X$, zero-coupon bond prices and forward rates can be obtained by solving the Riccati system of ODE's \eqref{equ:hw:riccati}.
\begin{theorem}[Zero-coupon bond price and forward rate]\label{thm:hw:price}
Let $X$ satisfy moment condition \eqref{equ:hw:moment} and let $(\Phi_y,\Psi_y)$ be the unique solution of the Riccati equations with parameter $y$, given by \autoref{lem:hw:moment}. Then the \emph{bond prices} in the short rate model \eqref{equ:hw:x}--\eqref{equ:hw:r} satisfy
\begin{equation}\label{equ:hw:p}\begin{aligned}
\log(P(t,T))=-\ell(T-t)
+\int_t^T \theta(s)\langle\Psi_y(T-s),e_1\rangle ds
+ \Phi_y(T-t)
+\langle\Psi_y(T-t), X(t)\rangle,
\end{aligned}\end{equation}
for all $T\geq t\geq 0$, and the \emph{forward rates} are given by
\begin{equation}\label{equ:hw:h}\begin{aligned}
h(t,\tau)=\ell-\int_0^\tau \theta(t+s)\langle \Psi_y'(\tau-s),e_1\rangle ds
-\Phi_y'(\tau)-\langle \Psi_y'(\tau), X(t)\rangle,
\end{aligned}\end{equation}
for all $t,\tau \geq 0$.
\end{theorem}

\begin{proof}
We borrow from the proof of \autoref{lem:hw:moment}, where moment condition \eqref{equ:hw:moment} was shown to be equivalent to the existence of solutions $(\phi,\psi_1,\psi_2)$ of the Riccati system associated to the process $Z=(X,\int X)$. Moreover, $(\phi,\psi_1,\psi_2)$ are closely related to the solutions $(\Phi_y,\Psi_y)$ of Riccati system \eqref{equ:hw:riccati}. By the main theorem in \cite{harms2015exponential} and its conditional version, the affine transform formula
\begin{equation*}
\mathbb E\left[e^{-\langle\lambda,Z_2(T)\rangle}\middle|\mathcal F(t)\right]=e^{\phi(t,T)+\langle \psi_1(t,T),Z_1(t)\rangle+\langle \psi_2(t,T),Z_2(t)\rangle},
\qquad
\text{for all } T\geq t\geq 0,
\end{equation*}
holds. A direct calculation shows this formula to be equivalent to formula \eqref{equ:hw:p} for bond prices. Formula \eqref{equ:hw:h} for forward rates is obtained by taking the logarithm and differentiating with respect to $\tau$.
\end{proof}

\subsection{Heath-Jarrow-Morton equation}\label{sec:hw:hjm}

The evolution of forward rate curves is described by the \emph{HJM equation}. For each $(x,y)\in\mathbb X\times\mathbb Y$, let $\mu^{\mathrm{HJM}}_y(x) $ and $\sigma^{\mathrm{HJM}}_y(x) $ be given by
\begin{align*}
\mu^{\mathrm{HJM}}_y(x)&=\langle\Psi_y,A_y(x)\Psi_y'\rangle\in C^\infty(\mathbb R_+),
&
\sigma^{\mathrm{HJM}}_y(x)&=-\sqrt{A_y(x)} \Psi'_y\in C^\infty(\mathbb R_+;\mathbb R^d).
\end{align*}
Note that the familiar HJM drift condition holds:
\begin{equation}\label{equ:hw:drift}
\mu^{\mathrm{HJM}}_y(x)(\tau)=
\left\langle\sigma^{\mathrm{HJM}}_y(x)(\tau),
\int_0^\tau \sigma^{\mathrm{HJM}}_y(x)(s)ds\right\rangle,
\qquad
\text{for all } \tau\geq 0.
\end{equation}
Let $\mathbb H$ be a Hilbert space, destined to contain the forward rate curves of the model. By abuse of notation the symbol $h$ is used interchangeably to denote an element of $\mathbb H$ and the forward rate process.

\begin{assumption}\label{ass:hw:h}
 $\mathbb H$ is a Hilbert space with the following properties:
\begin{enumerate}[(i)]
\item \label{ass:hw:h:item1} $\mathbb H  \subset C(\mathbb R_+) $ and the evaluation map $\mathrm{eval}_{\tau}\colon h\mapsto h(\tau)$ is continuous on $\mathbb H$, for each $\tau\in\mathbb R_+$;
\item \label{ass:hw:h:item2} for each $(x,y,z)\in\mathbb X\times\mathbb Y\times\mathbb R^d$, $\mu^{\mathrm{HJM}}_y(x)$ and $\langle\sigma^{\mathrm{HJM}}_y(x),z\rangle$ are elements of $\mathbb H$;
\item \label{ass:hw:h:item3} the right shifts $(\mathcal S(t))_{t\geq 0}$ mapping $h$ to $h(t+\cdot)$ define a strongly continuous semigroup on $\mathbb H$ with infinitesimal generator $\mathcal A$.
\end{enumerate}
\end{assumption}

Hilbert spaces of forward rate curves which comply with the requirements of \autoref{ass:hw:h} are constructed in \cite[Sections 5, 7.4.1 and 7.4.2]{filipovic2001consistency} for the Vasi\v cek and CIR models. In the domain $\mathcal D(\mathcal A)\subset\mathbb H\cap C^1(\mathbb R_+)$ (c.f.\ \cite[Lemma 4.2.2]{filipovic2001consistency}) of the infinitesimal generator $\mathcal A$ we can characterise the process $(h,X)$ as follows.
\begin{theorem}[HJM equation]\label{thm:hw:hjm}
Let $(h,X)$ be given by \autoref{thm:hw:price} and assume that $h(t)\in\mathcal D(\mathcal A)$, for each $t\geq 0$. Then the process $(h,X)$ is a strong solution of the following SPDE on $\mathbb H\times \mathbb X$:
\begin{equation}
\label{equ:hw:hjm}
\begin{aligned}
dh(t)&=\Big(\mathcal Ah(t)+\mu^{\mathrm{HJM}}_y\big(X(t)\big)\Big)dt+\sigma^{\mathrm{HJM}}_y\big(X(t)\big)dW(t),
\\
dX(t)&=\sqrt{A_y\big(X(t)\big)}dW(t)+\Big(\theta(t)e_1+B_y\big(X(t)\big)\Big)dt.
\end{aligned}
\end{equation}
\end{theorem}

Outside of the domain of $\mathcal A$ the forward rate process can be characterised as a mild solution of \autoref{equ:hw:hjm}. For the concepts of mild, weak, and strong solutions of SPDEs we refer to \cite[Section 6.1]{daPrato2014se}. 

In the one-factor case, the factor process $(X(t))_{t\geq 0}$ is a simple functional of the forward rate. Then \autoref{equ:hw:hjm} can be rewritten as an evolution equation for the forward rate process alone (c.f. \autoref{equ:va:hjm}). This is also possible in the multi-factor case, but the corresponding functional is more complicated, which is why we choose to present the HJM equation in the form \eqref{equ:hw:hjm}. 

\begin{proof}
Differentiating formula \eqref{equ:hw:h} for forward rates with respect to $t$ and $\tau$ and using $\Psi_y'(0)=-\lambda$, one obtains for each $\tau\geq0$
\begin{align*}
dh(t,\tau)
&=\bigg(-\int_t^{t+\tau} \theta(s)\langle \Psi_y''(t+\tau-s),e_1\rangle ds
+\theta(t+\tau)\langle \lambda,e_1\rangle 
\\&\qquad
+\theta(t)\langle \Psi_y'(\tau),e_1\rangle \bigg)dt
-\langle \Psi_y'(\tau), dX(t)\rangle,
\\
\mathcal Ah(t,\tau)&=
-\int_t^{t+\tau} \theta(s)\langle \Psi_y''(t+\tau-s),e_1\rangle ds
+\theta(t+\tau)\langle \lambda,e_1\rangle 
\\&\qquad
-\Phi_y''(\tau)-\langle \Psi_y''(\tau), X(t)\rangle.
\end{align*}
Subtracting the equations and cancelling out the integral as well as the term next to it yields
\begin{align*}
dh(t,\tau)
&=\left(\mathcal Ah(t,\tau)+\Phi_y''(\tau)+\langle \Psi_y''(\tau), X(t)\rangle
+\theta(t)\langle \Psi_y'(\tau),e_1\rangle \right)dt
-\langle \Psi_y'(\tau), dX(t)\rangle.
\end{align*}
When $dX(t)$ is replaced by the right-hand side of SDE \eqref{equ:hw:x}, the $\theta(t)$-term cancels out and one obtains for each $\tau\geq0$
\begin{align*}
dh(t,\tau)
&=\left(\mathcal Ah(t,\tau)+\Phi_y''(\tau)+\langle \Psi_y''(\tau), X(t)\rangle 
-\langle \Psi_y'(\tau), B_y(X(t))\rangle\right)dt
\\&\qquad
-\left\langle \Psi_y'(\tau), \sqrt{A_y(X(t))}dW(t)\right\rangle.
\end{align*}
The symmetric matrix $\sqrt{A_y(X(t))}$ can be moved to the other side of the scalar product, and one immediately recognises the volatility $\sigma_y^{\mathrm{HJM}}(X(t))$. A direct calculation shows that the drift is equal to $\mu_y^{\mathrm{HJM}}(X(t))$. Indeed, \begin{align*}
\Phi''_y+\left\langle \Psi''_y,x\right\rangle-\left\langle\Psi'_y,B_y(x)\right\rangle
&=F_y'\circ\Psi_y \cdot \Psi_y'
+\langle R_y'\circ\Psi_y \cdot\Psi_y',x\rangle
-\langle \Psi_y',B_y(x)\rangle
\\
&=\langle \Psi_y,A_y(x) \Psi_y'\rangle =\mu_y^{\mathrm{HJM}}(x),
\end{align*}
which follows from the relations
\begin{align*}
    F_y'(u)\cdot v&=\langle u, a_y v\rangle +\langle v,b_y\rangle, 
    &
    (R_y^i)'(u)\cdot v &=\langle u, \alpha^i_y v\rangle +\langle v,\beta^i_y\rangle.
    \qedhere
\end{align*}
\end{proof}

\subsection{Forward rates and Hull-White extensions}\label{sec:hw:calibration}

Relation \eqref{equ:hw:h} between the forward rate $h(t)$ and the Hull-White extension $\theta$ plays a key role in \emph{calibration} (and recalibration) of the model. It can be expressed concisely as 
\begin{equation*}
h(t)=\mathcal H_y\big(\mathcal S(t)\theta,X(t)\big), 
\qquad
\mathcal S(t)\theta = \mathcal C_y\big(h(t),X(t)\big),
\qquad
\text{for all } t\geq 0,
\end{equation*}
where $\mathcal S(t)$ is the right shift operator, see \autoref{ass:hw:h}(\ref{ass:hw:h:item3}), $\mathcal H_y$ calculates the initial forward rate curve from the Hull-White extension given parameter $y$, and $\mathcal C_y$ performs the inverse operation of calibrating a Hull-White extension to an initial forward rate curve. Formally, for each $(t,x,\theta)\in\mathbb R_+\times\mathbb X\times C(\mathbb R_+)$, these operators are given by
\begin{align*}
\mathcal S(t)\theta&=\theta(t+\cdot) \in C(\mathbb R_+),
\vphantom{\int}
\\
\mathcal H_y(\theta,x)&=\ell-\mathcal I_y(\theta)-\Phi_y'-\left\langle \Psi_y',x\right\rangle \in C^1(\mathbb R_+),
\\
\mathcal I_y(\theta) &= \int_0^\cdot \theta(s)\langle \Psi_y'(\cdot-s),e_1\rangle ds \in C^1(\mathbb R_+).
\end{align*}
Note that $\mathcal H_y$ involves the Volterra integral operator $\mathcal I_y$. 
The operator $\mathcal C_y$ (the letter $\mathcal C$ standing for calibration) is defined as the partial inverse of $\mathcal H_y$ given by the following theorem.

\begin{theorem}[Calibration to initial forward rate curves] \label{thm:hw:volterra}
Let $(h,x) \in C^1(\mathbb R_+)\times \mathbb X$ satisfy $h(0)=\ell+\langle\lambda,x\rangle$. Then the Volterra integral equation $h=\mathcal H_y(\theta,x)$ has a unique solution $\theta\in C(\mathbb{R}_+)$, which we denote by $\mathcal C_y(h,x)$.
\end{theorem}

The theorem is a direct consequence of the following lemma.

\begin{lemma}\label{lem:hw:volterra}
For each $y\in\mathbb Y$, the Volterra integral operator
\begin{equation*}
\mathcal I_y\colon C(\mathbb R_+)\to \left\{h\in C^1(\mathbb R_+) \colon h(0)=0\right\}
\end{equation*}
is bijective.
\end{lemma}

\begin{proof}
This follows from \cite[Theorem 2.1.8]{brunner2004collocation}, noting that the integral kernel  
\begin{equation}\label{equ:hw:kernel}
K_y(s,t)=\left\langle \Psi_y'(t-s),e_1\right\rangle, \qquad\text{for all } t\geq s\geq 0,
\end{equation}
satisfies $|K_y(t,t)|=|\langle\lambda,e_1\rangle|>0$ and both $K_y$ and $\partial_t K_y$ are continuous.
\end{proof}

Note that calibration of a Hull-White extension $\theta$ requires the inversion of the \emph{Volterra integral operator} $\mathcal I_y$. Here the assumption $\langle\lambda,e_1\rangle\neq 0$ is needed. 

\subsection{Numerical solution of the Volterra equation}\label{sec:hw:numerical_volterra}

In the absence of analytical formulas, the Volterra equation has to be solved numerically. We are aiming at a second order approximation to keep the global error of the simulation scheme of order one.
Thus, we approximate the Volterra integral operator $\mathcal I_y$ by the trapezoid rule, which yields an operator $\widehat{\mathcal I}_y$ given by
\begin{equation*}
\widehat{\mathcal I}_y(\theta)(\tau_n)
=
\delta\left(\frac{1}{2}\langle\Psi_y'(\tau_n),e_1\rangle\theta(0)
\vphantom{\sum_{i=1}^{n-1}}
+\sum_{i=1}^{n-1}\langle\Psi_y'(\tau_n-\tau_i),e_1\rangle\theta(\tau_i)
+\frac{1}{2}\langle\Psi_y'(0),e_1\rangle\theta(\tau_n)\right),
\end{equation*}
for each $n\in\mathbb N^+$, where $\tau_n=n\delta$ constitutes a uniform grid of step size $\delta>0$.
Approximate solutions $\widehat \theta$ can be constructed by solving for continuous piecewise linear (i.e.\ linear on each interval $[\tau_n,\tau_{n+1}]$) functions $\widehat\theta$ satisfying
\begin{equation}\label{equ:hw:numerical_volterra}
\widehat \theta(0)=\frac{g'(0)}{\langle\Psi_y'(0),e_1\rangle},
\qquad
\widehat{\mathcal I}_y(\widehat \theta)(\tau_n)=g(\tau_n),
\qquad
\text{for all } n\in \mathbb N^+.
\end{equation}
As $\widehat{\mathcal I}_y$ is a second order approximation of $\mathcal I_y$, it is not surprising that $\widehat \theta$ is a second order approximation of $\theta$.

\begin{lemma}\label{lem:hw:numerical_volterra}
Let $(x,y)\in\mathbb X\times\mathbb Y$ and $g\colon \mathbb R_+\to\mathbb R$ piecewise $C^4$ with continuous second derivatives. If $g(0)=0$, then there is a unique piecewise linear function $\widehat \theta \in C(\mathbb R_+)$ satisfying \eqref{equ:hw:numerical_volterra}.
Moreover, $\widehat \theta$ is a second order approximation of the exact solution $\theta$ of the Volterra equation $\mathcal I_y(\theta)=g$ in the sense that for each $T\in\mathbb R_+$,
\begin{align*}
\sup_{t\in[0,T]} \lvert \widehat\theta(t)-\theta(t)\rvert \leq C \delta^2,
\end{align*}
where $C$ is a constant depending only on $T$ and $g$. 
\end{lemma}

The lemma will be used to show the numerical invertibility of the calibration operator $\mathcal H_y(\cdot,x)$. In this context, the smoothness assumption on the right-hand side $g$ of the Volterra equation is satisfied if the yields are interpolated sufficiently smoothly. 

\begin{proof}
This follows from \cite[Theorem 3]{linz1969numerical} and \cite[Example 2.4.5 and Theorems 2.4.5, 2.4.8]{brunner2004collocation}, noting that the integral kernel \eqref{equ:hw:kernel} is $C^4$ and strictly bounded away from zero along the diagonal by our assumption $\langle\Psi_y'(0),e_1\rangle=\langle\lambda,e_1\rangle\neq0$.
\end{proof}

Note that solving for $\widehat \theta$ can be performed efficiently because \eqref{equ:hw:numerical_volterra} is a linear system for $(\widehat \theta(\tau_0),\dots,\widehat \theta(\tau_n))$ of lower triangular form. 

\subsection{Estimation of the affine coefficients}\label{sec:hw:estimation}

We first discuss how, in principle, the affine coefficients can be identified from covariations of yields and then present some practical considerations on the construction of estimators. 

Let $r(t,\tau)$ denote the yield of a zero-coupon bond held from $t$ to $t+\tau$, i.e.\
\begin{equation}\label{equ:hw:yields}
r(t,\tau) = -\frac{1}{\tau}\log P(t,t+\tau),
\qquad
\text{for all } t,\tau\in\mathbb R_+.
\end{equation}
Then by \autoref{equ:hw:p} the quadratic covariation of yields with maturity $\tau_i$ and $\tau_j$ satisfies
\begin{equation}\label{equ:hw:inf_covariation}
\frac{d}{dt}\left[r(\cdot,\tau_i),r(\cdot,\tau_j)\right](t)
=
\frac{1}{\tau_i\tau_j}\Psi_y(\tau_i)^\top \left(a_y+\sum_{k=1}^d\alpha_y^k X^k(t)\right)\Psi_y(\tau_j).
\end{equation}
Assume that the left-hand side of \eqref{equ:hw:inf_covariation} is known as a function of $\tau_i$ and $\tau_j$, and that the components of $\Psi_y$ are functionally independent. Then the  components of $\Psi_y$ and the matrix $A_y(X(t))$ can be identified. The function $\Psi_y$ usually determines the coefficients $\alpha_y$ and $\beta_y$ uniquely (see \autoref{equ:hw:riccati_psi}). Furthermore, taking account of the admissibility conditions on the matrices $a_y$ and $\alpha^k_y$ (see \cite[Definition~2.6]{duffie2003affine}) one can identify the $\mathbb R_+$-valued components of $X(t)$ and the matrix $a_y$. 

Note that \eqref{equ:hw:inf_covariation} is derived solely from the diffusion coefficient of the yield dynamics and therefore is invariant under Girsanov's change of measure. Thus, the coefficients $a_y,\alpha_y,\beta_y$ can be estimated from real world observations without specifying the market price of risk. Of course the market price of risk enters as a bias in the estimation, but the estimators do not depend on it. Moreover, under the model hypothesis the estimates do not depend on the choice of $\tau_i,\tau_j$, which provides a means to reject ill-suited models. 

The remaining $\mathbb R$-valued components of $X(t)$ and the coefficient $b_y$ do not appear in the quadratic covariations~\eqref{equ:hw:inf_covariation}. We now discuss how they can be estimated. First, note that for one-factor models $b_y$ is redundant and can be normalised to zero thanks to the Hull-White extension. In the multi-factor case only the first component of the vector $b_y$ is redundant. Second, note that the short end of the forward rate curve gives a scalar condition on $X(t)$, which allows one to fully identify $X(t)$ if $X(t)$ has only a single $\mathbb R$-valued component. In the general multi-factor case, however, some components of $b_y$ and $X(t)$ remain undetermined. They may be calibrated to the prevailing market yield curve by regression methods. Alternatively, they may be estimated by econometric methods. However, these require a market price of risk specification. We do not discuss this topic here and refer to \cite{harms2016consistent} for further details.

In practise, the quadratic covariation \eqref{equ:hw:inf_covariation} must be estimated from yields $\widehat r(t_n,\tau_i)$ given by the market for times $t_n=n\delta$ and times to maturity $\tau_i$. A naive estimator is the realised covariation, which is defined as
\[
[\widehat r(\cdot,\tau_i),\widehat r(\cdot,\tau_j)](t_n) 
=\sum_{k=1}^n 
\big(\widehat r(t_k,\tau_i)-\widehat r(t_{k-1},\tau_i)\big)
\big(\widehat r(t_k,\tau_j)-\widehat r(t_{k-1},\tau_j)\big).
\]
Here also other estimators such as, e.g., Fourier estimators, as introduced by Paul Malliavin and Maria-Elvira Mancino, could be used, see, e.g., \cite{cuchiero2013fourier} for some recent developments. 
Fixing a time window of length $M$ and a time $t_n$, one has
\begin{equation}\label{equ:hw:covariation}
\begin{aligned}
\frac{[\widehat r(\cdot,\tau_i),\widehat r(\cdot,\tau_j)](t_n)
-[\widehat r(\cdot,\tau_i),\widehat r(\cdot,\tau_j)](t_{n-M})}
{t_n-t_{n-M}}
\approx
\frac{1}{\tau_i\tau_j}\Psi_y(\tau_i)^\top a_y\Psi_y(\tau_j)
\\
+\frac{\delta}{\tau_i\tau_j(t_n-t_{n-M})}\sum_{k=1}^d\Psi_y(\tau_i)^\top \alpha^k_y\Psi_y(\tau_j)\sum_{m=n-M+1}^{n} X^{k}(t_m).
\end{aligned}
\end{equation}
Therefore, for any time $t_n$ and any selection of times to maturity $\tau_i,\tau_j$, estimators $\widehat a_y$, $\widehat\alpha_y$, $\widehat \beta_y$, $\widehat X^1(t_n)$, \dots, $\widehat X^{d_1}(t_n)$ can be constructed by solving for the best fit in \autoref{equ:hw:covariation}.  

\section{Consistent recalibration of affine short rate models}\label{sec:crc}

\subsection{Overview}

The constant parameter process $y$ of the previous section is now replaced by a stochastic process $Y=(Y(t))_{t\geq 0}$. The situation is particularly simple when $Y$ is piecewise constant. In this case, the Hull-White extension is recalibrated to the prevailing yield curve (i.e., the yield curve given by the model with old parameters) each time the parameter process changes. Later on, the concepts are generalised to arbitrary parameter processes $Y$, resulting in our definition of general CRC models. Geometrically, these models ``locally look like'' Hull-White extended affine short rate models with fixed parameter $y$. This is made precise in \autoref{sec:crc:geometry}. A semigroup point of view is taken in \autoref{sec:crc:semigroup}, leading to an interpretation of CRC models with piecewise constant $Y$ as splitting schemes with respect to a time grid for more general CRC models. 
 
\subsection{Setup and notation}\label{sec:crc:setup}

We recall the notion of admissible parameters from \autoref{ass:hw:x}. 
For all $(s,x)\in \mathbb R_+\times\mathbb X$ and all admissible parameters $(y,\theta)\in\mathbb Y\times C(\mathbb R_+)$, we let $X=X^{s,x}_{y,\theta}$ denote the unique solution on $[s,\infty)$ of the SDE \eqref{equ:hw:x} with $\theta(t)$ replaced by $\theta(t-s)$ and initial condition $X(s)=x$. We assume that $\mathbb H$ is a Hilbert space of forward rate curves satisfying \autoref{ass:hw:h} simultaneously for all $y\in\mathbb Y$. We fix a strictly increasing sequence of non-negative deterministic times $(t_n)_{n \in \mathbb N_0}$.

\subsection{CRC models with piecewise constant parameter process}\label{sec:crc:piecewise}

We assume that the parameter process $Y$ is piecewise constant, i.e.\ for each $t\in[t_n,t_{n+1})$ we have $Y(t)=Y(t_n)$.
\begin{definition}[CRC models with piecewise constant $Y$]\label{def:crc}
A stochastic process $(h,X,Y)$ with values in $\mathbb H\times \mathbb X\times \mathbb Y$ is called a \emph{CRC model} if there exists a stochastic process $\theta$ with values in $C(\mathbb R_+)$ such that the following conditions are satisfied, for each $n\in\mathbb N_0$:
\begin{enumerate}[(i)]
\item \label{def:crc:item1} The Hull-White extension $\theta$ on $[t_n,t_{n+1}]$ is determined by calibration to $h(t_n)$: 
\begin{equation*}
h(t_n)(0)=\ell+\langle\lambda,X(t_n)\rangle,\qquad\theta(t_n)=\mathcal C_{Y(t_n)}\left(h(t_n),X(t_n)\right),
\end{equation*}
and for $t\in[t_n,t_{n+1}]$
\begin{equation*}
\theta(t)=\mathcal S(t-t_n)\theta(t_n). 
\end{equation*}
\item \label{def:crc:item2} The evolution of $X$ on $[t_n,t_{n+1}]$ corresponds to the Hull-White extended affine model determined by the parameters $(Y(t_n), \theta(t_n))$:
\begin{equation*}
X(t)=X^{t_n,X(t_n)}_{Y(t_n),\theta(t_n)}(t),\qquad t\in[t_n,t_{n+1}],
\end{equation*}
where $X^{s,x}_{y,\theta}$ is the solution operator of SDE \eqref{equ:hw:x} defined in \autoref{sec:crc:setup}. Here, \autoref{ass:hw:x} is assumed to hold for the parameters $(Y(t_n),\theta(t_n)) \in \mathbb Y\times C(\mathbb R_+)$. 
\item \label{def:crc:item3} The evolution of $h$ on $[t_n,t_{n+1}]$ is determined by $X$ according to the prevailing Hull-White extended affine model:
\begin{equation*}
h(t)=\mathcal H_{Y(t_n)}\big(\theta(t),X(t)\big),\qquad t\in[t_n,t_{n+1}].
\end{equation*}
\end{enumerate}
\end{definition} 
We use the same symbols $h$ and $X$ as in \autoref{equ:hw:hjm} to denote CRC models. The abuse of notation is motivated by the fact that $(h,X)$ evolves on $[t_n,t_{n+1}]$ according to \eqref{equ:hw:hjm} with parameters $(Y(t_n),\theta(t_n))$. Note that the process $X$ in \autoref{def:crc} is continuous because closed intervals $[t_n,t_{n+1}]$ are used in point (ii). We emphasise that the recalibration step \eqref{def:crc:item1}  happens on a discrete time scale because the parameter process $Y$ is constant on each $[t_n,t_{n+1}]$. By construction the process $(h,X)$ is continuous at each time $t_n$.

\subsection{Simulation}\label{sec:crc:sim}

If we assume that a stochastic model for the evolution of the parameter process $Y$ is specified, one can sample $Y$ on the time grid $(t_n)_{n\in\mathbb N_0}$. Then CRC models as in \autoref{def:crc} can be simulated by applying iteratively steps \eqref{def:crc:item1}--\eqref{def:crc:item3}.

\begin{algorithm}[Simulation]\label{alg:crc}
Given $h(t_0)$ and the process $Y$, calculate $(h,X,Y,\theta)$ on the time grid $(t_n)_{n\in\mathbb N_0}$ by iteratively executing steps \eqref{def:crc:item1}--\eqref{def:crc:item3} of \autoref{def:crc}. Abort with an error if the assumption in step \eqref{def:crc:item2} is not satisfied, for any $n\in\mathbb N_0$.
\end{algorithm}

The algorithm is illustrated in \autoref{fig:crc_affine_main}. Note that the forward rate increments are calculated from increments of the affine factor process $X$, which can typically be simulated with high orders of accuracy and proper treatment of boundary conditions. These advantages are thanks to the affine structure of the CRC increments and are not available for general HJM models.

\begin{figure}[h]
\centerline{\xymatrix@=37pt{ 
&
{\circ}
\ar`r/32pt[rd] [rd]^(-1){\text{update }h}
\save 
+<0pt,\baselineskip>*{\big(h(t_n),X(t_{n+1}),Y(t_n),\theta(t_n)\big)}
\restore
&
\\
{\circ}
\ar`u[ru]^(2){\text{update } X} [ru]
\save 
+<-5pt,0pt> *!R{\big(h(t_{n}),X(t_{n}),Y(t_n),\theta(t_n)\big)} 
\restore
&
&
{\circ}
\ar`d[ld]^(2){\text{update } Y} [ld]
\save
+<5pt,0pt> *!L{\big(h(t_{n+1}),X(t_{n+1}),Y(t_n),\theta(t_n)\big)}
\restore
\\
&
{\circ}
\ar`l[lu] [lu]^(-1){\substack{\text{update } \theta \\ \text{increase } n}}
\save
+<0pt,-\baselineskip>*{\big(h(t_{n+1}),X(t_{n+1}),Y(t_{n+1}),\theta(t_n)\big)}
\restore
&
}}
\caption{Simulation of CRC models. Updating $\theta$, $X$, $h$ is done using \eqref{def:crc:item1}, \eqref{def:crc:item2}, \eqref{def:crc:item3} of \autoref{def:crc}, respectively. Updating $Y$ is done using the exogenously given model for $Y$. }
\label{fig:crc_affine_main}
\end{figure}
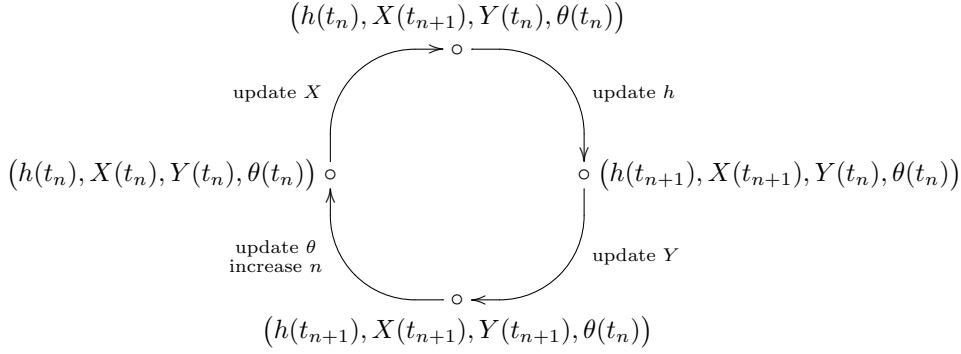

\subsection{Efficient updating of forward rate curves}\label{sec:crc:eff}

Updating the curve of forward rates as prescribed by \autoref{def:crc}\eqref{def:crc:item3} involves calculating integrals on time intervals $[0,\tau]$, for large values of $\tau$ (see \autoref{sec:hw:calibration} for the formulas). A significant speed-up can be obtained when this update is done using the alternative formula provided by the following lemma, which involves only integrals over time intervals of length $\delta=t_{n+1}-t_n$. 

\begin{lemma}[Efficient updating of forward rate curves] \label{lem:crc:eff} 
\autoref{def:crc}\eqref{def:crc:item3} can be rewritten as
\begin{multline}\label{equ:im:eff}
h(t_{n+1})=\mathcal S(\delta)h(t_n)+\mathcal S(\delta)\Phi_{Y(t_n)}'-\Phi_{Y(t_n)}'
+\left\langle \mathcal S(\delta)\Psi_{Y(t_n)}',X(t_n)\right\rangle
\\
-\left\langle \Psi_{Y(t_n)}',X(t_{n+1})\right\rangle
+\int_0^{\delta} \theta(t_n)(s)\left\langle \mathcal S(\delta-s)\Psi_{Y(t_n)}',e_1\right\rangle ds,
\end{multline}
where $\delta=t_{n+1}-t_n$.
\end{lemma}

\begin{proof}
By conditions \eqref{def:crc:item1} and \eqref{def:crc:item3} of \autoref{def:crc}, 
\begin{align*}
h(t_{n+1})-\mathcal S(\delta)h(t_n) 
&= 
\;\mathcal H_{Y(t_n)}\big(\mathcal S(\delta)\theta(t_n),X(t_{n+1})\big)
-\mathcal S(\delta)\mathcal H_{Y(t_n)}\big(\theta(t_n),X(t_n)\big)
\\
&= 
\ell-\mathcal I_{Y(t_n)}\big(\mathcal S(\delta)\theta(t_n)\big)
-\Phi_{Y(t_n)}'-\langle\Psi_{Y(t_n)}',X(t_{n+1})\rangle
\\&\quad
-\ell+\mathcal S(\delta)\mathcal I_{Y(t_n)}\big(\theta(t_n)\big)
+\mathcal S(\delta)\Phi_{Y(t_n)}'\langle\mathcal S(\delta)\Psi_{Y(t_n)}',X(t_n)\rangle.
\end{align*}
Now the assertion of the lemma follows from the relationships
\begin{equation*}
\mathcal S(\delta) \mathcal I_y(\theta)-\mathcal I_y(\mathcal S(\delta) \theta)
=\int_0^\delta \theta(s)\langle \mathcal S(\delta-s)\Psi_y',e_1\rangle ds,
\qquad\text{for all } (\delta,\theta) \in \mathbb R_+\times C(\mathbb R_+),
\end{equation*}
which can easily be verified from the definition. 
\end{proof}

\subsection{Bond prices and forward rates}

\begin{theorem}[Zero-coupon bond price and forward rate]\label{thm:crc}
Let $(h,X,Y)$ be a CRC model as in \autoref{def:crc} with corresponding process $\theta$. Define 
\begin{align*}
P(t,T)&=e^{-\int_t^T h(t,s-t)ds}, 
&& 
r(t)=h(t,0)=\ell+\langle\lambda,X(t)\rangle, 
&& 
B(t)=e^{\int_0^t r(s)ds}.
\end{align*}
Then the discounted price process $t\mapsto P(t,T)/B(t)$ is a $\mathbb P$-local martingale, for each $T\geq 0$. In this sense, the bond market is \emph{free of arbitrage}. Moreover, the following \emph{affine bond pricing formulas} hold:
\begin{align*}
\log(P(t,T))&=
-\ell(T-t)
+\int_0^{T-t} \theta(t)(s)\langle\Psi_{Y(t)}(T-t-s),e_1\rangle ds
\\&\quad
+ \Phi_{Y(t)}(T-t)+\langle\Psi_{Y(t)}(T-t), X(t)\rangle,
\\
h(t,\tau)&=\ell-\int_0^\tau \theta(t)(s)\langle \Psi_{Y(t)}'(\tau-s),e_1\rangle ds
-\Phi_{Y(t)}'(\tau)-\langle \Psi_{Y(t)}'(\tau), X(t)\rangle.
\end{align*}
\end{theorem}

Note: the following proof shows the stronger statement that discounted bond prices are true martingales. 

\begin{proof}
On each interval $[t_n,t_{n+1}]$, the evolution of forward rate curves $h(t)$ stems from a Hull-White extended affine short rate model. Therefore, for each $T\geq 0$, the discounted price process $t\mapsto P(t,T)/B(t)$ is a martingale on each interval $[t_n,t_{n+1}]$. Moreover, the process is continuously concatenated at the boundaries $t_n$ of the intervals. It follows that the process is a martingale on $[0,\infty)$. The affine bond pricing formulas are equivalent to $h(t)=\mathcal H_{Y(t)}(\theta(t),X(t))$, which holds by \autoref{def:crc}\eqref{def:crc:item3}.
\end{proof}

\subsection{Heath-Jarrow-Morton equation}

\begin{theorem}[HJM equation]\label{thm:crc:hjm}
Let $(h,X,Y)$ be a CRC model as in \autoref{def:crc} with corresponding process $\theta$ and assume that $h(t)\in\mathcal D(\mathcal A)$, for each $t\geq 0$. Then the following properties hold:
\begin{enumerate}[(i)]
\item \label{thm:crc:hjm:item1} the expression $\mathcal C_{Y(t)}(h(t),X(t))$ is well-defined and equals $\theta(t)$, for all $t\geq 0$;
\item \label{thm:crc:hjm:item2} the parameters $(Y(t),\theta(t))$ are admissible, for all $t\geq 0$; and
\item \label{thm:crc:hjm:item3} the process $(h,X)$ is a strong solution of the following SPDE on $\mathbb H\times \mathbb X$:
\begin{equation}\label{equ:crc_hjm}\begin{aligned}
dh(t)&=\Big(\mathcal Ah(t)+\mu^{\mathrm{HJM}}_{Y(t)}\big(X(t)\big)\Big)dt+\sigma^{\mathrm{HJM}}_{Y(t)}\big(X(t)\big)dW(t),
\\
dX(t)&=\sqrt{A_{Y(t)}\big(X(t)\big)}dW(t)+\Big(\mathcal C_{Y(t)}\big(h(t),X(t)\big)(0)e_1+B_{Y(t)}\big(X(t)\big)\Big)dt.
\end{aligned}\end{equation}
\end{enumerate}
\end{theorem}

\begin{proof}
This follows from \autoref{def:crc} and \autoref{thm:hw:hjm}.
\end{proof}

\subsection{Geometric interpretation}\label{sec:crc:geometry}

The consistent recalibration scheme has a nice geometric interpretation. 
Forward rate curves of a Hull-White extended affine short rate model remain within the finite dimensional manifold with boundary given by
\[
\bigg\{-\int_0^\tau \theta(t+s)\langle \Psi_y'(\tau-s),e_1\rangle ds+\ell-\Phi_y'(\tau)-\langle \Psi_y'(\tau),x\rangle\bigg|(t,x)\in\mathbb R_+\times\mathbb R^d \bigg\},
\]
as can be seen from \autoref{thm:hw:price}.
These submanifolds foliate the space of forward rate curves or large portions thereof. Let $h$ be a forward rate curve. Then, for every choice of functional characteristics $(F_y,R_y)$, there is at most one leaf through $h$. However, if $(F_y,R_y)$ is allowed to vary, there are in general many leaves through $h$. A choice of leaf corresponds to a choice of foliation and thus to a choice functional characteristics $(F_y,R_y)$. 

A CRC model is constructed by concatenating forward rate evolutions on leaves belonging to different foliations. This allows the otherwise constant coefficients $(F_y,R_y)$ to change over time. The result is an HJM model which is \emph{``tangent''} to Hull-White extended affine short rate models. This is illustrated in \autoref{fig:foliations}.

\begin{figure}
\centering 
\adjustbox{trim={.2\width} {.28\height} {.2\width} {.24\height},clip}%
  {\includegraphics[width=.7\textwidth]{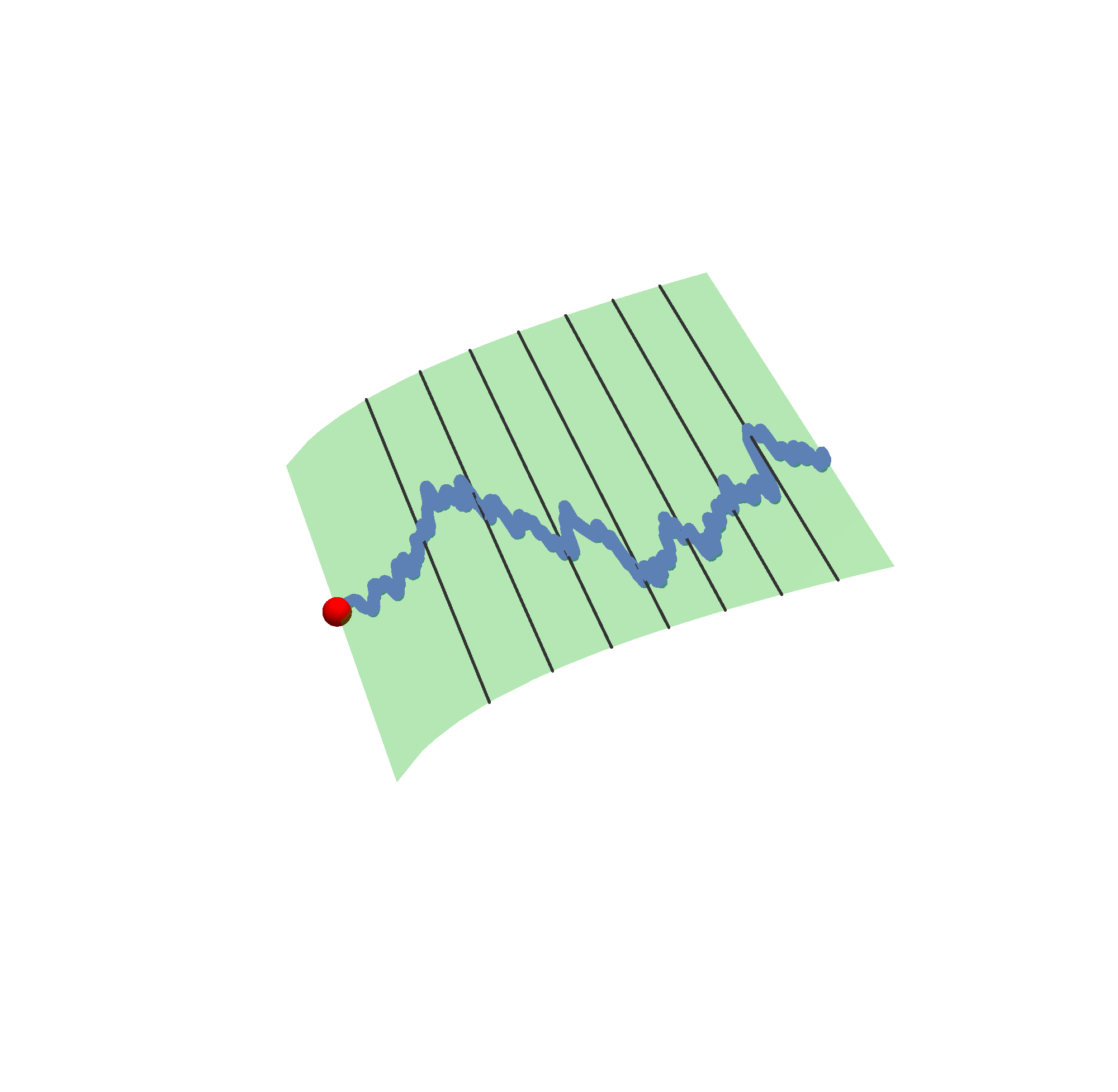}}
\adjustbox{trim={.2\width} {.28\height} {.2\width} {.24\height},clip}%
  {\includegraphics[width=.7\textwidth]{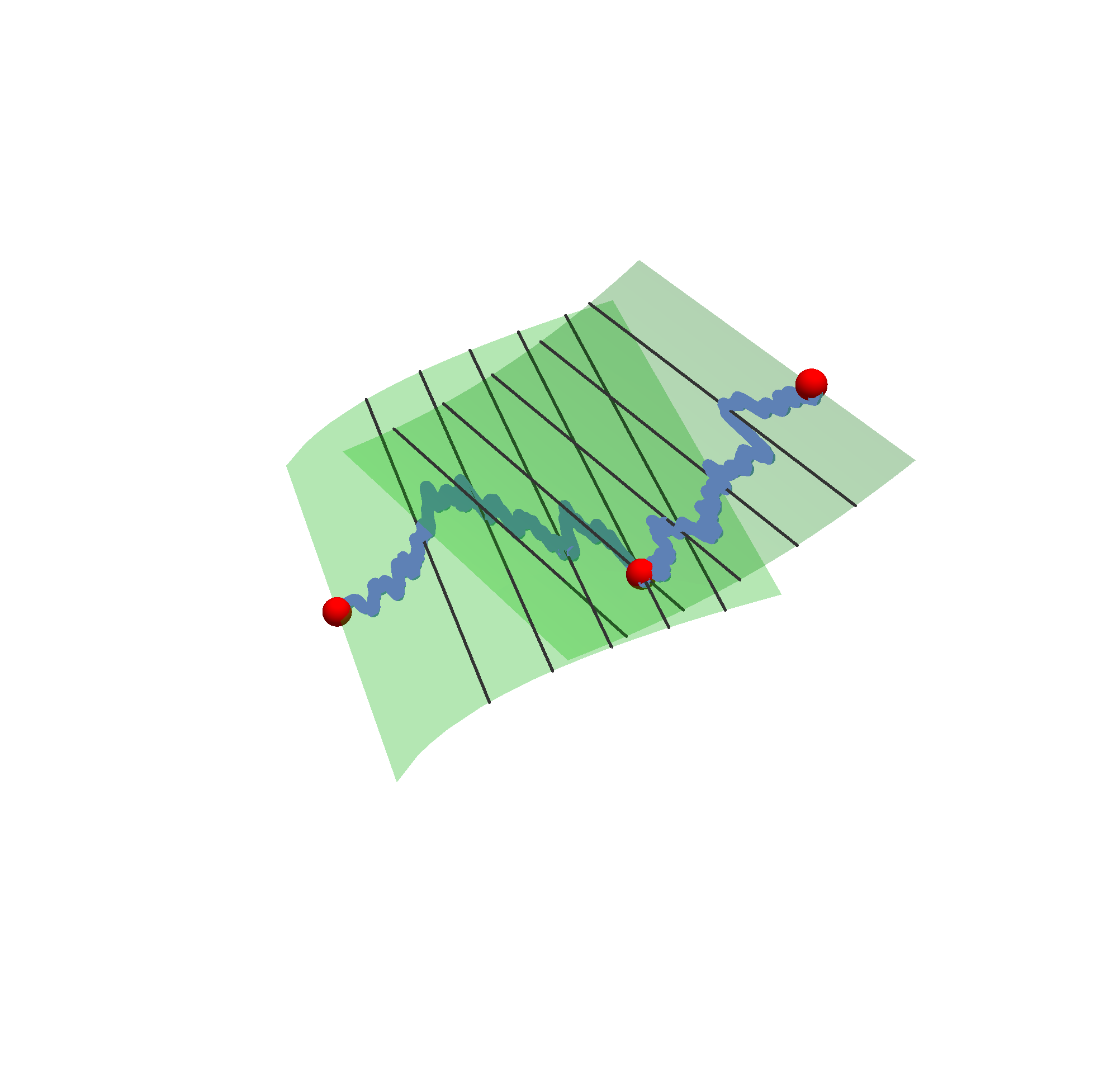}}
\caption{Each affine short rate model foliates the space of forward rate curves into invariant leaves. CRC models are concatenations of forward rate evolutions belonging to different foliations or limits of such concatenations.}
\label{fig:foliations}
\end{figure}

\subsection{CRC models}\label{sec:crc:continuous}

We generalise CRC models of \autoref{sec:crc:piecewise} to arbitrary parameter processes $Y$. In this section $Y$ is not required to be piecewise constant as in the last sections. To characterise such models, we use the SPDE derived in \autoref{thm:crc:hjm}. 
\begin{definition}[CRC models]\label{def:crc:continuous}
A \emph{CRC model} is a process $(h,X,Y)$ with values in $\mathbb H\times \mathbb X\times\mathbb Y$ satisfying conditions \eqref{thm:crc:hjm:item1}--\eqref{thm:crc:hjm:item3} of \autoref{thm:crc:hjm} with $\theta(t)=\mathcal C_{Y(t)}\big(h(t),X(t)\big)$ for each $t\geq0$. 
\end{definition}
We may think of CRC models as continuous-time limits of the concatenations described in \autoref{sec:crc:geometry}. Note that these models satisfy all conclusions of \autoref{thm:crc}: they are free of arbitrage because discounted bond prices $P(t,T)/B(t)$ are local martingales thanks to the HJM drift condition \cite[Theorem~6.1]{filipovic2009term}, and the short rate can be written as $r(t)=h(t)(0)=\ell+\langle \lambda,X(t)\rangle$. 

\subsection{Semigroup interpretation}\label{sec:crc:semigroup}

Assume that the parameter process $Y$ is Markov on $\mathbb Y$, and that the SPDE \eqref{equ:crc_hjm} has a unique mild solution which depends continuously on the initial condition. Then the CRC model $(h,X,Y)$ in \autoref{def:crc:continuous} is Markov on $\mathbb H\times\mathbb X\times\mathbb Y$ \cite[c.f.~Theorem~9.14]{daPrato2014se}. Let $\mathcal P=(\mathcal P(t))_{t\geq0}$ denote the corresponding semigroup on the Banach space $C_b(\mathbb H\times\mathbb X\times\mathbb Y)$ of bounded continuous functions, i.e., 
\begin{align*}
\mathcal P(t)f(h_0,x_0,y_0)
&=
\mathbb E\left[f\big(h(t),X(t),Y(t)\big)
\middle|\big(h(0),X(0),Y(0)\big)=\big(h_0,x_0,y_0\big)\right].
\end{align*}
Moreover, let $\mathcal Q=(\mathcal Q(t))_{t\geq0}$ denote the semigroup on $C_b(\mathbb H\times\mathbb X\times\mathbb Y)$ obtained by holding the parameter process $Y(t)\equiv y$ fixed, i.e., 
\begin{align*}
\mathcal Q(t)f(h_0,x_0,y_0)
&=
\mathbb E\left[f\big(h(t),X(t),y_0\big)
\middle|\big(h(0),X(0)\big)=\big(h_0,x_0\big)\right],
\end{align*}
where $(h,X)$ are as in \autoref{thm:hw:hjm} with $y=y_0$. Finally, let $\mathcal R=(\mathcal R(t))_{t\geq0}$ denote the semigroup on $C_b(\mathbb H\times\mathbb X\times\mathbb Y)$ describing the evolution of $Y$, i.e., 
\begin{align*}
\mathcal R(t)f(h_0,x_0,y_0)
&=
\mathbb E\left[f\big(h_0,x_0,Y(t)\big)
\middle|Y(0)=y_0\right].
\end{align*}
Then, the concatenation $(\mathcal R(\delta)\mathcal Q(\delta))^n f$ of semigroups describes CRC models with a piecewise constant parameter process as in \autoref{def:crc}. Indeed, 
\begin{equation*}
(\mathcal R(\delta)\mathcal Q(\delta))^n f (h_0,x_0,y_0)
=
\mathbb E\left[h(t_n),X(t_n),Y(t_n)\right],
\qquad\text{for all } n\in\mathbb N_0,
\end{equation*}
where $\delta=t_{n+1}-t_n$ is the step size, $(h(t_n),X(t_n))$ is obtained by executing the simulation scheme of \autoref{sec:crc:sim} and $Y(t_n)$ by sampling the Markov process $Y$ on that time grid.

\subsection{Simulation of CRC models by splitting schemes}

The semigroup interpretation of \autoref{sec:crc:semigroup} allows one to view  \autoref{alg:crc} as an exponential Euler \emph{splitting scheme} for general CRC models as in \autoref{def:crc:continuous}. To see this, let $f:\mathbb H\times\mathbb X\times\mathbb Y\to \mathbb R$ be twice differentiable with derivatives being uniformly continuous on bounded sets and assume that $\mathbb H \subseteq \operatorname{dom}(\mathcal A)$. Then, under appropriate assumptions on $Y$, It\=o's formula holds for $f(h(t),X(t),Y(t))$ by \cite[Theorem 4.17]{daPrato2014se}. It follows that $f$ lies in the common domain of the generators $\mathcal G^{\mathcal P},\mathcal G^{\mathcal Q},\mathcal G^{\mathcal R}$ of the semigroups $\mathcal P,\mathcal Q,\mathcal R$, and, if $Y$ is independent of $W$,
\begin{align*}
\mathcal G^{\mathcal P} f = \mathcal G^{\mathcal Q} f+\mathcal G^{\mathcal R} f.
\end{align*}
The exponential Euler splitting scheme with respect to this splitting is defined as
\begin{equation*}
\mathcal P(n\delta)f\approx
\big(\exp(\delta \mathcal G^{\mathcal R})\exp(\delta\mathcal G^{\mathcal Q})\big)^n f=\big(\mathcal R(\delta)\mathcal Q(\delta)\big)^n f,
\qquad\text{for all } n\in\mathbb N_0.
\end{equation*}
By the considerations in \autoref{sec:crc:semigroup}, it coincides with the simulation scheme of \autoref{sec:crc:sim}. The advantages of this simulation scheme in comparison to other methods are discussed in Sections~\ref{sec:intro:hjm} and \ref{sec:crc:properties}.

\subsection{Calibration of CRC models}\label{sec:crc:calibration}

In order to calibrate CRC models we need to estimate a time series of the parameter process $Y$ from market data, and fit a model for this time series. Estimating a time series of the parameter process $Y$ can be done as explained in  \autoref{sec:hw:estimation} for Hull-White extended affine models. The resulting time series $Y$ consists of model parameters under a risk neutral measure, but they can be estimated from real world observations since the estimators are obtained solely from the volatility of the forward rate process. 

Calibrating CRC models requires the additional task of selecting and fitting a model for the estimated time series of $Y$. This completes the model specification under a risk neutral probability measure. We do not discuss the market price of risk specification in this paper but refer to \cite{harms2016consistent} for more details. 

\subsection{Robust calibration, consistency, and analytic tractability}\label{sec:crc:properties}

It is time to address the question to what extent CRC models satisfy the interest rate modelling principles set forth in the introduction. First, we formalise the notion of consistency and the consistent recalibration property. We do this for CRC models, but it is unproblematic to generalise the definition to other forward rate models. Let $\operatorname{pr}_{\mathbb H}$ denote the projection of $\mathbb H\times\mathbb X\times \mathbb Y$ onto $\mathbb H$.

\begin{definition}\label{def:consistency}
Let $(h(t),X(t),Y(t))_{t\geq 0}$ be a CRC model and $\mathcal I\subseteq \mathbb H\times \mathbb X\times \mathbb Y$. Then the model is called \emph{consistent} with $\mathcal I$ if $(h(t),X(t),Y(t))\in\mathcal I$ holds with probability one for any $t>0$ and  initial condition $(h_0,x_0,y_0)\in\mathcal I$. Moreover, the model satisfies the \emph{consistent recalibration property} with respect to $\mathcal I$ if the support of the law of $h(t)$ contains $\operatorname{pr}_{\mathbb H}(\mathcal I)$ for any $t>0$ and initial condition $(h_0,x_0,y_0)\in\mathcal I$.
\end{definition}

The consistent recalibration property is equivalent to any open subset of $\operatorname{pr}_{\mathbb H}(\mathcal I)$ being reached at any time $t>0$ with positive probability. Observe that the consistent recalibration property does not hold on any reasonably large set $\mathcal I$ for Hull-White extended affine factor models as in \autoref{sec:hw}. Indeed, as explained in \autoref{sec:crc:geometry}, for any given initial value the process $h$ remains within a finite dimensional submanifold of $\mathbb H$. CRC models, on the other hand, enjoy the following properties:

\begin{itemize}
\item {\bf Robust calibration:} the robust calibration principle is satisfied perfectly. Indeed, the method described in \autoref{sec:crc:calibration} allows us to use the entire present and past market data of yields to select a model. Whenever possible, the parameters are estimated from realised covariations of yields, which allows one to bypass the usual inverse problems in calibration.

Moreover, requiring parameters to remain constant throughout the life time of the model is less restrictive in CRC models than in the underlying affine factor models. The reason is that the parameters of the underlying affine model are turned into state variables of the CRC model. 

\item {\bf Consistency:} the canonical state space $\mathcal I$ of CRC models is the subset of $\mathbb H\times\mathbb X\times\mathbb Y$ determined by the admissibility condition on the underlying Hull-White extended affine factor model (see \autoref{ass:hw:x}). Under sensible specifications of the affine factor model, $\mathcal I \cap (\mathbb H\times \{x\}\times\{y\})$ is large enough to contain all realistic market curves, for each fixed $(x,y) \in \mathbb X\times \mathbb Y$. If the Hilbert space $\mathbb H$ is continuously embedded in $C^1(\mathbb R_+)$, then $\mathcal I$ is also large in the topological sense of having non-empty interior. 

In this setup consistency holds by construction because the state process $(h,X,Y)$ of CRC models does not leave the set $\mathcal I$. The consistent recalibration property can be verified using standard arguments: the support of $(h(t),X(t),Y(t))$ is the closure of the reachable set of an associated control problem \cite{nakayama2004support}, the reachable set is stable under the flows of the driving vector fields and their Lie brackets, and generically speaking, as soon as there is noise in the parameter process $Y$, these vector fields together with their brackets span dense subspaces of $\mathbb H$. The exact conditions are worked out for the Vasi\v cek case in \autoref{sec:va:con}. 

\item {\bf Analytic tractability:} the simulation scheme for CRC models (\autoref{alg:crc}) transfers the task of sampling state variable increments to a finite-dimensional setting. Namely, instead of simulating forward rate increments from an infinite-dimensional space, it is sufficient to simulate increments of the finite-dimensional processes $X$ and $Y$. This allows one to take advantage of the existing high-order schemes for the simulation of affine processes. Note that all operations in \autoref{alg:crc} which involve infinite-dimensional objects are deterministic. The complexity for simulation reduces dramatically when high order methods for finite dimensional equations are applied which are often not at hand for infinite dimensional equations. Additionally often exact schemes are available in the affine finite dimensional setting, e.g., for CIR or Wishart type processes besides of course Gaussian processes.
\end{itemize}

\section{Consistent recalibration of \texorpdfstring{Vasi\v cek}{Vasicek} models} \label{sec:va}

\subsection{Overview}

We describe CRC models based on the Hull-White extended Vasi\v cek model in full detail. Moreover, we show using semigroup theory that the simulation scheme of \autoref{sec:crc:sim} converges to the CRC model of \autoref{def:crc:continuous} in the continuous-time limit.

\subsection{Setup and notation}

We use the setup of \Autoref{sec:hw:setup,sec:crc:setup}, setting $\mathbb X=\mathbb R$, $\ell=0$, $\lambda=1$. We do not specify the parameter space $\mathbb Y$, yet, but we assume that for each $(x,y)\in\mathbb X\times\mathbb Y$, the volatility and drift coefficients are given by $A_y(x)=a_y \in [0,\infty)$ and $B_y(x)=\beta_y x$ with $\beta_y \in (-\infty,0)$. For simplicity, we choose equidistant grids of times $t_n=n\delta$ and times to maturity $\tau_n=n\delta$, for all $n\in\mathbb N_0$, where $\delta$ is a positive constant. 

\subsection{Hull-White extended \texorpdfstring{Vasi\v cek}{Vasicek} models}\label{sec:va:hw}

For each parameter $(y,\theta)\in\mathbb Y\times C(\mathbb R_+)$, the SDE for the short rate process is
\begin{equation}\label{equ:va:sde}
dr(t)=(\theta(t)+\beta_y r(t))dt+\sqrt{a_y}dW(t),
\end{equation}
where $W$ is one-dimensional $(\mathcal F(t))_{t\geq0}$-Brownian motion. 
\autoref{ass:hw:x} is satisfied for each parameter $(y,\theta)$. The functional characteristics $(F,R)$ from \autoref{sec:hw:riccati} are 
\begin{equation*}
F_y(u)=\frac{a_y}{2}  u^2, 
\qquad
R_y(u)=\beta_y u, 
\qquad
\text{for all } u \in \mathbb R,
\end{equation*}
and the solutions $(\Phi_y,\Psi_y)$ of the corresponding Riccati equations are
\begin{equation*}
\Phi_y(t)=\frac{a_y}{4\beta_y^3}\left(2\beta_y t-4e^{\beta_y t}+3+e^{2\beta_y t}\right),
\qquad
\Psi_y(t)=\frac{1}{\beta_y}\left(1-e^{\beta_y t}\right),
\qquad\text{for all } t\geq 0.
\end{equation*}
By \autoref{thm:hw:price}, the forward rates in the Hull-White extended Vasi\v cek model \eqref{equ:va:sde} with fixed parameters $(y,\theta)$ are given by $h(t)=\mathcal H_y(\mathcal S(t)\theta,r(t))\in C^1(\mathbb R_+)$, where 
\begin{align*}
\mathcal H_y(\theta,x)(\tau)
&=
\int_0^\tau \mathcal \theta(s)e^{\beta_y (\tau-s)} ds 
-\frac{a_y}{2\beta_y^2}\left(1-e^{\beta_y \tau}\right)^2
+e^{\beta_y \tau} x,
\end{align*}
for all $(x,\theta,\tau)\in\mathbb R\times C(\mathbb R_+)\times\mathbb R_+$. Due to the simple structure of the integral kernel $e^{\beta_y (\tau-s)}$, there is a closed-form expression for the calibration operator, 
\begin{equation}\label{equ:va:c}
\mathcal C_y(h)(\tau)=h'(\tau)-\beta_y h(\tau)-\frac{a_y}{2\beta_y}\left(1-e^{2\beta_y\tau}\right),
\qquad
\text{for all } (h,\tau)\in C^1(\mathbb R_+)\times\mathbb R_+.
\end{equation}
This can be verified using the definitions. 
Note that the calibration operator does not depend on $x$. Therefore, we dropped $x$ from the notation $\mathcal C_y(h,x)$. 

The HJM drift and volatility from \autoref{sec:hw:hjm} are 
\begin{equation}\label{equ:va:hjm_drift_vola}
\mu^{\mathrm{HJM}}_y(\tau)=-\frac{a_y}{\beta_y}e^{\beta_y\tau}\left(1-e^{\beta_y\tau}\right),
\qquad
\sigma^{\mathrm{HJM}}_{y}(\tau)=\sqrt{a_y}e^{\beta_y\tau}, 
\qquad\text{for all } \tau \in \mathbb R_+.
\end{equation}
Note that these expressions do not depend on $x$, which is why we again dropped $x$ from the previous notation $\mu^{\mathrm{HJM}}_y(x)(\tau), \sigma^{\mathrm{HJM}}_y(x)(\tau)$. The HJM equation for forward rates then reads as
\begin{equation}\label{equ:va:hjm}
dh(t)=\Big(\mathcal Ah(t)+\mu_y^{\mathrm{HJM}}\Big)dt+\sigma_y^{\mathrm{HJM}}dW(t).
\end{equation}

\subsection{\texorpdfstring{Vasi\v cek}{Vasicek} CRC models}\label{sec:va:crc}

Since the factor process is a function of the forward rate process, i.e., $X(t)=r(t)=h(t,0)$, the corresponding CRC models can be characterised by the process $(h,Y)$ instead of $(h,X,Y)$. Thus, in accordance with \autoref{thm:crc:hjm} and \autoref{def:crc:continuous}, a process $(h,Y)$ with values in $\mathbb H\times \mathbb Y$ may be called a CRC model if $h$ satisfies the SPDE 
\begin{align}\label{equ:va:crc}
dh(t)&=\Big(\mathcal Ah(t)+\mu^{\mathrm{HJM}}_{Y(t)}\Big)dt
+\sigma^{\mathrm{HJM}}_{Y(t)}dW(t),
\end{align}
with drift $\mu^{\mathrm{HJM}}_{Y(t)}$ and volatility $\sigma^{\mathrm{HJM}}_{Y(t)}$ defined in \eqref{equ:va:hjm_drift_vola}. Beyond the obvious requirement that these quantities are well-defined, for all $t \in \mathbb R_+$, no further conditions are needed. In other words, the maximally admissible set $\mathcal I$ in the Vasi\v cek case is the entire Hilbert space $\mathbb H$. 

\subsection{Simulation of \texorpdfstring{Vasi\v cek}{Vasicek} CRC models} \label{sec:va:sim}

Given the parameter process $Y$ with values in $\mathbb Y$, the CRC model is simulated as described in \autoref{alg:crc}. The following observations make the algorithm particularly effective. First, the state process $X$ is a function of the forward rate and can be eliminated as a state variable. Second, the short rate process can be simulated exactly. Indeed, in the model with constant parameter $y$, $r(t)$ is normally distributed,
\begin{equation*}
r(t) \sim 
\mathcal{N}\bigg(e^{\beta_y t}r_0+\int_{0}^t e^{\beta_y\left(t-s\right)}\theta(s)ds,
\frac{a_y}{2\beta_y}\left(e^{2\beta_y t}-1\right)\bigg).
\end{equation*}
Third, inverting the Volterra integral operator can be avoided by using closed-form expression \eqref{equ:va:c} of the calibration operator.

Discretisation is done on the uniform grid $t_n=\tau_n=\delta n$ for a choice of finitely many times to maturity $\tau_n$. Integrals are approximated to second order by the trapezoid rule, which leads to a global error of order one (see \autoref{sec:va:convergence} and \autoref{sec:num:con}). The resulting scheme works as follows. 

\begin{algorithm}[Simulation]\label{alg:va:crc}
Given $(h(0),\mathcal Ah(0))$ and the parameter process $Y$, execute iteratively the following steps, for each $n\in\mathbb N_0$:

\begin{enumerate}[(i)]
\item The values of $\theta(t_n)=\mathcal C_{Y(t_n)}(h(t_n))$ at times to maturity $0$ and $\delta$ are calculated using \eqref{equ:va:c},
\begin{align*}
\theta(t_n)(0)&=\mathcal Ah(t_n)(0)-\beta_{Y(t_n)} h(t_n)(0),
\\
\theta(t_n)(\delta)&=\mathcal Ah(t_n)(\delta)-\beta_{Y(t_n)} h(t_n)(\delta)-\frac{a_{Y(t_n)}}{2\beta_{Y(t_n)}}\left(1-e^{2\beta_{Y(t_n)}\delta}\right),
\end{align*}
and $\mathcal I_{Y(t_n)}(\theta(t_n))(\delta)$ is approximated by the trapezoid rule as follows:
\begin{align*}
\widehat{\mathcal I}_{Y(t_n)}\big(\theta(t_n)\big)(\delta)
=-\frac{\delta}{2}\left(e^{\beta_{Y(t_n)}\delta}\theta(t_n)(0)
+\theta(t_n)(\delta)\right).
\end{align*}

\item A sample $r(t_{n+1})$ is drawn such that conditionally on $\mathcal F(t_n)$, $r(t_{n+1})$ has normal distribution
\begin{equation*}
r(t_{n+1})\sim \mathcal N \left(e^{\beta_{Y(t_n)} \delta}h(t_n)(0)
-\widehat{\mathcal I}_{Y(t_n)}\big(\theta(t_n)\big)(\delta),
\frac{a_{Y(t_n))}}{2\beta_{Y(t_n)}}\left(e^{2\beta_{Y(t_n)} \delta}-1\right)\right).
\end{equation*}

\item $\big(h(t_{n+1}),\mathcal Ah(t_{n+1})\big)$ is calculated from $\big(h(t_n), \mathcal Ah(t_n),r(t_{n+1})\big)$ using \autoref{lem:crc:eff}:

\begin{align*}
h(t_{n+1})(\tau)&=h(t_n)(\delta+\tau)
+\frac{a_{Y(t_n)}}{2\beta^2_{Y(t_n)}}
\left(\left(1-e^{\beta_{Y(t_n)}(\delta+\tau)}\right)^2-\left(1-e^{\beta_{Y(t_n)}\vphantom{(}\tau}\right)^2\right)
\\&\qquad
+e^{\beta_{Y(t_n)}\tau} \left(
-e^{\beta_{Y(t_n)}\delta}r(t_n)+r(t_{n+1}) 
+\widehat{\mathcal I}_{Y(t_n)}\big(\theta(t_n)\big)(\delta)
\right),
\\
\mathcal Ah(t_{n+1})(\tau)&=\mathcal Ah(t_{i})(\delta+\tau)
\\&\qquad
+\frac{a_{Y(t_n)}}{\beta_{Y(t_n)}}\left(e^{\beta_{Y(t_n)}\tau}+e^{2\beta_{Y(t_n)}(\tau+\delta)}-e^{2\beta_{Y(t_n)}\tau}-e^{\beta_{Y(t_n)}(\delta+\tau)}\right)
\\&\qquad
+\beta_{Y(t_n)}e^{\beta_{Y(t_n)}\tau} \left(
-e^{\beta_{Y(t_n)}\delta}r(t_n)+r(t_{n+1}) 
+\widehat{\mathcal I}_{Y(t_n)}\big(\theta(t_n)\big)(\delta)\right).
\end{align*}
Here, $h(t_{n+1})$ must be calculated at all times to maturity $\tau_i$, whereas $\mathcal Ah(t_{n+1})$ is needed only at $\tau_0=0$ and $\tau_1=\delta$.
\end{enumerate}
\end{algorithm}

\subsection{Convergence of the simulation scheme}\label{sec:va:convergence}

In this section, we show that scheme of \autoref{sec:va:sim} converges to the CRC model \eqref{equ:va:crc} as the size $\delta$ of the time grid tends to zero. We are not aiming for the highest generality. Instead, we show how the results follow from standard semigroup theory.

\begin{assumption}\label{ass:va:Y}
The parameter process $Y$ takes values in $\mathbb Y= \mathbb R^p$ and satisfies
\begin{align}\label{equ:va:y}
dY(t)&=\big(A Y(t)+\mu(Y(t))\big)dt + \sigma(Y(t))d\widetilde W(t),
\end{align}
where $A:\mathbb R^p\to\mathbb R^p$ is a linear mapping generating a semigroup of contractions on $\mathbb R^p$, $\mu \in C^\infty_b(\mathbb R^p;\mathbb R^p)$, $\sigma \in C^\infty_b(\mathbb R^p,\mathbb R^{p\times q})$, and $\widetilde W$ is $q$-dimensional $\mathcal F(t)$-Brownian motion, independent of $W$. We write $C^\infty_b$ for bounded functions with bounded derivatives of all orders. The above SDE has a unique solution for any initial condition $Y(s)=y$, where $(s,y)\in\mathbb R_+\times \mathbb Y$.
\end{assumption}

\begin{assumption}\label{ass:va:coeff}
The mappings $y\mapsto \sqrt{a_y}$ and $y\mapsto\beta_y$ are of class $C^\infty_b(\mathbb R^p)$ and $\sup_{y\in\mathbb Y}\beta_y<0$ holds.
\end{assumption}

As Vasi\v cek CRC models can be characterised in terms of $(h,Y)$ instead of $(h,X,Y)$, the semigroups $\mathcal P$, $\mathcal Q$, and $\mathcal R$ from \autoref{sec:crc:semigroup} are now assumed to be defined on $C_b(\mathbb H\times \mathbb Y)$ instead of $C_b(\mathbb H\times \mathbb X \times \mathbb Y)$. Recall that $\mathcal P$ describes the joint evolution of the process $(h,Y)$, $\mathcal Q$ the evolution of $h$ with $Y$ fixed, and $\mathcal R$ the evolution of $Y$ with $h$ fixed.

\begin{theorem}\label{thm:va:convergence}
There exists a separable Hilbert space $\mathbb H$ of continuous functions on $\mathbb R_+$ and a Banach space $\mathbb B$ of continuous functions on $\mathbb H\times \mathbb Y$ such that $\mathcal P$, $\mathcal Q$, and $\mathcal R$ are strongly continuous semigroups on $\mathbb B$. Moreover, for each $T \in \mathbb R_+$ there exists a constant $C$ such that
\begin{equation*}
\left\| \mathcal P(t) f-\left(\mathcal R(t/n)\mathcal Q(t/n)\right)^n f\right\|_{\mathbb B} 
\leq 
C n^{-1}\left\| f \right\|_{\mathbb B'}, 
\qquad
\text{for all } f\in\mathbb B',t\in[0,T],n\in\mathbb N^+,
\end{equation*}
where $\mathbb B'$ is a Banach space which is densely and continuously embedded in $\mathbb B$.
\end{theorem}

The space $\mathbb B'$ is large enough to be relevant in applications: any $C^4$ function on $\mathbb H_0\times \mathbb Y$ belongs to $\mathbb B'$, where $\mathbb H_0 \supset \mathbb H$ is defined in the proof below.

\begin{proof}
We proceed as in \cite{doersek2013efficient} and \cite{filipovic2001consistency}. Let $(\gamma_i)_{i\in\mathbb N_0}$ be a strictly increasing sequence of real numbers strictly greater than 3. For each $i \in \mathbb N_0$, define a separable Hilbert space $\mathbb H_i$ by
\begin{align*}
\mathbb H_i &= \left\{h \in L^1_{\mathrm{loc}}\colon h^{(j)} \in L^1_{\mathrm{loc}}\text{ and } \int_{(0,\infty)}h^{(j)}(\tau)^2 (1+\tau)^{\gamma_i}d\tau<\infty, \forall j =1,\dots,i+1\right\},
\end{align*}
where $L^1_{\mathrm{loc}}$ denotes the space of locally integrable functions on $(0,\infty)$.
Every function in $\mathbb H_0$ is continuous, bounded and has a well-defined limit 
$h(\infty)=\lim_{\tau\to\infty} h(\tau)$.
The scalar product on $\mathbb H_i$ is defined by
\begin{align*}
\langle h_1,h_2\rangle_{\mathbb H_i} &= h_1(\infty)h_2(\infty)
+\sum_{j=1}^i \int_{(0,\infty)} h_1^{(j)}(\tau)h_2^{(j)}(\tau) (1+\tau)^{\gamma_i}d\tau.
\end{align*}
For each $\zeta>0$ and $k,i \in\mathbb N_0$, we define the space $\mathbb B^\zeta_k(\mathbb H_i\times\mathbb Y)$ as the closure of $C^k_b(\mathbb H_i\times \mathbb Y)$ under the norm
\[
\lVert f \rVert_{\mathbb B^\zeta_k(\mathbb H_i\times\mathbb Y)} 
= \sum_{j=0}^k \sup_{(h,y)\in \mathbb H_i\times\mathbb Y} 
\left(\cosh(\zeta \lVert h\rVert_{\mathbb H_i})
+\lVert y\rVert_{\mathbb Y}^2\right)^{-1}
\lVert D^j f(h,y) \rVert_{L((\mathbb H_i\times \mathbb Y)^j)}.
\]
Together with \autoref{ass:va:Y} and \ref{ass:va:coeff}, this implies that the conditions of \cite[Sections 3.1.1 and 3.1.2]{doersek2013efficient} are satisfied for SPDE \eqref{equ:va:crc}, \eqref{equ:va:y} characterising the evolution of $(h,Y)$. (Note that $\beta_y$ needs to be bounded away from zero for $\mu_y^{\mathrm{HJM}}$ and $\sigma_y^{\mathrm{HJM}}$ to be bounded with bounded derivatives.)
Thus, this SPDE admits unique solutions on each space $\mathbb H_i\times\mathbb Y$, given that the initial condition is smooth enough. The same applies to the SPDE for $h$ with fixed $y$ and the SDE for $Y$ with fixed $h$.

Fix $\zeta_0>\zeta>0$ and define $\mathbb H=\mathbb H_2$, $\mathbb B=\mathbb B^{\zeta_0}_0(\mathbb H_2\times\mathbb Y)$, and $\mathbb B'=\mathbb B^\zeta_4(\mathbb H_0\times\mathbb Y)$. Then $(\mathcal P(t))_{t\geq 0}$, $(\mathcal Q(t))_{t\geq 0}$, and $(\mathcal R(t))_{t\geq 0}$ are strongly continuous semigroups on $\mathbb B$ by \cite[Lemma 13]{doersek2013efficient} and quasicontractive by \cite[Lemma 7]{doersek2013efficient}. Their generators are denoted by $\mathcal G^{\mathcal P}$, $\mathcal G^{\mathcal Q}$, and $\mathcal G^{\mathcal R}$. By the same lemma, $\mathbb B'$ is stable under $(\mathcal P(t))_{t\geq 0}$. 
Together with \cite[Theorem 11]{doersek2013efficient} this implies that for each $f\in\mathbb B'$, the expressions 
\begin{align*}
\mathcal G^{\mathcal P} \mathcal P(t)f, && 
\mathcal G^{\mathcal Q} \mathcal P(t)f, &&
\mathcal G^{\mathcal R} \mathcal P(t)f, &&
\mathcal G^{\mathcal Q} \mathcal G^{\mathcal Q} \mathcal P(t)f, &&
\mathcal G^{\mathcal Q} \mathcal G^{\mathcal R} \mathcal P(t)f, &&
\mathcal G^{\mathcal R} \mathcal G^{\mathcal Q} \mathcal P(t)f
\end{align*}
are well-defined with $\mathbb B$-norm bounded uniformly in $t\in[0,T]$ and $\mathcal G^{\mathcal P} f=\mathcal G^{\mathcal Q} f+\mathcal G^{\mathcal R} f$. Thus, the splitting is of formal order one and the result follows from \cite[Theorem 2.3 and Section 4.4]{hansen2009exponential}.
\end{proof}

\subsection{Consistent recalibration property}\label{sec:va:con}

If the coefficient $\beta$ in the HJM volatility $\sqrt{a} e^{\beta \tau}$ is stochastic, one would expect the forward rate process to reach every point in the Hilbert space with positive probability, i.e., the consistent recalibration property holds. This is made precise here. 

\begin{assumption}\label{ass:support}
For each initial condition $Y(0)=y_0 \in \mathbb Y$ and each $T>0$, the support of $Y_T$ is all of $\mathbb Y$. Moreover, $\{\beta_y:y\in\mathbb Y\}$ contains an interval $[\underline\beta,\infty)$ for some $\underline\beta$. 
\end{assumption}

\begin{assumption}\label{ass:abscissa}
The Hilbert space $\mathbb H$ is contained in $L^1_{\mathrm{loc}}(\mathbb R_+)$, and each $h \in \mathbb H$ has a finite abscissa 
\begin{equation*}
\operatorname{abs}(h)=\inf\left\{\beta \in \mathbb R: \int_0^\infty h(\tau)e^{\beta\tau}d\tau<\infty\right\}<\infty.
\end{equation*}
\end{assumption}

The condition in \autoref{ass:abscissa} is mild; it is satisfied by the weighted Sobolev spaces in \cite{filipovic2001consistency} and in particular by the space $\mathbb H$ of \autoref{thm:va:convergence}. The above assumptions imply the consistent recalibration property, as the following theorem shows.

\begin{theorem}\label{thm:va:con}
The consistent recalibration property is satisfied for the Vasi\v cek CRC model \eqref{equ:va:crc}, \eqref{equ:va:y} with respect to the state space $\mathcal I=\mathbb H\times\mathbb Y$. 
\end{theorem}

\begin{proof}
Let $(h,Y)$ be the solution of \eqref{equ:va:crc} and \eqref{equ:va:y} with initial value $(h_0,Y_0)$ and let $T\geq 0$. By \cite[Theorem~1.1]{nakayama2004support} the support of $(h_T,Y_T)$ is equal to the closure $\overline{\mathcal L_T}$ of $\mathcal L_T$, where $\mathcal L_T$ is the reachable set at time $T$ of the following control problem: in \eqref{equ:va:crc} and \eqref{equ:va:y} the Brownian motions are replaced by piecewise continuously differentiable control functions. Let $(\hat h,\hat Y)$ be the solution of \eqref{equ:va:crc} and \eqref{equ:va:y} for the zero control. Taking variations in the control for $\hat Y$ implies that $\{\hat h_T\}\times\mathbb Y \subseteq \overline{\mathcal L_T}$ thanks to \autoref{ass:support}. Adding variations in the control of $h$ improves this to $(\{\hat h_T\}+\operatorname{span}\{\sigma_y^{\mathrm{HJM}}: y \in \mathbb Y\}) \times \mathbb Y \subseteq \overline{\mathcal L_T}$. The set $\operatorname{span}\{\sigma_y^{\mathrm{HJM}}: y \in \mathbb Y\}$ is dense in $\mathbb H$ because its orthogonal complement vanishes by \cite[Proposition~1.7.2]{arendt2011vector} and \autoref{ass:abscissa}. Therefore, $\mathbb H\times \mathbb Y \subseteq\overline{\mathcal L_T}$, and the consistent recalibration property holds.
\end{proof}

\subsection{Example}\label{sec:va:example}

We present an example of a Vasi\v cek CRC model based on \cite{chiarella2004class}. In this model, the volatility is stochastic, but the speed of mean reversion is not. Therefore, the conditions of \autoref{thm:va:con} are not satisfied, and it turns out that the model admits a finite-dimensional realisation. The explicit formula for bond prices in this model will serve as a reference for showing convergence of the numerical simulation scheme for CRC models in the continuous-time limit.

The parameter process $Y$ is a CIR process with values in $\mathbb Y=\mathbb R_+$ given by SDE \eqref{equ:va:y} with a possibly correlated Brownian motion $\widetilde W$ and coefficients $\mu(y)=m+\mu y$ and $\sigma(y)=\sigma y$ for some $m\geq0,\mu\leq 0$, and $\sigma\geq 0$. The Vasi\v cek drift and volatility in the HJM equation are given by \autoref{equ:va:hjm_drift_vola} with $\beta_y= \beta$ for some constant $\beta<0$ and $a_y=y$, i.e.,
\begin{equation*}
\mu^{\mathrm{HJM}}_y(\tau)
=-\frac{y}{\beta}e^{\beta\tau}\left(1-e^{\beta\tau}\right),
\qquad
\sigma^{\mathrm{HJM}}_y(\tau)
=\sqrt{y}e^{\beta\tau}, 
\qquad\text{for all } \tau\in\mathbb R_+.
\end{equation*}
If $h(0)\in C^1(\mathbb R_+)$, there is a closed-form solution of CRC equation \eqref{equ:va:crc},
\begin{equation*}
h(t,\tau)=h(0,t+\tau)-\int_{0}^{t}\frac{Y(s)}{\beta}e^{\beta(\tau+t-s)}\left(1-e^{\beta(\tau+t-s)}\right)ds
+\int_{0}^{t}\sqrt{Y(s)}e^{\beta(\tau+t-s)}dW(s).
\end{equation*}
Setting $\xi(t)=\int_0^tY(s)e^{2\beta(t-s)}ds$ and $r(t)=h(t,0)$, this can be rewritten as
\begin{align}\label{equ:va:example_h}
h(t,\tau)
&=h(0,t+\tau)+e^{\beta\tau}\left(r(t)-h(0,t)\right)+\frac{1}{\beta}\left(e^{2\beta\tau}-e^{\beta\tau}\right)\xi(t).
\end{align}
Setting $\tau=0$ in \autoref{equ:va:crc} and plugging in \autoref{equ:va:example_h} yields
\begin{align*}
dr(t)&=\left(\mathcal Ah(0,t)-\beta h(0,t)+\beta r(t)+\xi(t)\right)dt+\sqrt{Y(t)}dW(t).
\end{align*}
Summing up, the process $X=(r,\xi,Y)$ is given by the SDE
\begin{equation*}
\left\{
\begin{aligned}
dr(t)&=\left(\frac{\partial h}{\partial\tau}(0,t)-\beta h(0,t)+\beta r(t)+\xi(t)\right)dt+\sqrt{Y(t)}dW(t),
\\
d\xi(t)&=\big(2\beta \xi(t)+Y(t)\big) dt,
\\
dY(t)&=\big(m+\mu Y(t))dt+\sigma \sqrt{Y(t)} d\widetilde W(t),
\end{aligned}
\right.
\end{equation*}
where $h(0)\in C^1(\mathbb R_+)$ is a given initial forward rate curve. It follows that $X$ is an affine factor process for the short rate as described in \autoref{sec:hw} with $d=3$, $\ell=0$, $\lambda=(1,0,0)^\top$. Thus, the CRC model has a finite-dimensional realisation. If $\sigma=0$, the affine bond pricing formula is particularly simple: bond prices are given by
\begin{equation*}
P(t,T)=e^{\int_t^T \left(e^{\beta(s-t)}h(0,t)-h(0,s)\right)ds+\beta^{-1}\left(1-e^{\beta t}\right)r(t)-\frac{1}{2}\beta^{-2}\left(1-e^{\beta  t}\right)^2\xi(t)},
\end{equation*}
where $\xi(t)$ is deterministic and satisfies
\begin{equation}\label{equ:va:example:xi}
\xi(t)=
\left\{\begin{aligned}
&Y(0)\frac{e^{2 \beta  t}-e^{\mu  t}}{2 \beta -\mu}+\frac{m(2 \beta
-\mu -2 \beta e^{\mu  t}+\mu e^{2 \beta  t})}{2 \beta  \mu  (2 \beta
-\mu )},
& \text{if }\mu&< 0,
\\
&Y(0)\frac{e^{2 \beta  t}-1}{2 \beta }+\frac{m \left(e^{2 \beta  t}-2
\beta  t-1\right)}{4 \beta ^2},
&\text{if }\mu&=0,
\end{aligned}\right.
\end{equation}
and $r(t)$ is normally distributed with mean
\begin{equation}\label{equ:va:ex_mean}
e^{\beta t} r(0)+\int_0^t e^{\beta(t-s)} \big(\mathcal Ah(0,s)-\beta
h(0,s)+\xi(s)\big)ds,
\end{equation}
and variance
\begin{equation}\label{equ:va:ex_var}
\frac{Y(0)}{2\beta}\left(e^{2\beta t}-1\right)+\frac{m}{4\beta^2}\left(-2\beta t+e^{2\beta t}-1\right).
\end{equation}

\subsection{Calibration of \texorpdfstring{Vasi\v cek}{Vasicek} CRC models} \label{sec:va:ca}

As outlined in \autoref{sec:crc:calibration}, we first consider $y$ as fixed and suppress the dependence on $y$ in the notation. For any selection of times to maturity $\tau_i,\tau_j$, estimators for $a,\beta$ can be obtained as described in \autoref{sec:hw:estimation} by solving for those $\widehat a,\widehat \beta$ which achieve the best fit in \eqref{equ:hw:covariation}, i.e., 
\begin{equation}\label{equ:va:covariation}\begin{aligned}
\frac{[\widehat r(\cdot,\tau_i),\widehat r(\cdot,\tau_j)](t_n)
-[\widehat r(\cdot,\tau_i),\widehat r(\cdot,\tau_j)](t_{n-M})}
{t_n-t_{n-M}}
&\approx
a\, \frac{e^{\beta\tau_i}-1}{\beta\tau_i} \frac{e^{\beta\tau_j}-1}{\beta\tau_j}
\\&\approx
a\, \frac{\beta\tau_i+\beta^2\tau_i^2/2}{\beta\tau_i}\; \frac{\beta\tau_j+\beta^2\tau_j^2/2}{\beta\tau_j}.
\end{aligned}\end{equation}
Varying the calibration time $t_n$ creates a time series of coefficients $\widehat a(t_n),\widehat \beta(t_n)$ for which we need to specify and calibrate a model. Some models are described in \autoref{sec:num:params}, below.

\section{Consistent recalibration of Cox-Ingersoll-Ross models}\label{sec:cir}

\subsection{Overview}

We give a brief overview of CRC models based on CIR short rates. The overview is sufficient to set the notation for the empirics in \autoref{sec:num}. A detailed description is provided in the online appendix to this paper. For comparison, we briefly digress to the CIR++ model and its CRC version. 

\subsection{Hull-White extended Cox-Ingersoll-Ross models}\label{sec:cir:hw}

We use the setup of \Autoref{sec:hw:setup,sec:crc:setup}, setting $\mathbb X=\mathbb R_+$, $\ell=0$, $\lambda=1$.  We do not specify the parameter space $\mathbb Y$, yet, but we assume that for each $(x,y)\in\mathbb X\times\mathbb Y$, the volatility and drift coefficients are given by $A_y(x)=\alpha_y x$ and $B_y(x)=\beta_y x$ for some $\alpha_y \in (0,\infty)$ and $\beta_y \in (-\infty,0)$. For simplicity, we again choose equidistant grids of times $t_n=n\delta$ and times to maturity $\tau_n=n\delta$, for all $n\in\mathbb N_0$, where $\delta$ is a positive constant. 

The CRC algorithm is similar to the Vasi\v cek model, with the following important differences:

\begin{itemize}
\item The admissibility condition in \autoref{ass:hw:x} is satisfied if and only if $\theta(t)\geq 0$, for all $t\in\mathbb R_+$. This condition can be problematic in practise, as discussed in \autoref{sec:num:negtheta}. Moreover, if this condition is expressed in terms of forward rates $h(t)=\mathcal H_y(\theta(t),X(t))$ instead of Hull-White extensions $\theta(t)$, it becomes apparent that the set of admissible forward rate curves depends on the parameter $y$. This makes it difficult to formulate convergence results similar to those in the Vasi\v cek case.

\item In contrast to the Vasi\v cek model, there does not seem to be a closed-form expression for $\theta=\mathcal C_y(h,x)$ because the Volterra kernel in the integral operator $\mathcal H_y$ is more complicated. Therefore, the Volterra equation is solved by numerical approximation of order two using a discretisation in the time to maturity. 
\end{itemize}

\subsection{CIR++ models in the CRC framework}\label{sec:cir:cirpp}

In the CIR++ model \cite[Section 3.9]{brigo2007}, also known as deterministic shift-extended CIR model, the short rate process is defined by $r(t)=X(t)+\theta(t)$, where $X$ is a CIR process and $\theta$ is a deterministic function of time. Note that this is a different time-inhomogeneity than the one described in \autoref{sec:cir:hw}. In particular, the factor process $X$ is time-homogeneous and does not coincide with the short rate. The HJM equation of the CIR++ model is
\begin{equation}\label{equ:cir:hjmpp}\begin{aligned}
dh(t)&=\left(\mathcal Ah(t)+\mu_{y}^{\mathrm{HJM}}\big(X(t)\big)\right)dt+\sigma_{y}^{\mathrm{HJM}}\big(X(t)\big)dW(t), 
\\
dX(t)&=\left(b_y+\beta_y X(t)\right)dt + \sqrt{\alpha_y X(t)}dW(t),
\end{aligned}\end{equation}
where $\mu_{y}^{\mathrm{HJM}}$ and $\sigma_{y}^{\mathrm{HJM}}$ are the same as in the CIR case. In the CRC extension of this model, $y$ is replaced by a stochastic process $Y$. This model has both advantages and disadvantages over the consistently recalibrated CIR model:
\begin{itemize}
\item The SDE for $X$ does not depend on $h$. Therefore, existence and uniqueness of $X$ can be shown by standard methods. Then a mild solution $h$ can be constructed by stochastic convolution \cite[Section 6.1]{daPrato2014se}:
\begin{align*}
h(t) = \mathcal S(t)h(0) + \int_0^t \mathcal S_{t-s} \mu_{Y(s)}^{\mathrm{HJM}}\big(X(s)\big)ds
+\int_0^t \mathcal S_{t-s} \sigma_{Y(s)}^{\mathrm{HJM}}\big(X(s)\big) dW(s).
\end{align*}

\item The function $\theta$ is allowed to assume negative values and can be calibrated to a given yield curve without having to invert a Volterra integral operator. However, this calibration requires one to know the current value of $X$. This can be seen from the equation for forward rate curves
\[
h(t)=\mathcal S(t)\theta-b_y\Psi_y-\Psi'_yX(t),
\] 
where $\Psi_y$ is the same as in the CIR case (c.f. \autoref{sec:hw:riccati}). In contrast to the CIR model, the process $X$ is not directly observable from the short end of the term structure. Therefore, $X(t)$ and the parameters $\alpha_y$ and $\beta_y$ have to be estimated jointly from realised covariations of yields as described in \autoref{sec:hw:calibration}. Moreover, in contrast to the CIR model the parameter $b_y$ is not redundant and has to be estimated using the same methodology as for general multi-factor models (see \autoref{sec:hw:calibration}).
\end{itemize}

\section{Empirical results}\label{sec:num}

\subsection{Overview}

CRC models based on Vasi\v cek and CIR short rates are calibrated to Euro area yield curves. Properties of the calibrated models are studied in comparison to market data and models without consistent recalibration. 
Our empirical findings show that the assumption of constant parameters in the Vasi\v cek and CIR models is too restrictive. Therefore, the additional flexibility provided by CRC models is a useful tool to better capture the market dynamics. This is also reflected in better fits of the covariance matrix of yields. 

% All figures are shown at the end of the section.

\subsection{Description of the data}

We consider the zero-coupon yield curves released by the European Central Bank (ECB) on a daily basis. The yields are estimated from AAA-rated (Fitch Ratings) Euro area central government bonds being actively traded on the market. Estimation is done by the ECB using the Svensson family of curves, see \cite{svensson1994estimating,ecbNotes}. Data is available from September 6, 2004, and we set April 1, 2014 to be the last observation date. In total, this results in 2454 observed yield curves with times to maturity ranging from 3 months up to 30 years. Yields are continuously compounded (c.f.~\autoref{equ:hw:yields}) and denoted by $\widehat r(t,\tau)$, with $\tau$ being the time to maturity. A selection of yields is shown in \Autoref{fig:num_yields,fig:num_yield_curves}. The short rate is approximated by the yield with the lowest time to maturity (3 months) and is depicted in \autoref{fig:num_yields}.

\begin{figure}
\centering 
\includegraphics[scale=0.4]{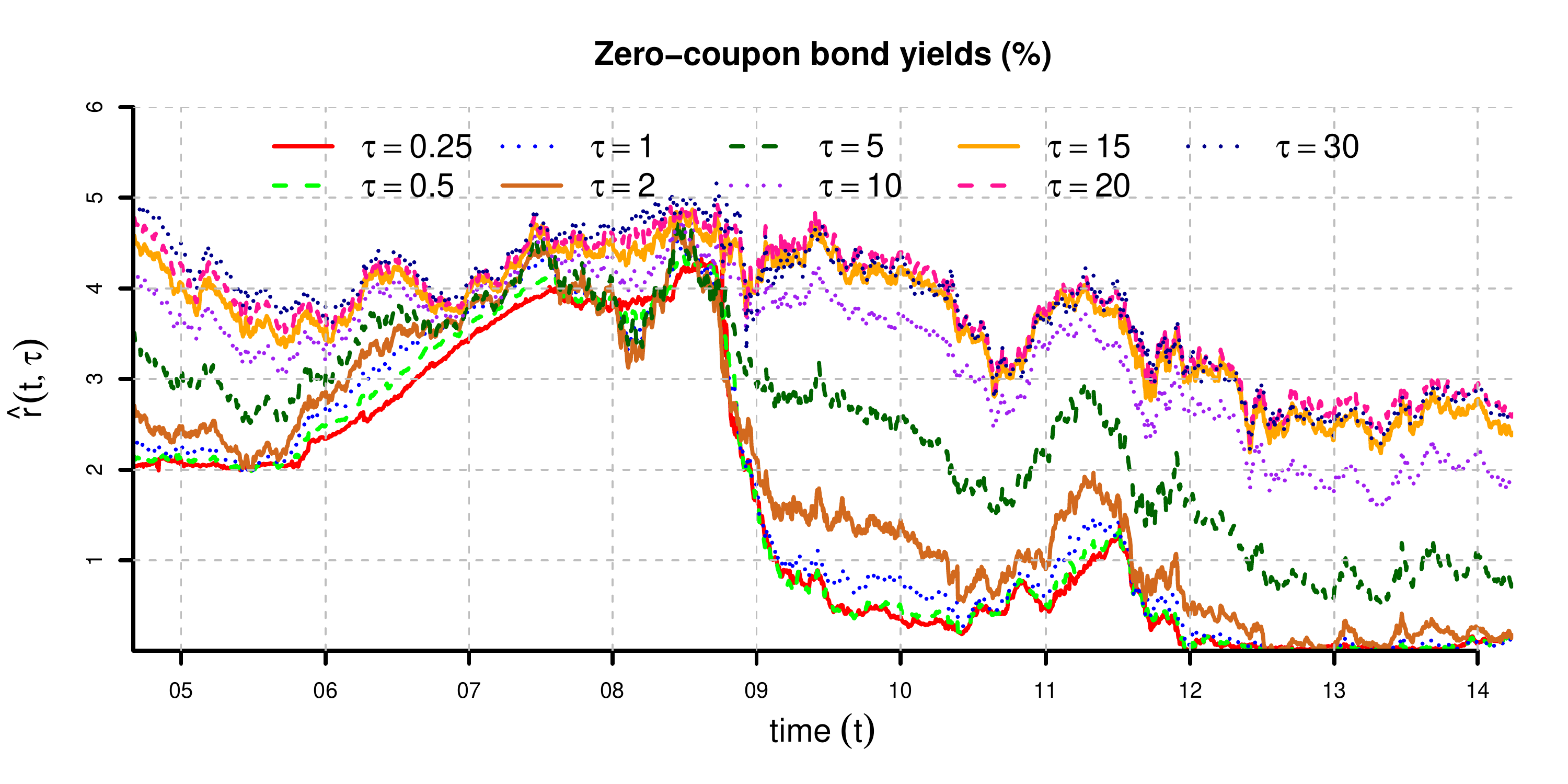} 
\caption{Historical zero-coupon yields estimated by the ECB for various times to maturity from 06/09/2004 to 01/04/2014. We use the 3-month yields ($\tau=0.25$) as a proxy for the short rate.}
\label{fig:num_yields}
\end{figure}

\begin{figure}
\centering 
\includegraphics[scale=0.4]{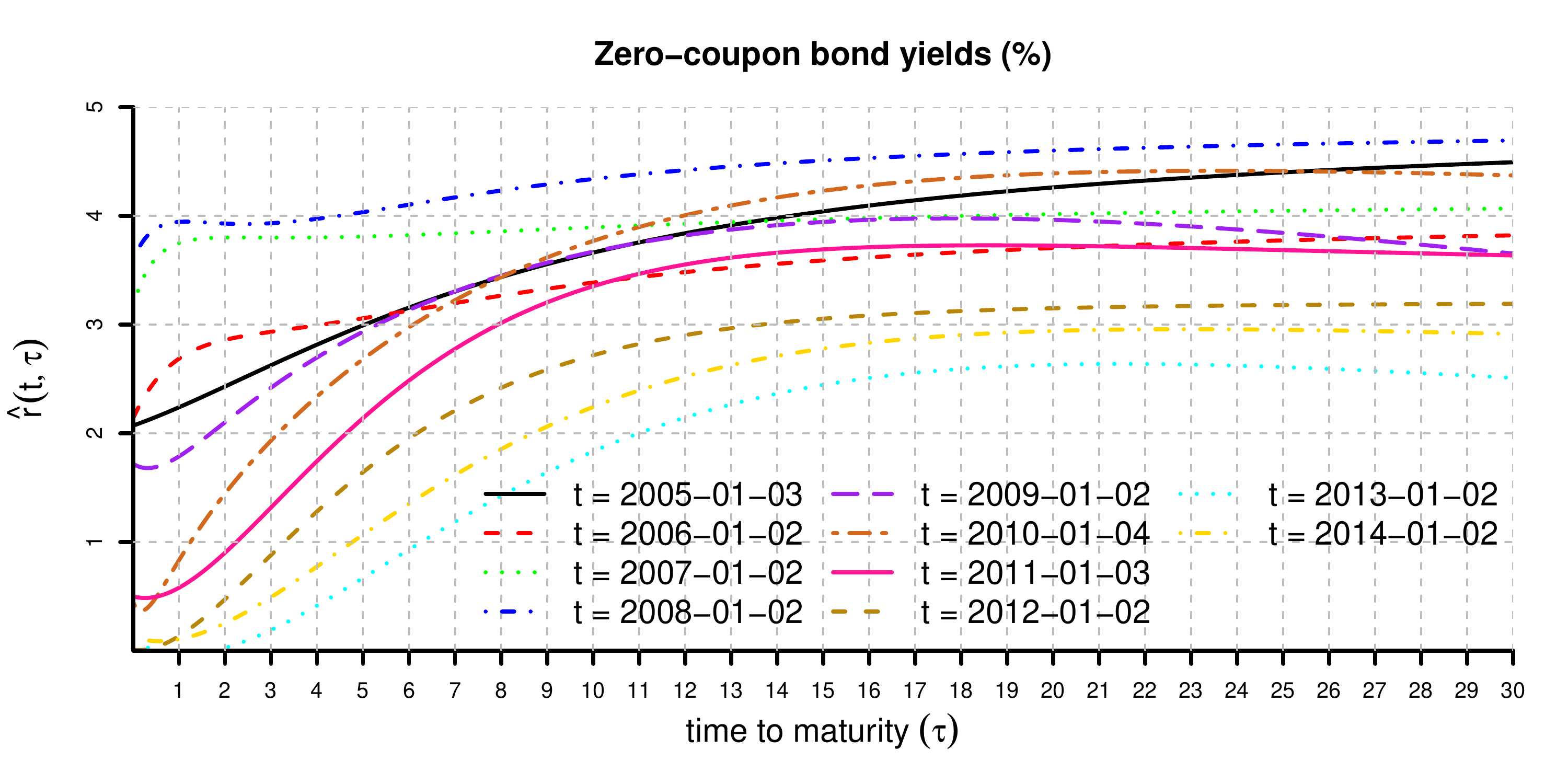} 
\caption{Zero-coupon yield curves estimated by the ECB for various observation dates.}
\label{fig:num_yield_curves}
\end{figure}

\subsection{Calibration of CRC models}\label{sec:num:ca}

The CRC models based on Vasi\v cek and CIR short rates are calibrated as described in \Autoref{sec:va:ca} and the online appendix. Time steps $\delta=240^{-1}$ of one business day and time windows of $M=100$ business days are used for the calibration. The choice $M=100$ is a  compromise between accuracy and over-smoothing and gives reasonable results over the time horizon of roughly a decade (see Figures~\ref{fig:num_sigma_vasicek} and \ref{fig:num_sigma_cir}). For $\tau_1\ll1$ one immediately obtains from  \autoref{equ:va:covariation} and its CIR counterpart the estimator
\begin{align}\label{equ:va:a}
\hat{a}(t)=\frac{[\widehat r(\cdot,\tau_1),\widehat r(\cdot,\tau_1)](t)-[\widehat r(\cdot,\tau_1),\widehat r(\cdot,\tau_1)](t-\delta M)}{\delta M},
\end{align}
in the Vasi\v cek case, and 
\begin{align}\label{equ:cir:alpha}
\hat{\alpha}(t)=\frac{[\widehat r(\cdot,\tau_1),\widehat r(\cdot,\tau_1)](t)-[\widehat r(\cdot,\tau_1),\widehat r(\cdot,\tau_1)](t-\delta M)}{\delta\sum_{m=0}^{M-1}\hat{r}(t-\delta m,\tau_1)},
\end{align}
in the CIR case, where the quadratic variation is estimated by \eqref{equ:va:covariation}. On the other hand, taking $\tau_2\gg1$, one can solve \eqref{equ:va:covariation} and its CIR counterpart for $\beta$ and obtain the estimator 
\begin{align}\label{equ:va:beta}
\hat{\beta}(t)=-\frac{1}{\tau_2}\left(\frac{\delta M\hat{a}(t)}{[\widehat r(\cdot,\tau_2),\widehat r(\cdot,\tau_2)](t)-[\widehat r(\cdot,\tau_2),\widehat r(\cdot,\tau_2)](t-\delta M)}\right)^{\frac{1}{2}},
\end{align}
in the Vasicek case, and 
\begin{equation}\label{equ:cir:beta}
\begin{split}
\hat{\beta}(t)&=\frac{\sqrt{\hat{\alpha}(t)}}{2}\tau_2\left(\frac{[\widehat r(\cdot,\tau_2),\widehat r(\cdot,\tau_2)](t)-[\widehat r(\cdot,\tau_2),\widehat r(\cdot,\tau_2)](t-\delta M)}{\delta\sum_{m=0}^{M-1}\hat{r}(t-\delta m,\tau_1)}\right)^{\frac{1}{2}}
\\
&\quad-\frac{\sqrt{\hat{\alpha}(t)}}{\tau_2}\left(\frac{[\widehat r(\cdot,\tau_2),\widehat r(\cdot,\tau_2)](t)-[\widehat r(\cdot,\tau_2),\widehat r(\cdot,\tau_2)](t-\delta M)}{\delta \sum_{m=0}^{M-1}\hat{r}(t-\delta m,\tau_1)}\right)^{-\frac{1}{2}},
\end{split}
\end{equation}
in the CIR case. The resulting trajectories of the estimated Vasi\v cek volatility $\sqrt{a}$ and CIR volatility $\sqrt{\alpha}$ are shown in \Autoref{fig:num_sigma_vasicek, fig:num_sigma_cir}. 
The trajectories of the speed of mean reversion $-\beta$ for both models are plotted in \Autoref{fig:num_beta_vasicek,fig:num_beta_cir}. 

\begin{figure}  
\centering 
\includegraphics[scale=0.4]{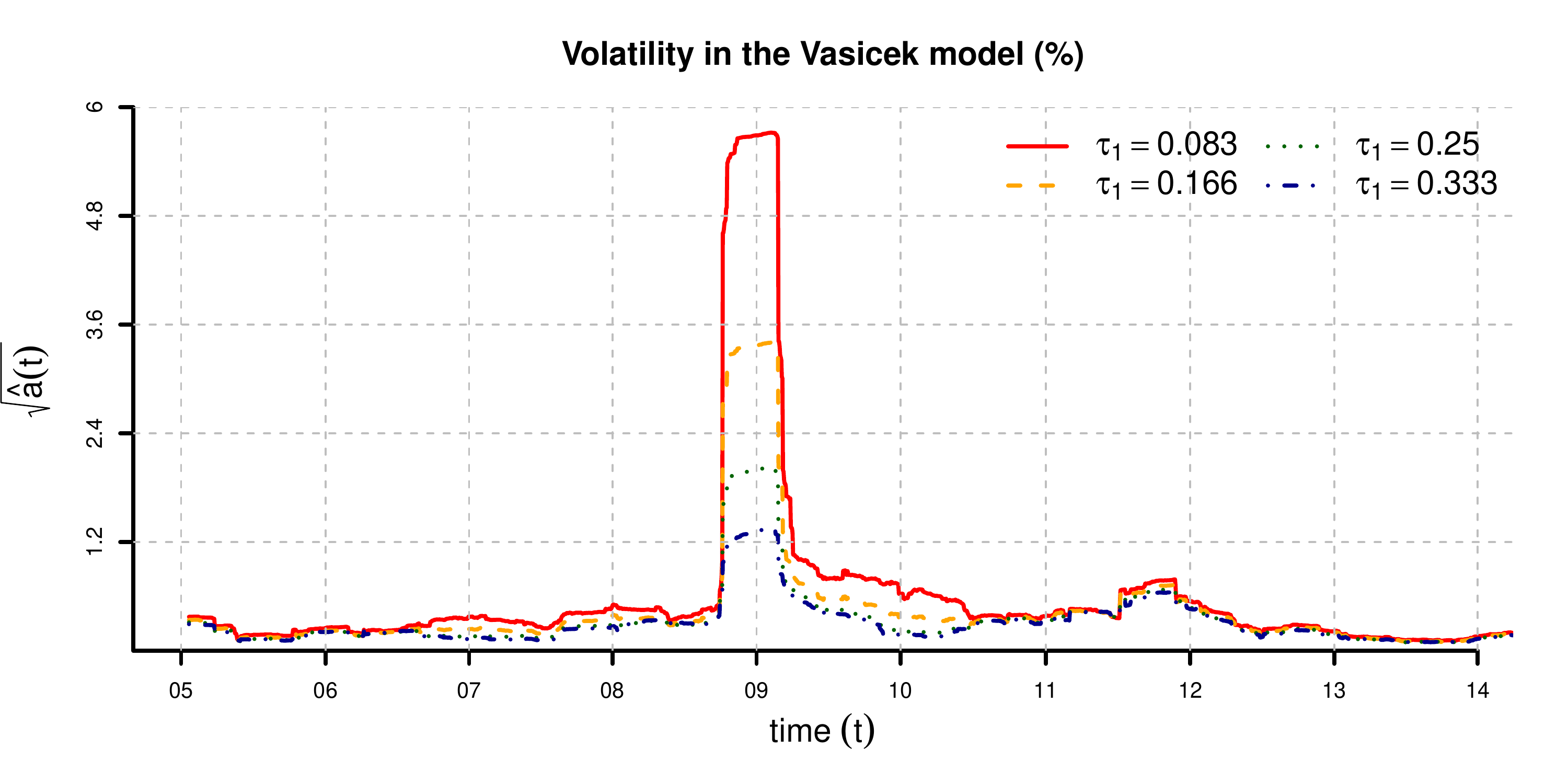} 
\caption{Volatility parameter of the Vasi\v cek model estimated by \eqref{equ:va:a} using a time window of $M=100$ yields with time to maturity $\tau_1$.}
\label{fig:num_sigma_vasicek}
\end{figure}

\begin{figure}
\centering 
\includegraphics[scale=0.4]{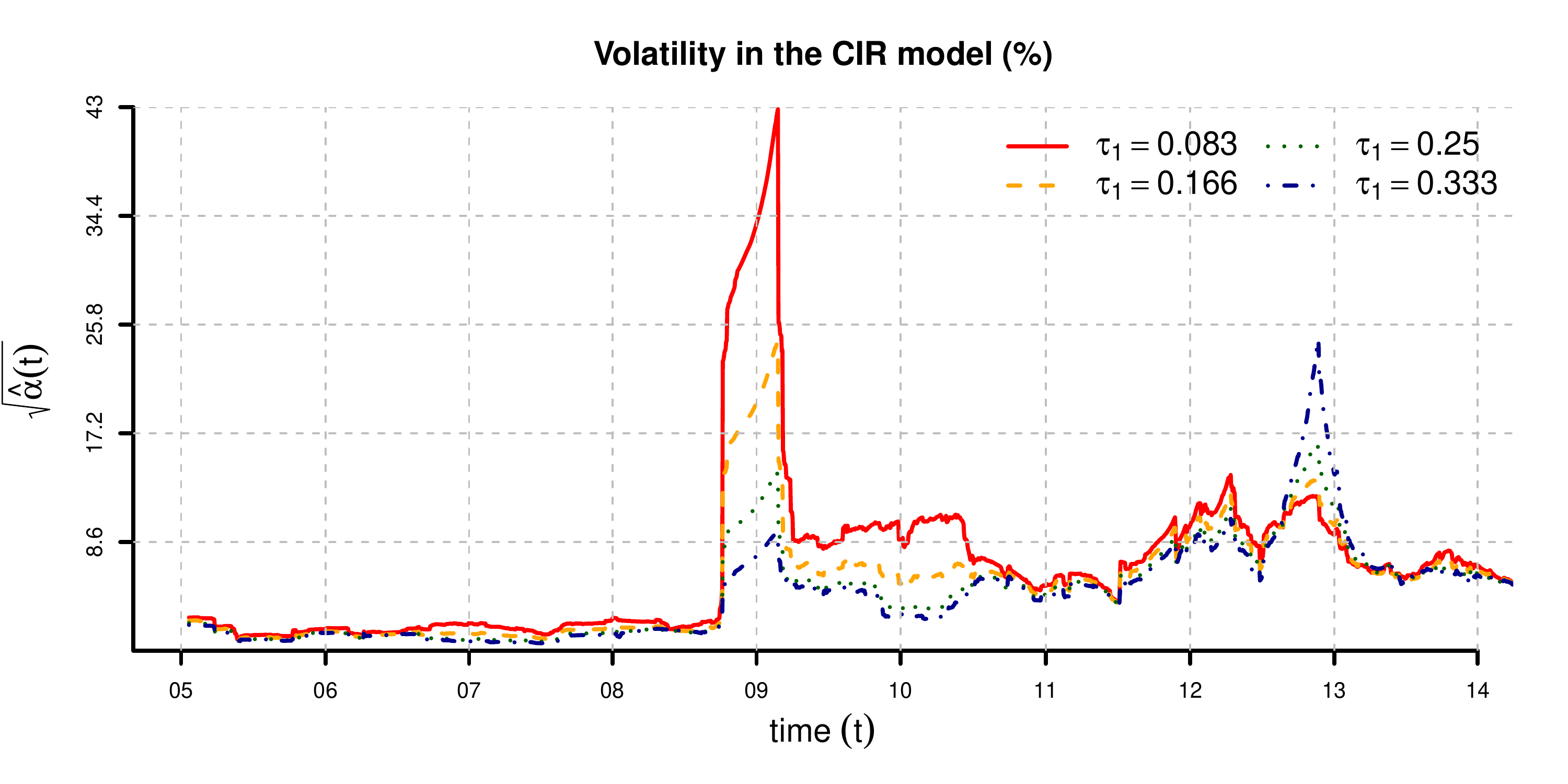} 
\caption{Volatility parameter of the CIR model estimated by \eqref{equ:cir:alpha} using a time window of $M=100$ yields with time to maturity $\tau_1$.}
\label{fig:num_sigma_cir}
\end{figure}

\begin{figure}
\centering 
\includegraphics[scale=0.4]{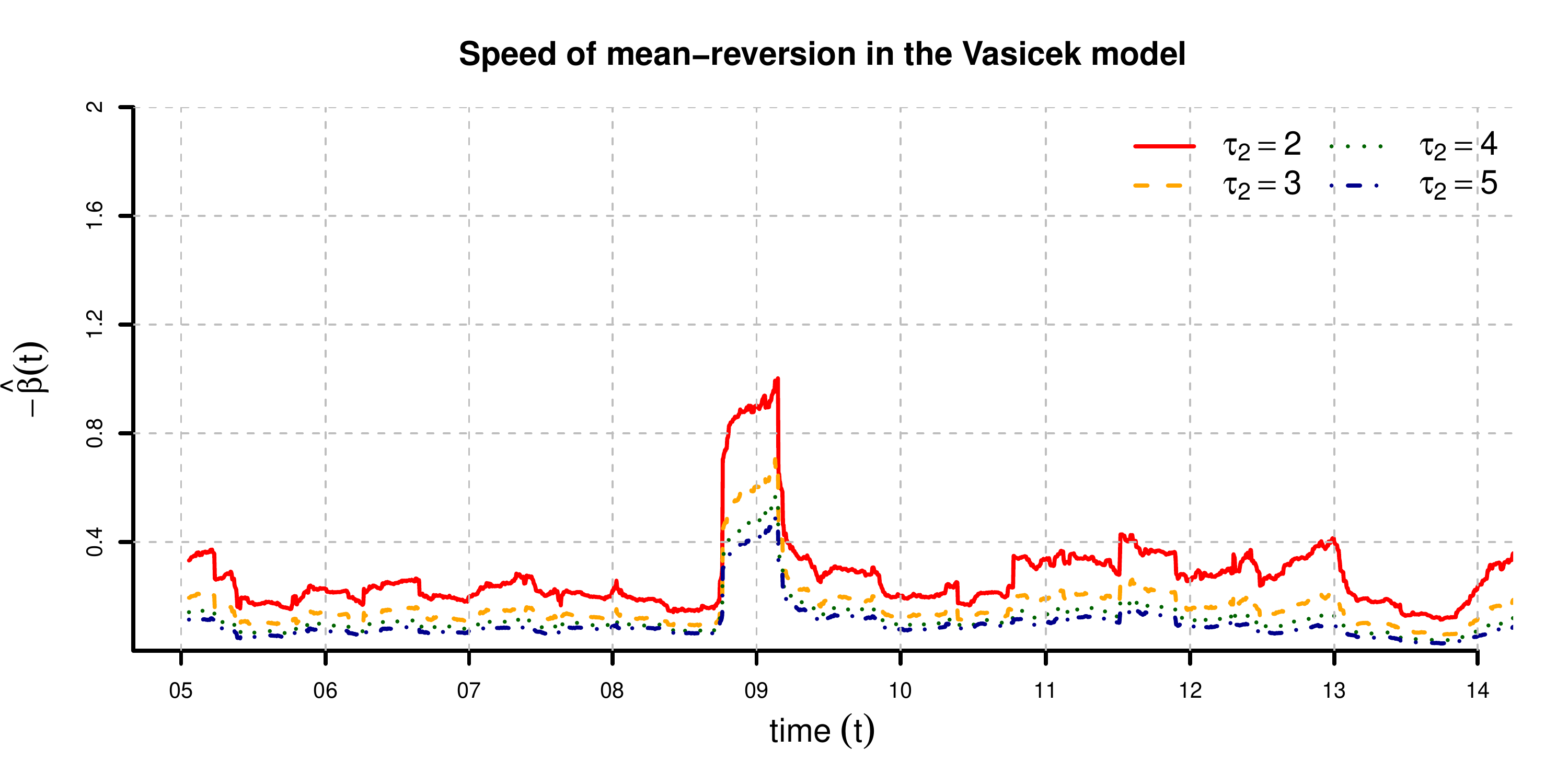} 
\caption{Speed of mean-reversion parameter of the Vasi\v cek model estimated by \eqref{equ:va:beta} using a time window of $M=100$ yields with times to maturity $\tau_1=0.25$ and various values of $\tau_2$.}
\label{fig:num_beta_vasicek}
\end{figure}

\begin{figure}
\centering 
\includegraphics[scale=0.4]{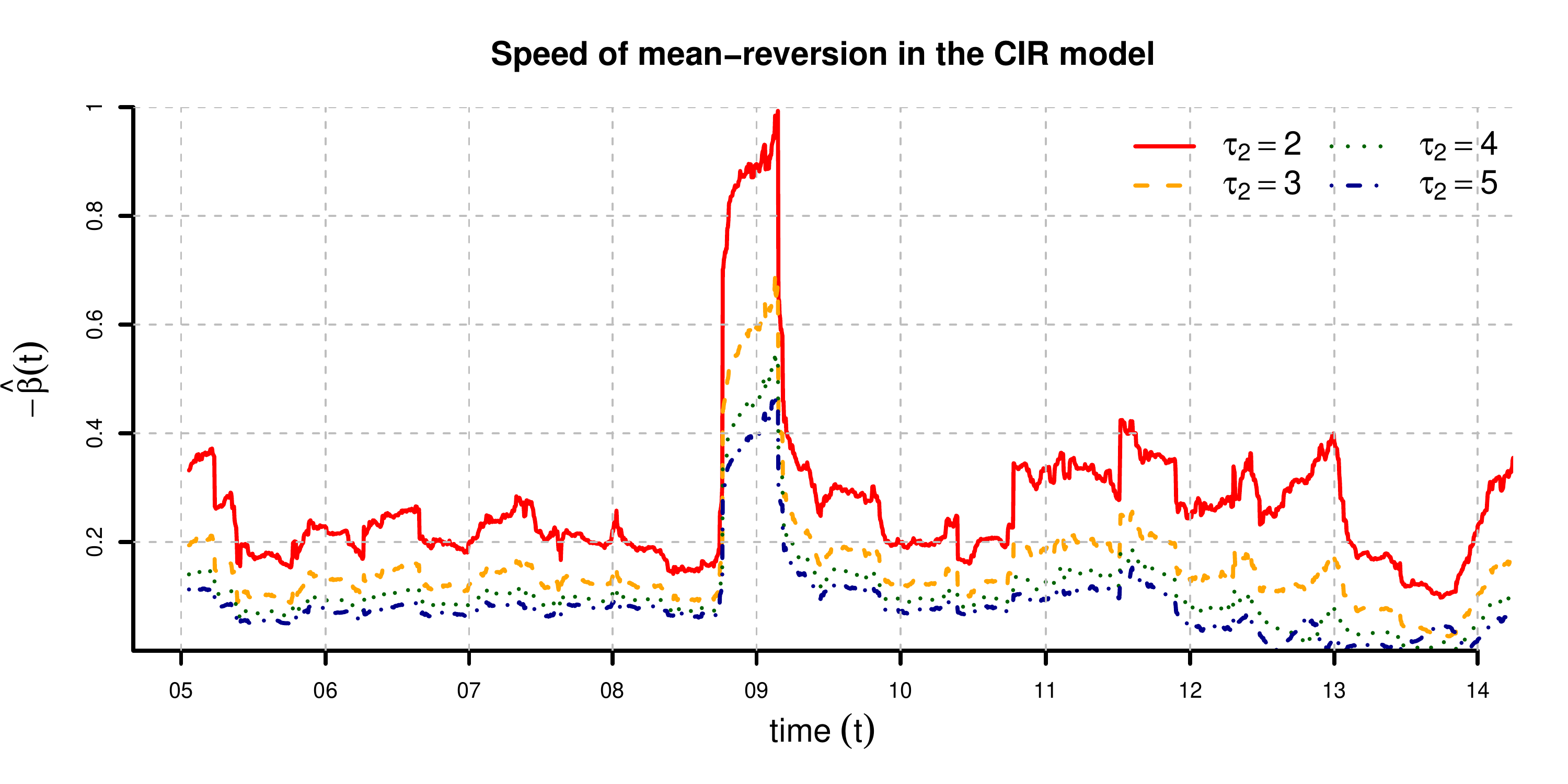} 
\caption{Speed of mean-reversion parameter of the CIR model estimated by \eqref{equ:cir:beta} using a time window of $M=100$ yields with times to maturity $\tau_1=0.25$ and various values of $\tau_2$.}
\label{fig:num_beta_cir}
\end{figure}

It turns out that for most parts of the data, $a$ and $\alpha$ do not depend too much on the times to maturity used in the estimation, whereas $\beta$ does. Typically, smaller times to maturity result in larger values of $-\beta$, as shown in \Autoref{fig:num_beta_vasicek,fig:num_beta_cir}. This means that one-factor Vasi\v cek and CIR models are not flexible enough to reproduce the ECB yield curve movements in full accuracy, and so the choice of times to maturity used in the calibration procedure may have an impact on the results. In particular we set $\tau_1=0.25$ and $\tau_2=2$ (i.e.\ 3 months and 2 years). 

The dependence of the estimated parameter $\beta$ on the choice of times to maturity suggests to use multi-factor models as building blocks for CRC models. In the empirical part of this paper, we aim for a detailed understanding of the one-factor case and leave the extension to multiple factors for future research. As pointed out, accurate modelling of long-term rates might require one to enlarge the model class even further to allow the forward rate volatilities to decay slower than exponentially.

By the theory of Hull-White extensions, an exact match to the initial yield curve is always achieved, but the corresponding time-homogeneous models often do not match the initial yield curve well (\Autoref{fig:num_hwe3,fig:num_hwe4}). This is not surprising as they are calibrated to the yield curve \emph{dynamics} and not to the initial yield curve. This separation of dynamics and initial calibration is actually one of the strengths of our approach. 

\begin{figure}
\centering 
\includegraphics[scale=0.4]{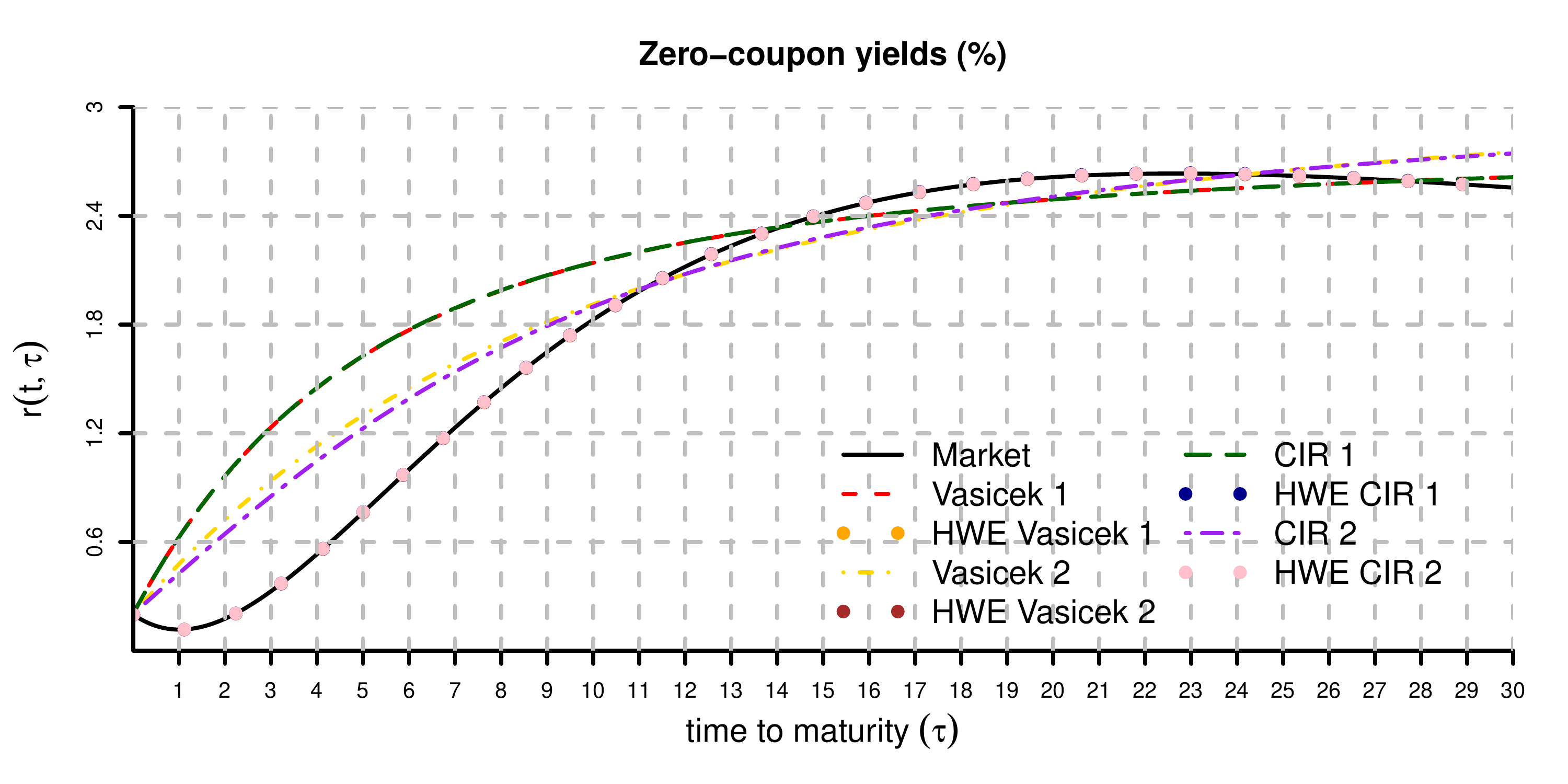} 
\caption{Calibrations of some homogeneous and Hull-White extended models as of 1 April 2014. Vasicek 1 and CIR 1 are homogeneous models calibrated to the yield curve dynamics using \eqref{equ:va:a}-\eqref{equ:cir:beta} with $\tau_1=0.25$ and $\tau_2=2$. Vasi\v cek 2 and CIR 2 are homogeneous models calibrated to the prevailing yield curve by least squares. The Hull-White extended models match the initial yield curve exactly.}
\label{fig:num_hwe3}
\end{figure}

\begin{figure}
\centering 
\includegraphics[scale=0.4]{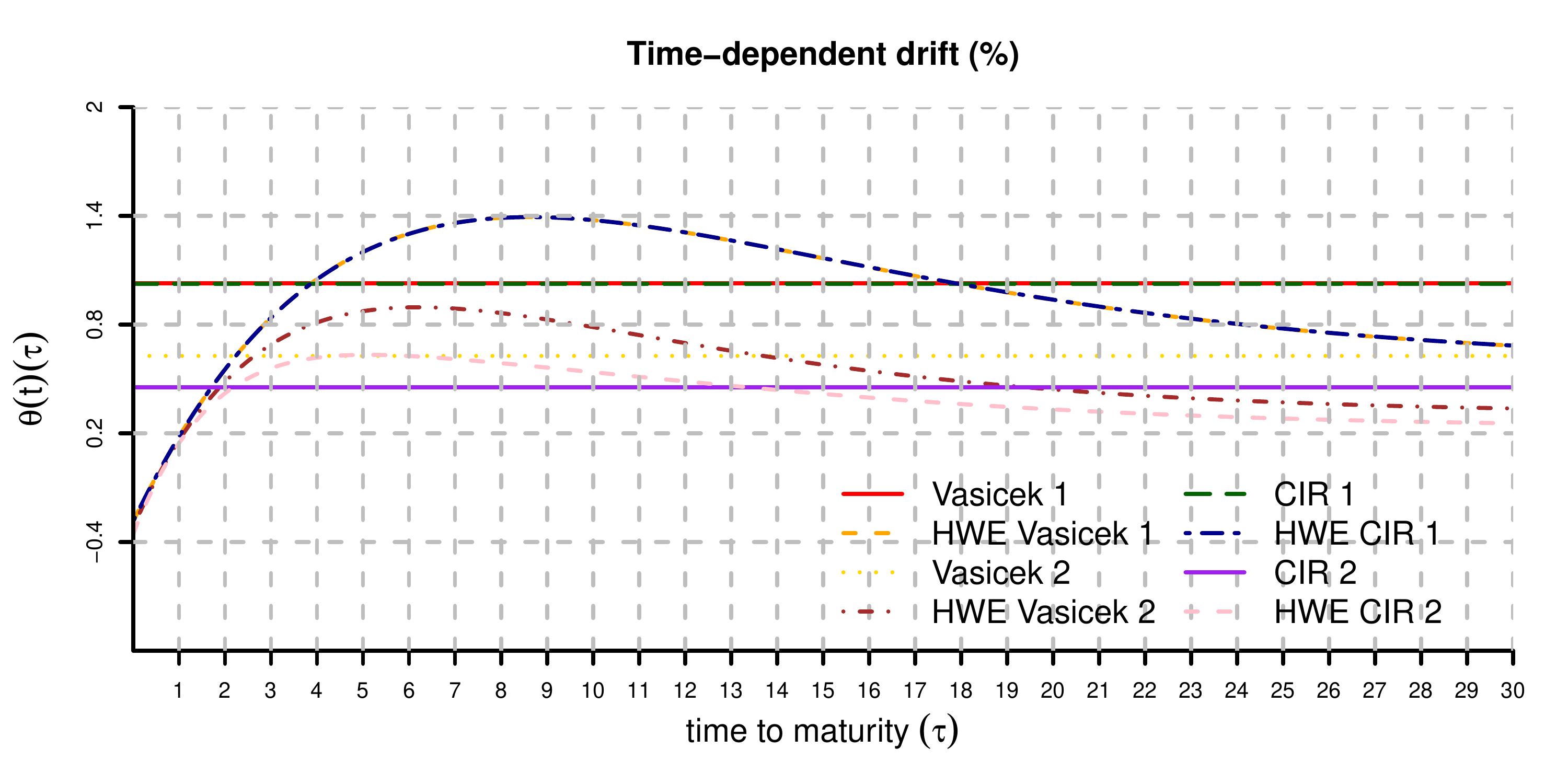} 
\caption{The Hull-White extensions $\theta(t)$ corresponding to the models of Figure \ref{fig:num_hwe3}. Note that $\theta(t)$ assumes negative values, which is typical in situations where the yield curve is inverted at the short end.}
\label{fig:num_hwe4}
\end{figure}

To model the dynamics of the Vasi\v cek coefficients $a,\beta$ and the CIR coefficients $\alpha,\beta$, we use geometric Brownian motions and/or CIR processes, as laid out in \autoref{sec:num:params}. Note that the assumption of constant coefficients, which is implicit in affine factor models without the CRC extension, is not realistic over long time horizons in view of \Autoref{fig:num_sigma_vasicek, fig:num_sigma_cir}. 

\subsection{Negative levels of mean reversion}\label{sec:num:negtheta}

A problematic aspect is that the time-dependent drift $\theta$ obtained by the calibration to the initial yield curve can assume negative values, which are not admissible in the CIR model. The problem occurs mostly in low interest rate scenarios with partially inverted yield curves (\Autoref{fig:num_hwe3,fig:num_hwe4,fig:num_hwe5}). In contrast, negative values of $\theta$ are allowed in the Vasi\v cek  model and might even be desirable for modelling bond markets with negative interest rates. 

\begin{figure}%[H]
\includegraphics[scale=0.4]{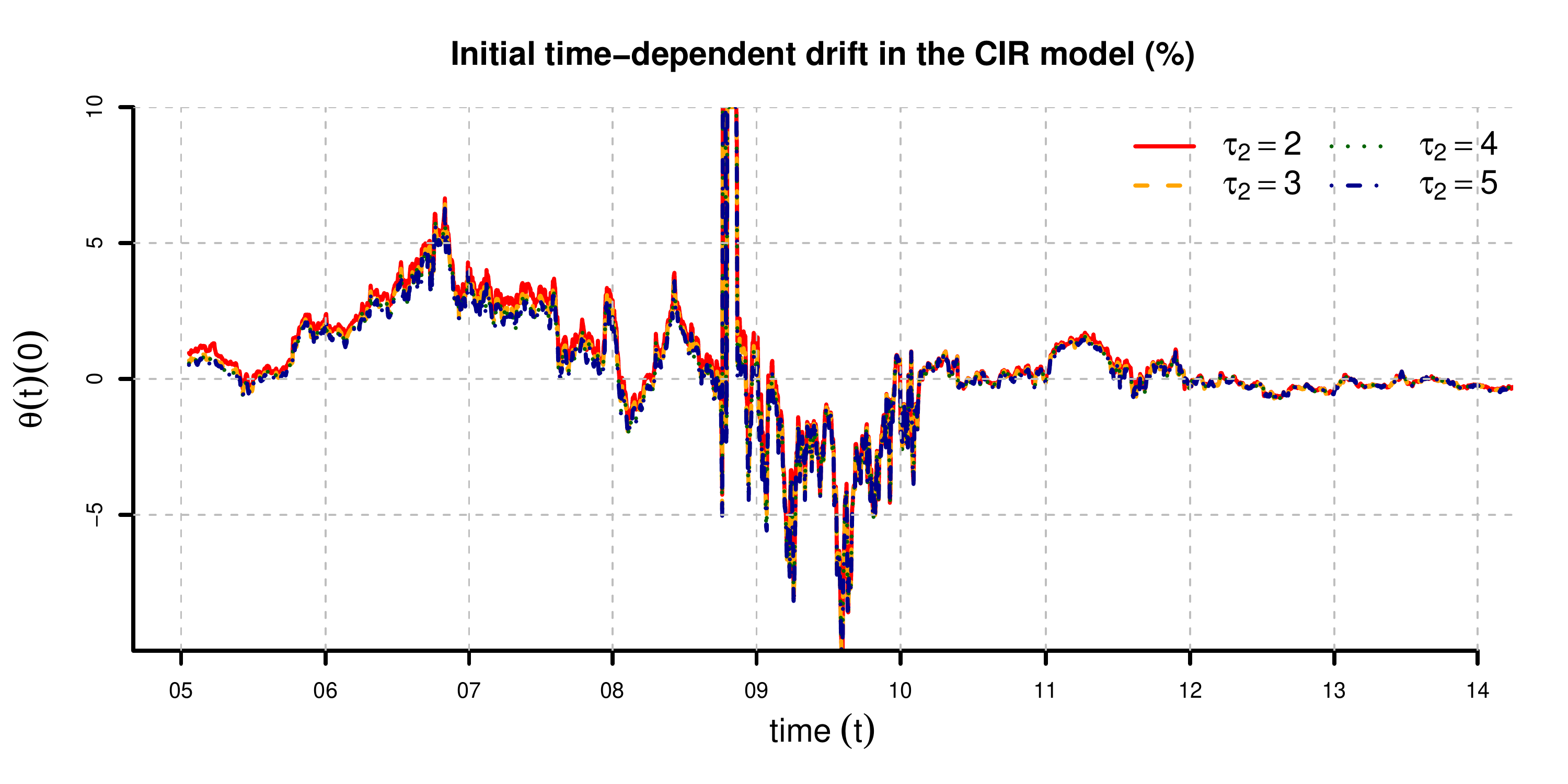} 
\caption{The historical values of $\theta(t)(0)$ in the CIR model calculated using the estimates of $\beta$ in Figure \ref{fig:num_beta_cir}. Negative values occur frequently in 2009 and 2012 to 2014. They are problematic for reasons laid out in \autoref{sec:num:negtheta}.}
\label{fig:num_hwe5}
\end{figure}

As only the short (left) end of $\theta$ is relevant for CRC models, at each step of the simulation scheme (cf. \Autoref{alg:va:crc}), it is sufficient to understand how $\theta(0)$ depends on the prevailing forward rate curve and the coefficients of the affine factor process. The general formula for $\theta(0)$ is
\[
\theta(0)=\mathcal C_{y}(h,x)(0)=\frac{1}{\langle\lambda,e_1\rangle}\left(h'(0)- F_y'(0)\cdot\lambda-\left\langle R_y'(0)\cdot\lambda,x\right\rangle\right),
\]
which follows by differentiating the relation $h=\mathcal H_y(\theta,x)$ with respect to the time to maturity $\tau$ and evaluating at $\tau=0$. 
In the Vasi\v cek and CIR case, the formula becomes $\theta(0)=h'(0)- \beta_y h(0)$. This shows that the problem can be alleviated to some extent by artificially choosing higher levels of $-\beta$, resulting in higher values of $\theta(0)$. For this reason we set $\tau_2=2$ instead of higher values of $\tau_2$ in the calibration of $\beta$ (\Autoref{fig:num_beta_vasicek,fig:num_beta_cir}). Despite this correction, $\theta(0)$ remains negative for most parts of 2009 and 2012--2014. During these periods of time, the CIR model cannot be fitted simultaneously to the yield curves as well as their volatilities and has to be rejected. 

\subsection{Models for parameter evolutions}\label{sec:num:params}

There are very few restrictions on the choice of parameter process. It can be chosen exogenously for scenario based simulation or calibrated to the market, and it is not restricted by the fit to the initial term structure. 

We consider here four models for the evolution of the Vasi\v cek and CIR parameters: a reference model where the parameters are constant, two toy models with constant mean reversion and time-varying volatility, and one fully stochastic model which is calibrated to the market. 
In the Vasi\v cek case, the four models are: 
\begin{enumerate}[(V1)]
\item \label{item:num:va1} 
a Hull-White extended Vasi\v cek model with constant coefficients $\beta_y=\beta_0$ and $a_y=a_0$;
\item \label{item:num:va2} 
a Vasi\v cek CRC model with constant mean reversion coefficient $\beta_y=\beta_0$ and deterministically increasing volatility given by $a_y=a_0y$, $Y(t)=1+3t$;
\item \label{item:num:va3} 
a Vasi\v cek CRC model with constant mean reversion coefficient $\beta_y=\beta_0$ and stochastically increasing volatility given by $a_y=y$, $dY(t)=\left(4a_0-Y(t)\right)dt+\sigma\sqrt{Y(t)}d\widetilde W(t)$, $Y(0)=a_0$, $\sigma=3\cdot10^{-3}$; and
\item \label{item:num:va4} 
a Vasi\v cek CRC model with stochastic coefficients given by geometric Brownian motion: $\beta_y=y_1$, $a_y=y_2$, $dY_1=\mu_1Y_1(t) + \sigma_1Y_1(t)d\widetilde W_1(t)$, $dY_2=\mu_2Y_2(t) + \sigma_2Y_2(t)d\widetilde W_2(t)$. The coefficients $\mu_{1,2}$ and $\sigma_{1,2}$ are deterministic and calibrated to $M$ observations as described in \autoref{sec:num:ca}.
\end{enumerate}

Note that models (V\ref{item:num:va2}) (V\ref{item:num:va3}) both describe scenarios where the standard deviation of the noise in the short rate process doubles over the period of a year. Indeed, in (V\ref{item:num:va2}) the volatility coefficient $a$ increases to four times its initial value, and in (V\ref{item:num:va3}) the level of mean reversion of $a$ increases to four times its initial value. 

Models (V\ref{item:num:va2}) and (V\ref{item:num:va3}) are special cases of \autoref{sec:va:example}. In (V\ref{item:num:va2}), which corresponds to $m=3a_0$, $\mu=0$, and $\sigma=0$, there is an explicit formula for the moment generating function of the short rate process. By equations \eqref{equ:va:ex_mean} and \eqref{equ:va:ex_var}, it is given by
\begin{equation}\label{equ:num_mgf}
\mathbb E\left[e^{\eta r(t)}\right]=e^{e^{\beta_0t}\eta r_0+\eta\int_0^t e^{\beta_0(t-s)} \big(h'_0(s)-\beta h_0(s)+\xi(s)\big)ds+\frac{\eta^2}{2}\xi(t)},
\qquad \eta \in \mathbb R, t\in\mathbb R_+,
\end{equation}
where 
\[
\xi(t)=a_0\frac{e^{2 \beta_0  t}-1}{2 \beta_0 }+\frac{3a_0 \left(e^{2 \beta_0  t}-2\beta_0  t-1\right)}{4 \beta_0 ^2}.
\]
Model (V\ref{item:num:va3}) corresponds to \autoref{sec:va:example} with $m=4a_0$, $\mu=-1$ and $\sigma=3\cdot10^{-3}$. 

The semigroup approach of \autoref{thm:va:convergence} implies convergence of the simulation scheme for (V\ref{item:num:va2}) and for (V\ref{item:num:va4}) with $Y_1$ replaced by $Y_1+\epsilon$ for some $\epsilon>0$. In our numerical simulations, we observe convergence for all models (see \autoref{sec:num:con}), including the following CIR counterparts of the models just described: 
\begin{enumerate}[({CIR}1)]
\item \label{item:num:cir1}
a Hull-White extended Cox-Ingersoll-Ross model with constant coefficients $\beta_y=\beta_0$ and $\alpha_y=\alpha_0$;
\item \label{item:num:cir2}
a Cox-Ingersoll-Ross CRC model with constant mean reversion coefficient $\beta_y=\beta_0$ and deterministically increasing volatility given by $\alpha_y=\alpha_0y$, $Y(t)=1+3t$;
\item \label{item:num:cir3} 
a Cox-Ingersoll-Ross CRC model with constant mean reversion coefficient $\beta_y=\beta_0$ and stochastically increasing volatility given by $\alpha_y=y$, $dY(t)=\left(4\alpha_0-Y(t)\right)dt+\sigma\sqrt{Y(t)}d\widetilde W(t)$, $Y(0)=\alpha_0$, $\sigma=5\cdot10^{-2}$; and
\item \label{item:num:cir4}
a Cox-Ingersoll-Ross CRC model with stochastic coefficients given by geometric Brownian motion: $\beta_y=y_1$, $\alpha_y=y_2$, $dY_1=\mu_1Y_1(t) + \sigma_1Y_1(t)d\widetilde W_1(t)$, $dY_2=\mu_2Y_2(t) + \sigma_2Y_2(t)d\widetilde W_2(t)$. The coefficients $\mu_{1,2}$ and $\sigma_{1,2}$ are deterministic and calibrated to $M$ observations as described in \autoref{sec:num:ca}.
\end{enumerate}

\subsection{Implementation}\label{sec:num:imp}

Simulating CRC models requires iterative sampling of the underlying affine short rate process and recalibrating Hull-White extensions. This is explained in detail in \Autoref{alg:va:crc} for the Vasi\v cek model and in the online appendix for the CIR model. The algorithms can be parallelised on a path-by-path level. Parallelisation on lower levels does not pay off because the individual time steps are dependent on each other. In our implementation in R, generating $10^5$ paths with 240 time steps on a cluster of 48 times 2.2GHz processors takes around 10 minutes in the Vasi\v cek case and 20 minutes in the CIR case. 

\subsection{Convergence analysis}\label{sec:num:con}

\autoref{thm:va:convergence} predicts first order convergence of the simulation scheme under suitable assumptions on the model. The objective of this section is to demonstrate this convergence in numerical examples for the models described in \autoref{sec:num:params}. 

The simulation is started with the initial forward rate curve $h_0$ of September 2, 2013. The parameters $\beta_0,a_0,\alpha_0,\theta_0$ are calibrated as in \autoref{sec:num:ca} with a time window of $M=100$ observations. Then the moment generating function of the short rate $r(1)$ after one year is calculated by Monte Carlo simulation with $10^5$ sample paths. In the model (V\ref{item:num:va2}), the exact value of the moment generating function is known and given by \autoref{equ:num_mgf}. In the other models, a reference value is calculated by extrapolation from the Monte Carlo estimates. The resulting errors are shown in \Autoref{fig:num_convergence1,fig:num_convergence2,fig:num_convergence3}. As expected from \Autoref{thm:va:convergence} we observe first order convergence for models (V\ref{item:num:va2}) and (V\ref{item:num:va4}). The errors in \Autoref{fig:num_convergence3} indicate convergence also in the CIR counterparts. 

\begin{figure}
\centering 
\includegraphics[scale=0.4]{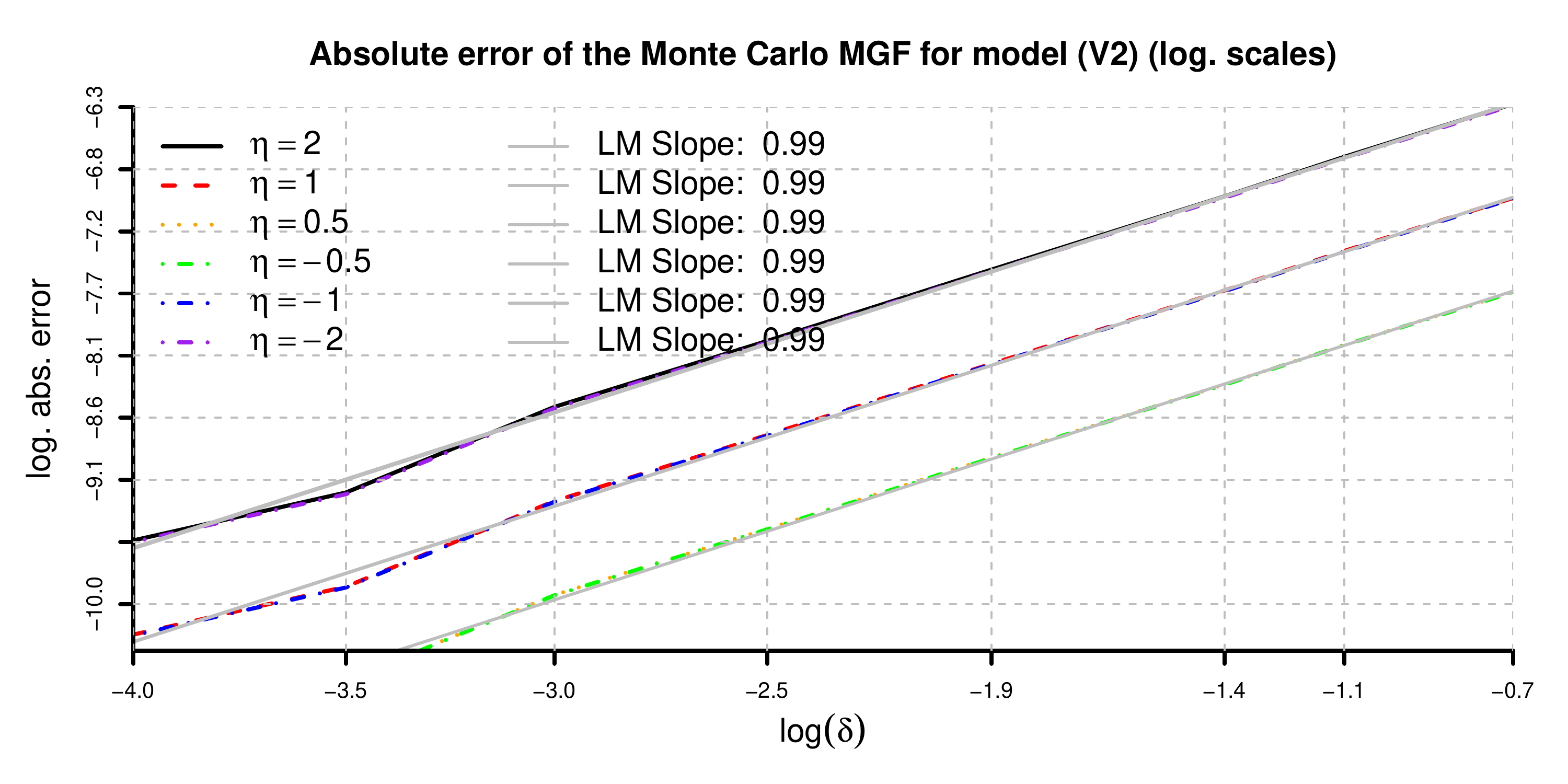} 
\caption{Absolute error (log-log plot) of the Monte Carlo estimate of the moment generating function $E\left[e^{\eta r(1)}\right]$ for model (V\ref{item:num:va2}). This is calculated as the absolute difference between the estimate and \eqref{equ:num_mgf} for different values of $\delta$. We simulate $10^5$ paths for the estimation.}
\label{fig:num_convergence1}
\end{figure}

\begin{figure}
\centering 
\includegraphics[scale=0.4]{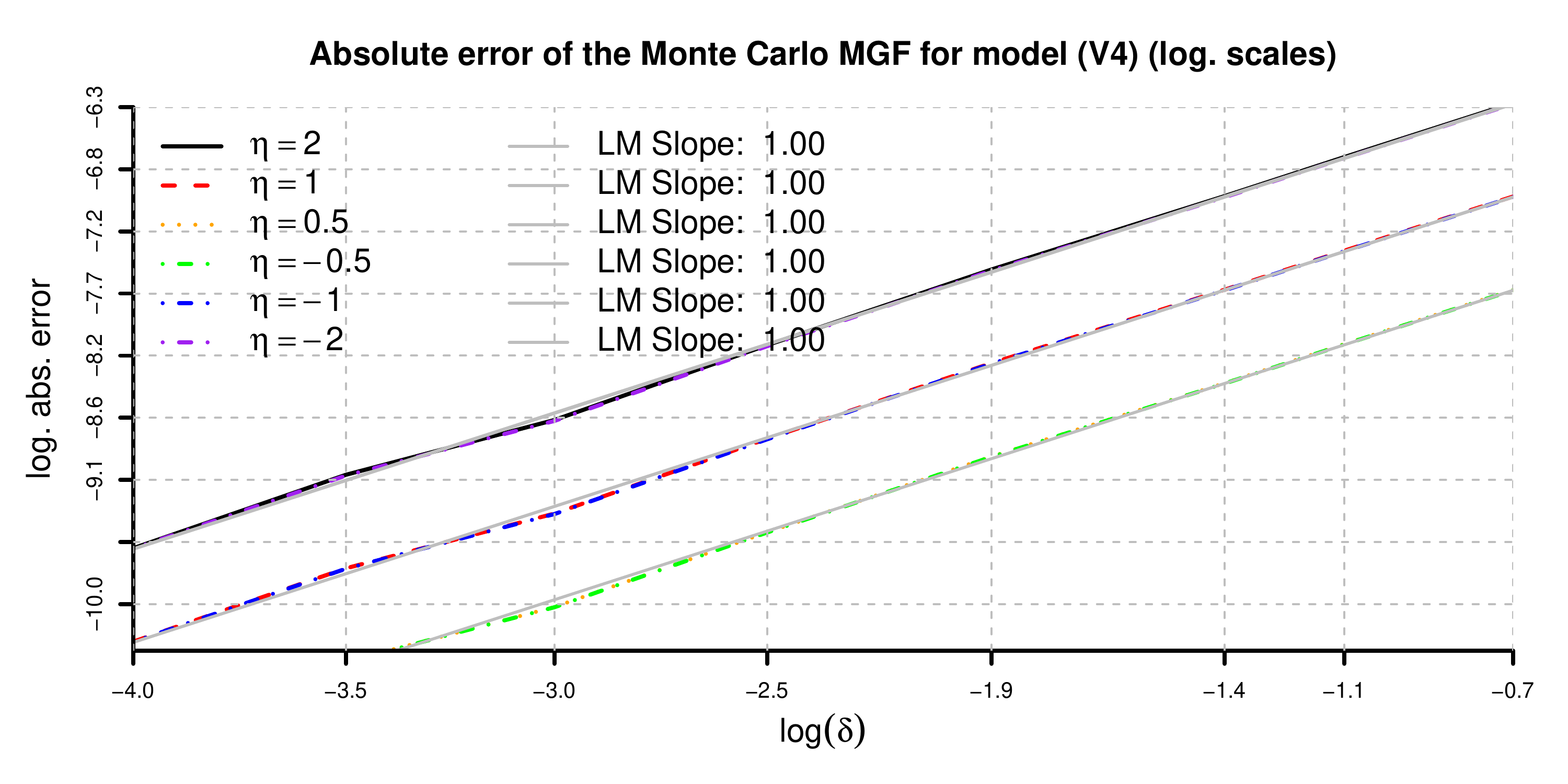} 
\caption{Absolute error (log-log plot) of the Monte Carlo estimate of the moment generating function $E\left[e^{\eta r(1)}\right]$ for model (V\ref{item:num:va4}) defined in \autoref{sec:num:params}. The true values are estimated by the intercept of the linear extrapolation of the Monte Carlo estimates. The errors are calculated as the absolute difference between the intercept and the estimates. $10^5$ paths were used in the simulation.}
\label{fig:num_convergence2}
\end{figure}

\begin{figure}
\centering 
\includegraphics[scale=0.4]{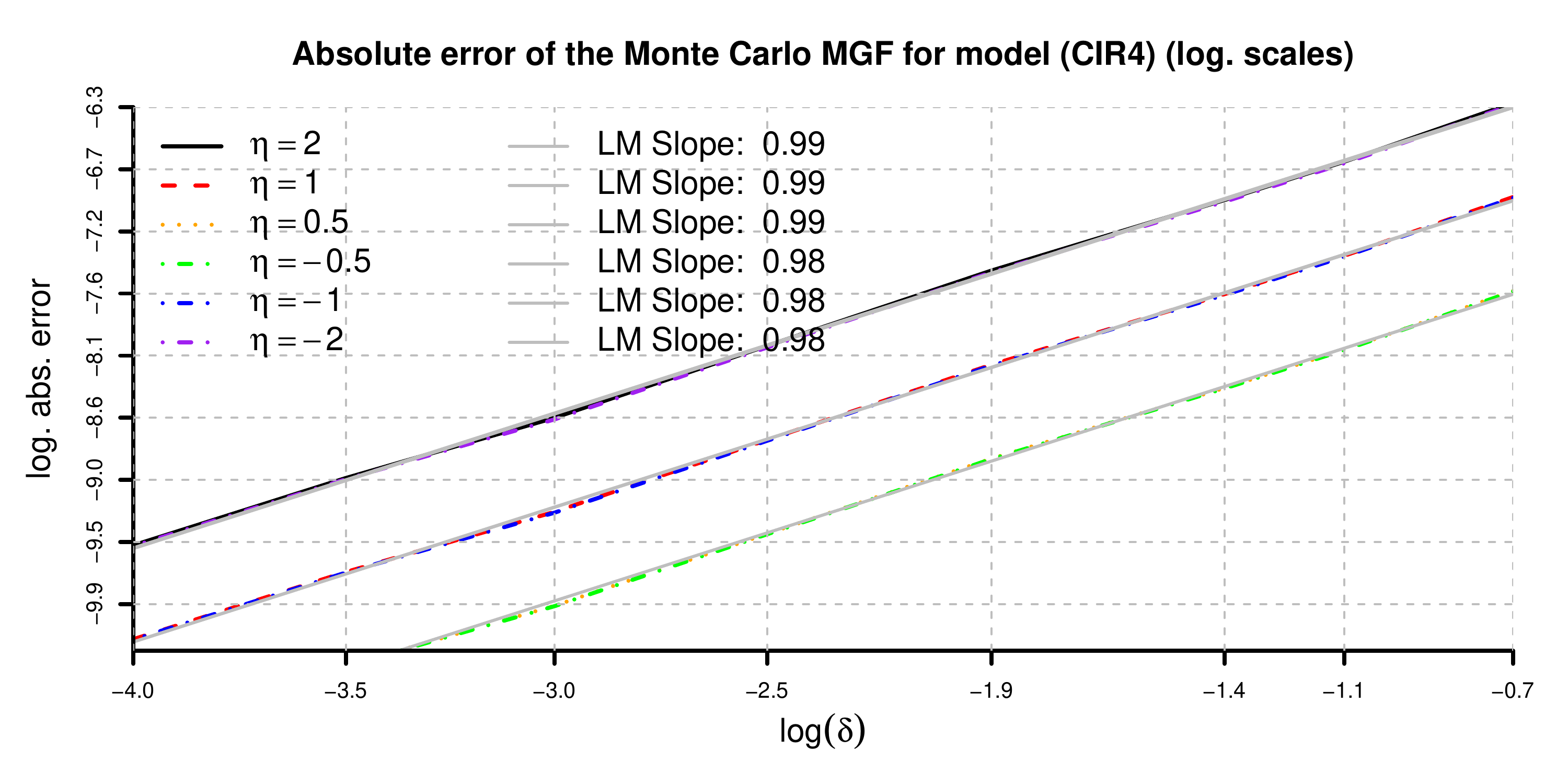} 
\caption{Absolute error (log-log plot) of the Monte Carlo estimate of the moment generating function $E\left[e^{\eta r(1)}\right]$ for model (CIR\ref{item:num:cir4}). The true values are estimated by the intercept of the linear extrapolation of the Monte Carlo estimates. The errors are calculated as the absolute difference between the intercept and the estimates for different values of $\delta$. $10^5$ paths were used in the simulation.}
\label{fig:num_convergence3}
\end{figure}

\subsection{Distributional properties}

Making parameters in the CIR and Vasi\v cek model stochastic in the sense of CRC models has considerable impact on the distribution of short rates and prices. As an example, statistics of the short rate $r(1)$ obtained by simulation are presented in \autoref{tab:num_distribution}. In the models (V\ref{item:num:va1}) and (V\ref{item:num:va2}) with deterministic parameters, the short rate process is Gaussian (see \autoref{sec:va:example}). As expected, the simulations show skewness and excess kurtosis values close to zero. In contrast, leptokurtosis appears in the models (V\ref{item:num:va3}) and (V\ref{item:num:va4}) with stochastic parameters. In the CIR examples, the distribution of $r(1)$ is also affected considerably by the stochasticity of the parameters.

\begin{table}
\centering
\includegraphics[scale=0.5]{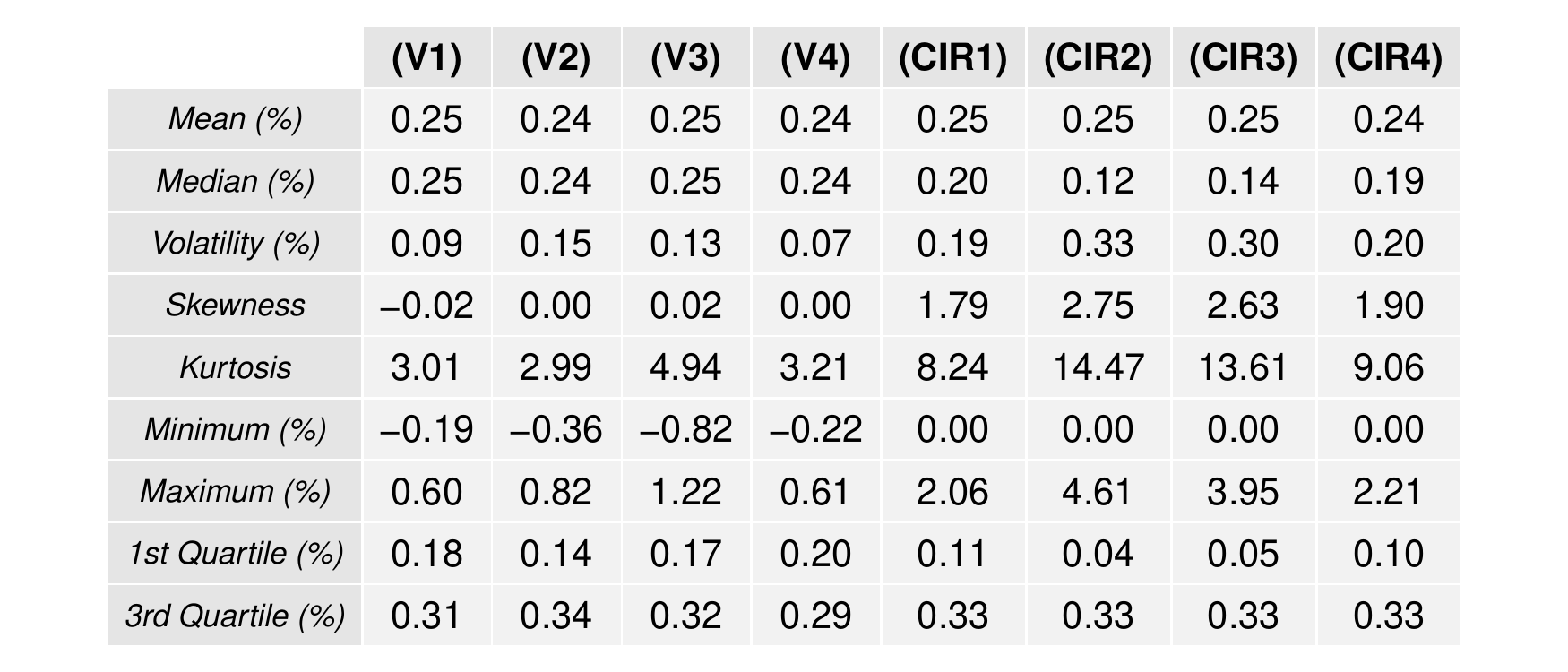}
\caption{Statistics of the short rate $r(1)$ at time 1 in the models defined in \autoref{sec:num:params} obtained by Monte-Carlo simulations. $10^5$ paths and a step size of $\delta=0.02$ were used in the simulation.}
\label{tab:num_distribution}
\end{table}

\subsection{Covariation of yields}

A further example where empirical differences between CRC and non-CRC models become apparent is the covariation matrix of yields. By the Dubins-Schwarz theorem, the covariation process determines the distribution of yields and therefore the prices of bond options. On short time intervals, the covariations are closely related to covariances, which are a key determinant of the prices of call and put options.

In the Hull-White Vasi\v cek and CIR models without the CRC extension, the covariation matrix of yields with different times to maturity has rank 1. This is in stark contrast to the covariations observed in the market. For instance, the $33\times 33$ covariation matrix of market yields with times to maturity ranging from 3 months to 30 years typically has rank between 5 and 9, as shown in \autoref{fig:num_market_covariation}. The ranks produced by the CRC models (V\ref{item:num:va4}) and (CIR\ref{item:num:cir4}) typically lie between 3 and 5. Thus, they are higher than those of the non-CRC models, but not as high as those of the market.

Numerically, the covariation matrix is calculated as in \eqref{equ:va:covariation} and the ranks are defined as the number of singular values differing significantly from zero. A comparison with \Autoref{fig:num_sigma_sigma,fig:num_beta_sigma} shows that higher ranks are also related to higher volatility of the parameters.

\begin{figure}
\centering
\includegraphics[scale=0.4]{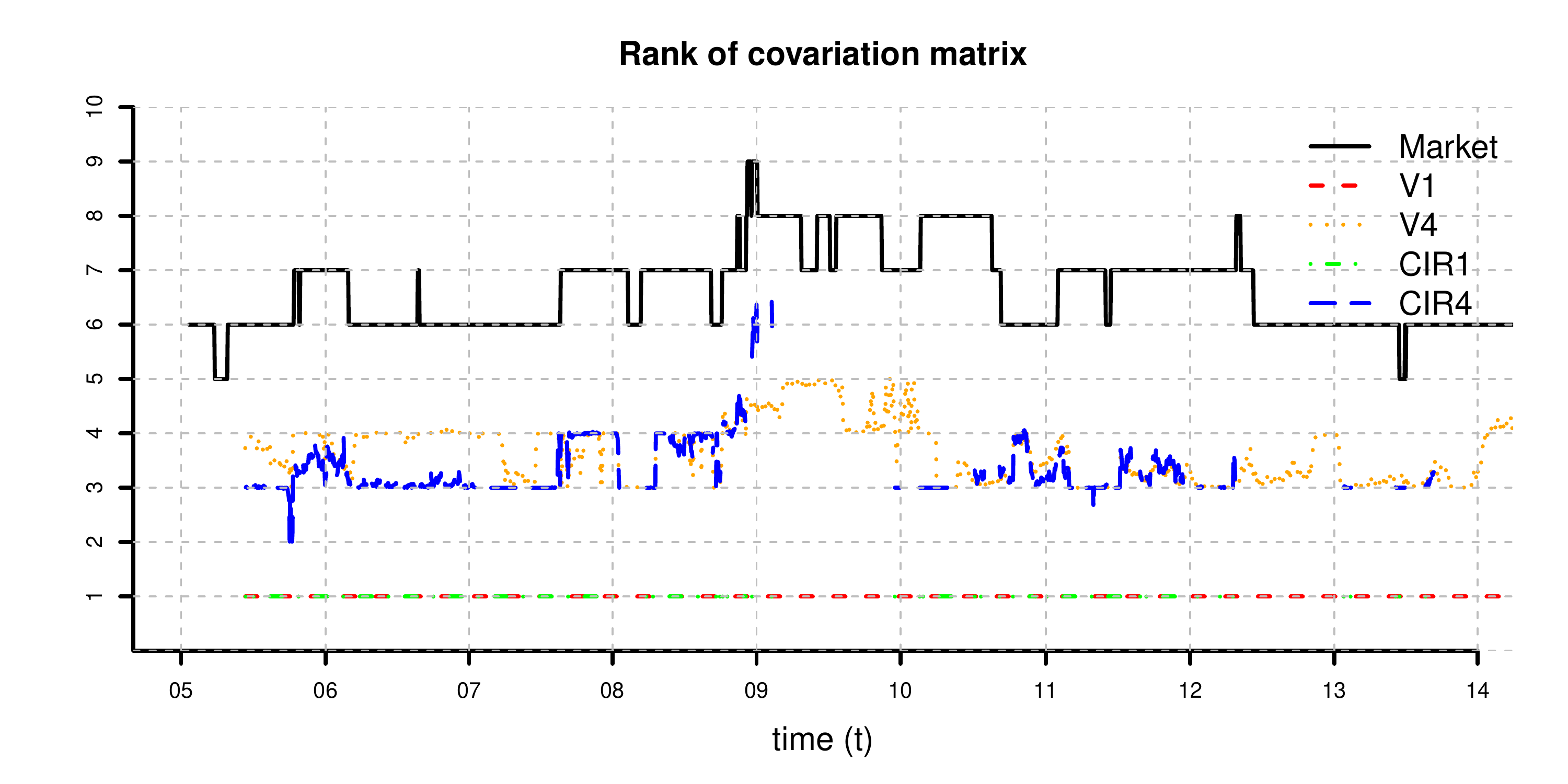}
\caption{Historical rank of the empirical covariation matrix \eqref{equ:va:covariation} based on time windows of $M=100$ market yields with 33 different times to maturity $\tau_i\in\{0.25,0.5,0.75,1,2,3,\ldots,30\}$. For comparison, the plot also features the average ranks obtained in simulations of the Hull-White extended affine models (V\ref{item:num:va1}) and (CIR\ref{item:num:cir1}) as well as their CRC counterparts (V\ref{item:num:va4}) and (CIR\ref{item:num:cir4}). These models were calibrated using time windows of $M=100$ observations. The missing values in (CIR\ref{item:num:cir1}) and (CIR\ref{item:num:cir4}) are due to non-admissible negative levels of mean reversion at these dates, see \autoref{sec:num:negtheta} and \autoref{fig:num_hwe5}. The averages are taken over $10^3$ simulated paths. In the numerical computation of the rank, eigenvalues which are $10^{-6}$ times smaller than the largest eigenvalue are rounded down to zero.}
\label{fig:num_market_covariation}
\end{figure}

\begin{figure}
\centering
\includegraphics[scale=0.4]{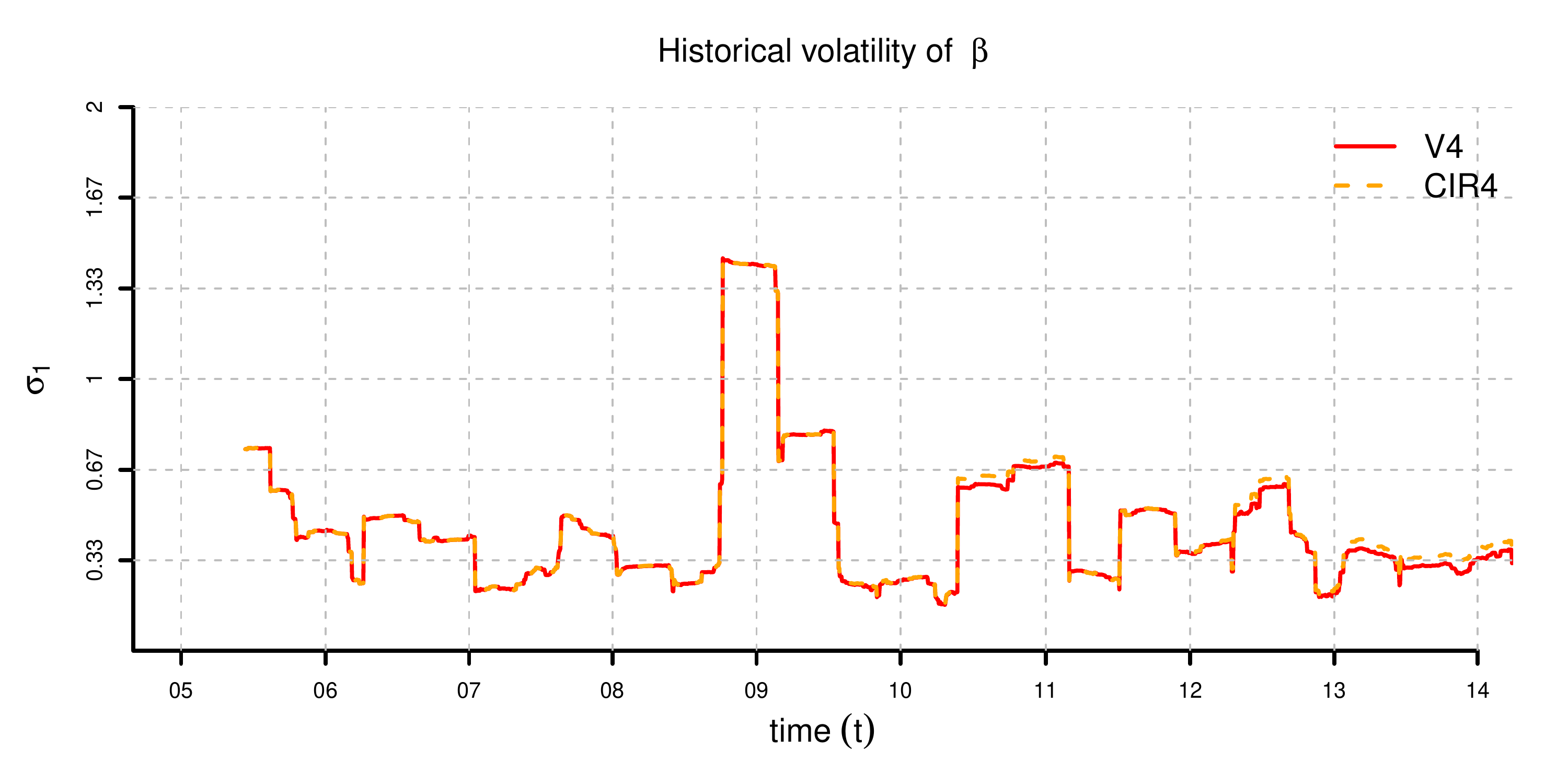}
\caption{Historical values of the parameter $\sigma_1$ in the models (V\ref{item:num:va4}) and (CIR\ref{item:num:va4}) defined in \autoref{sec:num:params} estimated using time windows of $M=100$ observations.}
\label{fig:num_beta_sigma}
\end{figure}
\begin{figure}
\centering
\includegraphics[scale=0.4]{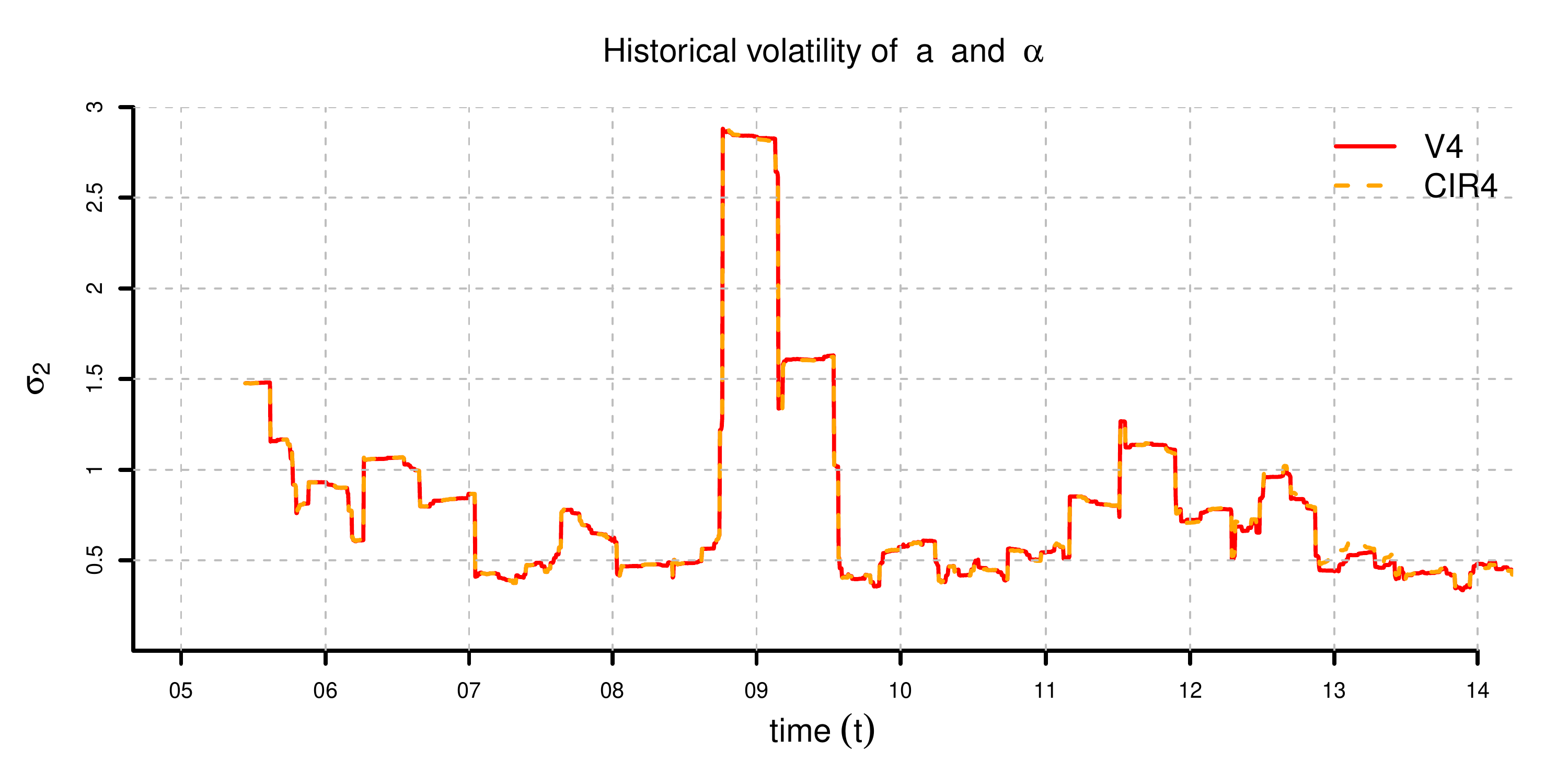}
\caption{Historical values of the parameter $\sigma_2$ in the models (V\ref{item:num:va4}) and (CIR\ref{item:num:va4}) defined in \autoref{sec:num:params} estimated using time windows of $M=100$ observations.}
\label{fig:num_sigma_sigma}
\end{figure}

\appendix

\section{Consistent recalibration of Cox-Ingersoll-Ross models}\label{sec:cir:app}

\subsection{Overview}

We describe CRC models based on CIR short rates, giving a detailed description of the simulation and calibration schemes. For comparison, we briefly digress to the CIR++ model and its CRC version. 

\subsection{Setup and notation}

We use the setup of \Autoref{sec:hw:setup,sec:crc:setup}, setting $\mathbb X=\mathbb R_+$, $\ell=0$, $\lambda=1$.  We do not specify the parameter space $\mathbb Y$, yet, but we assume that for each $(x,y)\in\mathbb X\times\mathbb Y$, the volatility and drift coefficients are given by $A_y(x)=\alpha_y x$ and $B_y(x)=\beta_y x$ for some $\alpha_y \in (0,\infty)$ and $\beta_y \in (-\infty,0)$. For simplicity, we again choose equidistant grids of times $t_n=n\delta$ and times to maturity $\tau_n=n\delta$, for all $n\in\mathbb N_0$, where $\delta$ is a positive constant. 

\subsection{Hull-White extended Cox-Ingersoll-Ross models}\label{sec:cir:hw:app}

For each fixed set of parameters $(y,\theta)\in\mathbb Y\times C(\mathbb R_+)$, the SDE for the short rate process is
\begin{equation}\label{equ:cir:sde:app}
dr(t)=\big(\theta(t)+\beta_y r(t)\big)dt+\sqrt{\alpha_y r(t)}dW(t),
\end{equation}
where $W$ is one-dimensional $(\mathcal F(t))_{t\geq0}$-Brownian motion. 
Thus, $(y,\theta)$ satisfies \autoref{ass:hw:x} if and only if $\theta(t)\geq 0$, for all $t\in\mathbb R_+$. 
The functional characteristics $(F,R)$ from \autoref{sec:hw:riccati} are 
\begin{equation*}
F_y(u)=0, 
\qquad
R_y(u)=\frac{\alpha_y}{2}  u^2+\beta_y u, 
\qquad
\text{for all } u \in \mathbb R.
\end{equation*}
Letting $\gamma_y=\sqrt{\beta_y^2+2\alpha_y}$, the solutions of the corresponding Riccati equations are
\begin{equation}\label{equ:cir:phipsi:app}
\Phi_y(t)=0,
\qquad
\Psi_y(t)=\frac{-2\left(e^{\gamma_y t}-1\right)}{\gamma_y\left(e^{\gamma_y t}+1\right)-\beta_y\left(e^{\gamma_y t}-1\right)},
\qquad\text{for all } t\geq 0.
\end{equation}
Thus, by \autoref{thm:hw:price}, the forward rates in the Hull-White extended CIR model \eqref{equ:cir:sde:app} with fixed parameters $(y,\theta)$ are given by $h(t)=\mathcal H_y(\mathcal S(t)\theta,r(t))$, where 
\begin{equation*}
\mathcal H_y(\theta,x)(\tau)
=-\int_0^\tau \theta(s) \Psi_y'(\tau-s) ds- \Psi_y'(\tau) x,
\qquad
\text{for all } (x,\tau)\in\mathbb R\times\mathbb R_+.
\end{equation*}
In contrast to the Vasi\v cek model, the integral kernel $\Psi_y'$ is more complicated, 
\begin{align*}
\Psi_y'(\tau)&=\frac{-2\gamma_y e^{\gamma_y\tau}}{\gamma_y\left(e^{\gamma_y\tau}+1\right)-\beta_y\left(e^{\gamma_y\tau}-1\right)}+\frac{2\left(e^{\gamma_y\tau}-1\right)\gamma_y e^{\gamma_y\tau}\left(\gamma_y-\beta_y\right)}{\left(\gamma_y\left(e^{\gamma_y\tau}+1\right)-\beta_y\left(e^{\gamma_y\tau}-1\right)\right)^2},
\end{align*}
and there does not seem to be a closed-form expression for $\theta=\mathcal C_y(h,x)$. Instead, it must be calculated numerically as described in \autoref{sec:hw:numerical_volterra}.

The HJM drift and volatility are
\begin{align}\label{equ:cir:hjm_drift_vola:app}
\mu^{\mathrm{HJM}}_y(x)(\tau)&=\Psi_y'(\tau) \Psi_y(\tau) \alpha_y x ,
&
\sigma^{\mathrm{HJM}}_y(x)(\tau)&=-\sqrt{\alpha_y x} \Psi'_y(\tau),
\end{align}
and the HJM equation for forward rates reads as
\begin{align}\label{equ:cir:hjm:app}
dh(t)&=\left(\mathcal Ah(t)+\mu_{y}^{\mathrm{HJM}}\big(h(t)(0)\big)\right)dt+\sigma_{y}^{\mathrm{HJM}}\big(h(t)(0)\big)dW(t).
\end{align}

\subsection{Cox-Ingersoll-Ross CRC models}\label{sec:cir:crc:app}

As the factor process is a function of the forward rate process (i.e., $X(t)=r(t)=h(t,0)$), the corresponding CRC models can be characterised by the process $(h,Y)$ instead of $(h,X,Y)$. Thus, in accordance with \autoref{thm:crc:hjm} and \autoref{def:crc:continuous}, a process $(h,Y)$ with values in $\mathbb H\times \mathbb Y$ may be called a CRC model if $h$ satisfies the SPDE 
\begin{align}\label{equ:cir:crc:app}
dh(t)&=\left(\mathcal Ah(t)+\mu_{Y(t)}^{\mathrm{HJM}}\big(h(t)(0)\big)\right)dt+\sigma_{Y(t)}^{\mathrm{HJM}}\big(h(t)(0)\big)dW(t),
\end{align}
with drift $\mu^{\mathrm{HJM}}_{Y(t)}$ and volatility $\sigma^{\mathrm{HJM}}_{Y(t)}$ as defined in \autoref{equ:cir:hjm_drift_vola:app}. To ensure that the drift and volatility are well-defined, for all $t \in \mathbb R_+$, it must be assumed that $h(t)(0)\geq 0$ holds, for all $t\in\mathbb R_+$. The maximally admissible set $\mathcal I$ in the CIR model is exactly characterised by this condition. 

\subsection{Simulation of Cox-Ingersoll-Ross CRC models}\label{sec:cir:sim:app}

The CRC model is simulated as described in \autoref{alg:crc}. Discretisation in time and time to maturity is done as for the Vasi\v cek model. However, in contrast to the Vasi\v cek model, simulating the short rate process and calibrating Hull-White extensions is done by numerical approximations of order two. 
The resulting algorithm is presented below.

\begin{algorithm}[Simulation]\label{alg:cir:crc:app}
Given an initial curve of forward rates $h(0)$ and the parameter process $Y$, execute iteratively the following steps, for each $n\in\mathbb N_0$:

\begin{enumerate}[(i)]
\item The values of $\theta(t_n)=\mathcal C_{Y(t_n)}(h(t_n),r(t_n))$ at times to maturity $0$ and $\delta$ are calculated by applying \autoref{lem:hw:numerical_volterra} to $g=-h(t_n)-\Psi_{Y(t_n)}'r(t_n)$:
\begin{align*}
\theta(t_n)(0)&=
\mathcal Ah(t_n)(0)-\beta_{Y(t_n)}h(t_i)(0),
\\
\theta(t_n)(\delta)&=
\frac{2}{\delta}\big(h(t_n)(\delta)+\Psi_{Y(t_n)}'(\delta)\, r(t_n)\big)+\Psi_{Y(t_n)}'(\delta)\,\theta(t_n)(0).
\end{align*}

\item  A random draw $r(t_{n+1})=X^{t_n,X(t_n)}_{Y(t_n),\theta(t_n)}(t_{n+1})$ of the CIR process with parameter $Y(t_n)$ and time-dependent drift $\theta(t_n)$ is created using the second-order scheme of \cite{alfonsi2010high}. 

\item $\big(h(t_{n+1}),\mathcal Ah(t_{n+1})\big)$ is calculated from $\big(h(t_n),\mathcal Ah(t_n),r(t_{n+1})\big)$ using \autoref{lem:crc:eff} with integrals approximated by the trapezoid rule:
\begin{align*}
h(t_{n+1})(\tau)&= h(t_n)(\delta+\tau)
+ \Psi_{Y(t_n)}'(\delta+\tau)\,r(t_n)
-\Psi_{Y(t_n)}'(\tau)\,r(t_{n+1})
\\&\qquad
+\frac{\delta}{2} \left(\theta(t_n)(0)\Psi_{Y(t_n)}'(\delta+\tau) 
+ \theta(t_n)(\delta)\Psi_{Y(t_n)}'(\tau) \right),
\\
\mathcal Ah(t_{n+1})(\tau)&= \mathcal Ah(t_n)(\delta+\tau)
+ \Psi_{Y(t_n)}''(\delta+\tau)\,r(t_n)
-\Psi_{Y(t_n)}''(\tau)\,r(t_{n+1})
\\&\qquad
+\frac{\delta}{2} \left(\theta(t_n)(0)\Psi_{Y(t_n)}''(\delta+\tau) 
+ \theta(t_n)(\delta)\Psi_{Y(t_n)}''(\tau) \right).
\end{align*}
Here, $h(t_{n+1})$ must be calculated at all times to maturity $\tau_i$, whereas $\mathcal Ah(t_{n+1})$ is needed only at $\tau_0=0$.
\end{enumerate}
\end{algorithm}

\subsection{Calibration of Cox-Ingersoll-Ross CRC models}\label{sec:cir:ca:app}

We proceed as in the Vasi\v cek case described in \autoref{sec:va:ca}, with \autoref{equ:va:covariation} replaced by
\begin{align}\label{equ:cir:covariation:app}
\frac{[\widehat r(\cdot,\tau_i),\widehat r(\cdot,\tau_j)](t_n)-[\widehat r(\cdot,\tau_i),\widehat r(\cdot,\tau_j)](t_{n-M})}
{\delta \sum_{m=n-M+1}^n \widehat r(t_m,\tau_k)}
\approx
\alpha \frac{\Psi(\tau_i)}{\tau_i}\frac{\Psi(\tau_j)}{\tau_j}.
\end{align}
The function $\Psi$ depends on $\alpha,\beta$ as shown in \autoref{equ:cir:phipsi:app}. 

\subsection{CIR++ models in the CRC framework}\label{sec:cir:cirpp:app}

In the CIR++ model \cite[Section 3.9]{brigo2007}, also known as deterministic shift-extended CIR model, the short rate process is defined by $r(t)=X(t)+\theta(t)$, where $X$ is a CIR process and $\theta$ is a deterministic function of time. Note that this is a different time-inhomogeneity than the one described in \autoref{sec:cir:hw:app}. In particular, the factor process $X$ is time-homogeneous and does not coincide with the short rate. 

Forward rate curves are given by
\[
h(t)=\mathcal S(t)\theta-b_y\Psi_y-\Psi'_yX(t),
\]
where $\Psi_y$ is the same as in the CIR case, see \autoref{equ:cir:phipsi:app}. Given the parameter vector $y$ and the factor $X$, this equation allows to calibrate $\theta$ to a given yield curve without having to invert a Volterra integral operator. The HJM equation of the CIR++ model is
\begin{equation}\label{equ:cir:hjmpp:app}\begin{aligned}
dh(t)&=\left(\mathcal Ah(t)+\mu_{y}^{\mathrm{HJM}}\big(X(t)\big)\right)dt+\sigma_{y}^{\mathrm{HJM}}\big(X(t)\big)dW(t), 
\\
dX(t)&=\left(b_y+\beta_y X(t)\right)dt + \sqrt{\alpha_y X(t)}dW(t),
\end{aligned}\end{equation}
where $\mu_{y}^{\mathrm{HJM}}$ and $\sigma_{y}^{\mathrm{HJM}}$ are the same as in the CIR case, see \autoref{equ:cir:hjm_drift_vola:app}. 

The CRC extension of the CIR++ model is obtained by replacing the constant parameter vector $y$ in \eqref{equ:cir:hjmpp:app} by a stochastic process $(Y(t))_{t\geq 0}$. The resulting equation is easier to handle than its CIR counterpart for two reasons. First, there are no boundary conditions on $h$. Indeed, $\theta$ is allowed to assume negative values and can be calibrated to any forward rate curve. Thus, \autoref{equ:cir:hjmpp:app} is defined on the entire space $\mathbb H\times\mathbb R_+$. Second, the SDE for $X$ does not depend on $h$. Therefore, one can first solve for $X$, and then construct a mild solution $h$ by stochastic convolution \cite[Section 6.1]{daPrato2014se}:
\begin{align*}
h(t) = \mathcal S(t)h(0) + \int_0^t \mathcal S_{t-s} \mu_{Y(s)}^{\mathrm{HJM}}\big(X(s)\big)ds
+\int_0^t \mathcal S_{t-s} \sigma_{Y(s)}^{\mathrm{HJM}}\big(X(s)\big) dW(s).
\end{align*}
The SDE for $X$ is finite-dimensional. Therefore, existence and uniqueness of $X$ can be shown by standard methods. For example, assuming that $Y$ is independent of $W$, one can condition on $Y$ and use results on time-inhomogeneous affine processes \cite{filipovic2005time} to construct $X$. 

Simulation of the CRC model is analogue to \Autoref{alg:crc}. The recalibration step is easier because no Volterra equation is involved. 

A disadvantage of the model is the presence of the hidden factor $X$. In contrast to the CIR version, $X$ is not a function of the forward rate curve and cannot be directly observed. This is a challenge for calibration. We suggest an analogue approach to \Autoref{sec:crc:calibration}. First, $\beta_{Y(t)}$, $\sigma_{Y(t)}$, and $X(t)$ can be identified from the instantaneous covariation 
\begin{equation*}
d\left[r(\cdot,\tau_i),r(\cdot,\tau_j)\right](t)=\alpha_{Y(t)} \frac{\Psi_{Y(t)}(\tau_i)}{\tau_i} \frac{\Psi_{Y(t)}(\tau_j)}{\tau_j} X(t)dt,
\end{equation*}
of yields with times to maturity $\tau_i, \tau_j$. Subsequently, $b_{Y(t)}$ can be calibrated by least squares to the prevailing yield curve. Note that in this approach, $X(t)$ is identified from the yield curve dynamics instead of extracted from the prevailing yield curve as in the Vasi\v cek and CIR cases. For this reason the calibration is expected to be numerically more difficult.

\printbibliography

\end{document}